%% file: main.tex
\documentclass[runningheads]{llncs}

\usepackage{url}
\usepackage{fullpage}
\usepackage{mathtools}
\usepackage{bm}
\usepackage{blindtext}

\usepackage{subcaption}

\usepackage{parskip}
\usepackage{pgfplots, pgfplotstable}
\usepackage{float}
\usepackage{boxedminipage}
\usepackage{enumitem}
\usepackage[ruled,linesnumbered]{algorithm2e}
\usepackage{graphicx}
\usepackage{framed}
\usepackage{esvect}
\usepackage{tikz}
\usepackage{latexsym}

\usepackage{blkarray}

\usepackage{prettyref}
\usepackage{multirow}
\usepackage{tablefootnote}

\usepackage{color, colortbl}
\usepackage{amsfonts,amsmath,amssymb,graphicx,setspace}

\usepackage{todonotes}

\newcommand{\lnum}{\refstepcounter{equation}{\st\theequation:}\hspace{1.5mm}}

\usepackage{longtable}
\usepackage{xspace}

\usepackage{ulem}

\input{macro}
\let\emph\relax

\usepackage[hidelinks]{hyperref}

\usepackage[most]{tcolorbox}
\newtcolorbox{mybox}[2][]{%
  attach boxed title to top center
               = {yshift=-10pt},
  colframe     =black,
  colbacktitle = black,
  title        = #2,#1,
  enhanced,
}

\begin{document}

\title{Verifiable Homomorphic Linear Combinations in\\ Multi-Instance Time-Lock Puzzles}



\author{%
Aydin Abadi\thanks{aydin.abadi@ncl.ac.uk}
\institute{Newcastle University} 
 }

\maketitle  

\begin{abstract}
Time-Lock Puzzles (TLPs) have been developed to securely transmit sensitive information into the future without relying on a trusted third party. 
Multi-instance TLP is a \textit{scalable} variant of TLP  that enables a server to {efficiently} find solutions to different puzzles provided by a client at once. 
%
%
Nevertheless, existing multi-instance TLPs lack support for (verifiable)  \textit{homomorphic computation}. 
To address this limitation, we introduce the ``Multi-Instance partially Homomorphic TLP'' (\mhtlp), a multi-instance TLP supporting efficient verifiable homomorphic linear combinations of puzzles belonging to a client. It ensures anyone can verify the correctness of computations and solutions. 
Building on  \mhtlp, we further propose the ``Multi-instance Multi-client verifiable partially Homomorphic TLP'' (\mmhtlp). It not only supports all the features of \mhtlp 
but also allows for verifiable homomorphic linear combinations of puzzles from different clients. 
Our schemes refrain from using asymmetric-key cryptography for verification and, unlike most homomorphic TLPs, do not require a trusted third party. 
A comprehensive cost analysis demonstrates that our schemes scale linearly with the number of clients and puzzles.

\end{abstract}



\input{introduction}

\input{Puzzle-literature-review}

\input{preliminaries}

\input{definition-of-HTLP}

\input{V-HLC}

\input{cost-analysis}

\input{conclusion}

\bibliographystyle{splncs04}

\bibliography{ref}

\appendix
\input{OLE-plus}

\input{V-TLP}




\input{RSA-based-TLP}

\input{sequential-squaring}


\end{document}

%% file: macro.tex
\newlist{steps}{enumerate}{1}
\setlist[steps, 1]{label = Case \arabic*}

\newcommand{\singclient}
{\ensuremath{\text{SingleClient}}\xspace}

\newcommand{\mulclient}{\ensuremath{\text{MultiClient}}\xspace}

\newcommand{\mhtlp}{\ensuremath{\text{MH-TLP}}\xspace}
\newcommand{\mmhtlp}{\ensuremath{\text{MMH-TLP}}\xspace}

\newcommand{\subprt}{\ensuremath{_{u}}\xspace}

\newcommand{\adv}{\ensuremath{\mathcal{A}}\xspace}

\newcommand{\chal}{\ensuremath{\mathcal{E}}\xspace}

\newcommand{\corupt}{\ensuremath{\mathcal{W}}\xspace}

\newcommand{\et}{\textit{et al}\xspace}
\newcommand{\cor}{\ensuremath{\bar{t}}\xspace}

\newcommand{\id}{\ensuremath{id}\xspace}

\newcommand{\rt}{\ensuremath{root}\xspace}
\newcommand{\vhtlp}{\ensuremath{\mathcal{TLP_{\st MH}}}\xspace}

\newcommand{\tf}{\ensuremath{\text{Tempora\text{-}Fusion}}\xspace}

\newcommand{\tl}{\ensuremath{\ddot{t}}\xspace}

\newcommand{\kp}{\ensuremath{K}\xspace}
\newcommand{\cmd}{\ensuremath {cmd}\xspace}

\newcommand{\mxsqr}{\ensuremath{max_{\st ss}}\xspace}

\newcommand{\prtt}{\ensuremath{C}\xspace}
\newcommand{\srv}{\ensuremath {S}\xspace}
\newcommand{\cln}{\ensuremath {C}\xspace}

\newcommand{\prm}{\ensuremath{{p}}\xspace}

\newcommand{\ssetup}{\ensuremath{\mathsf{S.Setup}}\xspace}
\newcommand{\csetup}{\ensuremath{\mathsf{C.Setup}}\xspace}
\newcommand{\pgen}{\ensuremath{\mathsf{\mathsf{GenPuzzle}}}\xspace}
\newcommand{\eval}{\ensuremath{\mathsf{\mathsf{Evaluate}}}\xspace}
\newcommand{\solv}{\ensuremath{\mathsf{\mathsf{Solve}}}\xspace}
\newcommand{\ver}{\ensuremath{\mathsf{\mathsf{Verify}}}\xspace}

\newcommand{\opgen}{\ensuremath{\mathsf{\mathsf{GenPuzzle^{\st Oracle}}}}\xspace}
\newcommand{\oeval}{\ensuremath{\mathsf{\mathsf{Evaluate^{\st Oracle}}}}\xspace}
\newcommand{\orev}{\ensuremath{\mathsf{\mathsf{Reveal^{\st Oracle}}}}\xspace}

\newcommand{\scp}{\ensuremath{\text{clientPzl}}\xspace}

\newcommand{\ep}{\ensuremath{\text{evalPzl}}\xspace}

\newcommand{\cl}{\ensuremath{ C}\xspace}
\newcommand{\se}{\ensuremath{ S}\xspace}

\newcommand{\prt}{\ensuremath{u }\xspace}

\newcommand{\idx}{\ensuremath{\mathcal{I} }\xspace}

\newcommand{\ole}{\ensuremath{\mathtt{OLE}}\xspace}

\newcommand{\prf}{\ensuremath{\mathtt{PRF} }\xspace}

\newcommand{\g}{\ensuremath{\mathtt{G} }\xspace}

\newcommand{\pnp}{\ensuremath{ \phi(N_{\st \prt}) }\xspace}

\newcommand{\set}{\ensuremath{\{\cln_{\st 1},\ldots, \cln_{\st n}\}}\xspace}

\newcommand{\comcom}{\ensuremath{\mathtt{Com}}\xspace}
\newcommand{\comver}{\ensuremath{\mathtt{Ver}}\xspace}

\newcommand{\st}{\scriptscriptstyle}

\newcommand{\func}{\ensuremath{\mathcal{F}^{\st \text{PLC}}}\xspace}

\newtheorem{assumption}{Assumption}

%% file: introduction.tex

\section{Introduction}\label{sec::intro}

Time-Lock Puzzles (TLPs) are interesting cryptographic primitives that enable the transmission of information to the future. They enable a party to encrypt a message in a way that no one else can decrypt it until a certain time has elapsed.\footnote{Certain  TLPs rely on third-party assistance to manage the timed release of a secret. However, this category of protocols is not the focus of this paper.} TLPs have a wide range of applications, including scheduled payments in cryptocurrencies \cite{ThyagarajanBMDK20}, timed-commitments \cite{KatzLX20}, zero-knowledge proofs \cite{dwork2000zaps}, e-voting \cite{ChenD12}, timed secret sharing \cite{kavousi2023timed}, sealed-bid auctions \cite{Rivest:1996:TPT:888615},  verifiable delay functions \cite{BonehBBF18}, and secure aggregation in federated learning \cite{tempora-fusion}.

Since Rivest \et. \cite{Rivest:1996:TPT:888615} introduced TLPs, they have evolved, leading to the development of two vital variants: \textit{multi-instance} TLPs and \textit{homomorphic} TLPs. Our work advances both of these variants. 
Multi-instance TLPs were introduced in \cite{Abadi-C-TLP}.  The multi-instance TLP setting entails a single client generating $n$ puzzles and simultaneously transmitting them to a server. This concept serves as a natural extension of the initial single-puzzle paradigm established in \cite{Rivest:1996:TPT:888615}.  In the multi-instance TLP, each puzzle's solution is found by the server at a different time.  
 Multi-instance TLPs allow the server to deal with each puzzle \textit{sequentially} rather than simultaneously handling them.  
Such an approach leads to notable reductions in computational overhead for the server. Specifically,  for a fixed time parameter $\Delta$ and $z$ puzzles, it saves $\frac{1}{z}\cdot {\sum\limits^{\st z}_{\st j=1}j}$   times modular squaring. For instance, when $z=100$, this approach yields $50$ times reduction in modular squaring.

Multi-instance TLPs have applications across various domains. For example, \textit{journalists} or \textit{whistleblowers} in hostile environments can use multi-instance TLPs to schedule the gradual release of sensitive information at the most impactful moments or after ensuring their safety. This method eliminates the need to be online at the time of release and removes the need to trust a third party with sensitive data. In \textit{online education}, multi-instance TLPs can benefit students with unreliable internet connections by allowing them to download multiple exam content in advance, ensuring access when their connection is stable. In this scenario, the exam materials become accessible only at the designated start time. Multi-instance TLPs can also be utilized for the continuous verification of cloud service availabilities \cite{Abadi-C-TLP}.

%
%

In a separate line of research, Malavolta \et. \cite{MalavoltaT19} introduced (fully) homomorphic TLP notion, enabling the execution of arbitrary functions over puzzles before their resolution is found. In general, fully homomorphic TLPs consider scenarios where there exist $n$ clients, each generating and transmitting a puzzle encoding its solution to a server. Upon receiving the puzzles, the server executes a homomorphic function across these puzzles, generating a unified puzzle. The solution to this puzzle represents the output of the function evaluated across all individual solutions. 
To enhance efficiency, partially homomorphic TLPs have also been introduced, including variants that enable homomorphic linear combinations or the multiplication of puzzles  \cite{MalavoltaT19}. More recently, Abadi \cite{tempora-fusion} proposed a partially homomorphic TLP that also enables efficient verification of the correctness of homomorphic operations performed on the puzzles. Homomorphic TLPs have found applications in numerous areas, including atomic swaps \cite{ThyagarajanMM22}, payment channels \cite{ThyagarajanMSS20}, as well as verifiable e-voting, and secure aggregation in federated learning \cite{tempora-fusion}.




\subsection{Limitations of State-of-the-Art TLPs}

Current TLPs exhibit notable limitations that constrain their practical use. Existing {multi-instance} TLPs' functionalities are confined to a few basic operations: (i) solution revelation,  where the server can access the solution to each puzzle only after a designated time period has elapsed, and (ii) solution verification, allowing a verifier to confirm the correctness of the solution obtained by the server. A key drawback of these multi-instance TLPs is their \textit{lack of support for (verifiable) homomorphic operations}.


Conversely, the state-of-the-art {homomorphic} TLPs  \textit{lack the essential feature of multi-instance TLPs}. Specifically, when a client employs a homomorphic TLP to generate multiple puzzle instances intended to be solved at different times and submits all these puzzles to a server simultaneously, the server must handle each puzzle instance separately and in parallel. This naive approach is highly resource-intensive and inefficient.


\subsection{Our Contributions}

\subsubsection{Multi-Instance Partially Homomorphic TLP.} To address the aforementioned limitations of multi-instance TLPs, we introduce Multi-instance verifiable partially Homomorphic TLP  (\mhtlp), the first multi-instance TLP that supports \textit{efficient verifiable homomorphic linear combinations} on puzzles. It enables a client to generate many puzzles and transmit them to the server at once. In this setting, the server does not need to simultaneously handle them; instead, it can solve them sequentially. 

\mhtlp enables the client to come back online at a later point to grant computation on its puzzles. The server will learn the linear combination of puzzles' solutions after a certain time. \mhtlp also supports public verification for (1) a single puzzle's solution, and (2) the computation's result.

Inspired by previous multi-instance TLPs, we also rely on the idea of chaining puzzles, such that when the server solves one puzzle it will obtain enough information to work on the next puzzle. 
However, we introduce a new technique for chaining puzzles that also facilitates homomorphic linear combinations. This method enables the client to derive the base for the next puzzle from the current puzzle's master key, without altering the structure of the underlying solution. 

We formally define our scheme and, for the first time, present the formal definition of multi-instance TLP (in Definition \ref{def:efficiency-vh-tlp}). Although the multi-instance TLP idea has been used previously, its core property, efficiently supporting multiple puzzle instances, has not been formally defined until now.

\subsubsection{Multi-Instance Multi-Client Partially Homomorphic TLP.} 

To address the limitation of the existing (partially) homomorphic TLPs that do not support multi-instance, we upgrade \mhtlp to a new variant of TLP called Multi-instance Multi-client verifiable partially Homomorphic TLP (\mmhtlp). This new variant offers the features of both partially homomorphic TLP and multi-instance TLP. Specifically, it supports verifiable partially homomorphic operations on the puzzles belonging to single or multiple clients while maintaining the multi-instance feature. \mmhtlp allows each client to independently generate different puzzles and send them at once to the server. In this case, the server does not have to deal with each client's puzzles simultaneously (offering the multi-instance feature). It enables single or multiple clients to ask the server to perform homomorphic linear combinations of their puzzles (each belonging to a different client in the multi-client setting). This scheme allows anyone to verify whether the server has performed the computation correctly and provided a correct solution. Thus, \mmhtlp bridges the gap between (a) (partially) homomorphic TLPs that support multiple clients (but not multi-instance) and (b) multi-instance (homomorphic) TLPs that support only a single client.

\input{cost-table}

\input{feature-table}

We compare the computational and communication complexities, as well as the features, of our protocols with those of previous work. Tables \ref{complexity-of-our-schemes} and \ref{feature-of-our-schemes} provide a summary of this comparison.  The complexities of our protocols are linear with the number of participating clients $n$ and the number of puzzle $z$  that each client possesses. Our schemes support efficient verification of solutions without relying on asymmetric key-based primitives.  Moreover, \mmhtlp provides a distinctive set of features that are not simultaneously available in any existing TLP.



\subsection{Applications}

\subsubsection{Scalable, Private, and Compliant Scheduled Payments in Online Banking.}

Outsider and insider attacks pose an immediate risk to various organizations and their clients, particularly financial institutions and their customers. There have been numerous incidents where customers' names, addresses, credit scores, credit limits, and balances have been stolen and in some cases revealed to the public, e.g., see \cite{leigh2015hsbc,JPMorgan,Capital-One-credit}.  Investment strategies devised by individuals or companies and managed through financial institutions are especially vulnerable, as they contain critical and valuable information that attackers could exploit.

Our schemes enable individuals and businesses to schedule payments and investments privately through their banks, without disclosing the transaction amounts prior to the scheduled transfer time. Leveraging the multi-instance feature, these schemes efficiently manage a high volume of scheduled transactions without placing significant strain on server resources. By utilizing a homomorphic linear combination, our schemes allow banks to precalculate the combined transfer amounts, such as the average or total. This functionality not only ensures compliance with regulatory standards but also enhances risk management. By verifying transfer amounts ahead of time, banks can confirm that all transactions meet regulatory requirements. Additionally, the ability to precalculate transfer amounts helps banks identify and address potential risks before they materialize, preventing issues like overdrafts or breaches of internal limits.

\subsubsection{Asynchronous Multi-Model Training in Secure Aggregation as a Service (SAaaS).}

Federated Learning (FL) is a machine learning framework that allows multiple parties to collaboratively build models without exposing their sensitive data to one another  \cite{McMahanMRA16}. Unlike traditional centralized methods, where data is collected and processed on a central server, FL enables model training on individual devices or clients, each holding private data. This approach preserves data privacy by ensuring that raw data never leaves the clients. Instead, only model updates are transmitted to a central server.

To securely compute the sum of model updates from clients, Bonawitz \et. \cite{BonawitzIKMMPRS17} developed a secure aggregation mechanism. This mechanism uses a trusted setup and a public-key-based verification system to detect any misbehavior by the server. In response to this, the author of verifiable partially homomorphic TLP in \cite{tempora-fusion} proposed an alternative that does not require a trusted setup and supports efficient verification. Our solution, \mmhtlp, can also substitute the TLP in \cite{tempora-fusion} (and the secure aggregation in \cite{BonawitzIKMMPRS17}) offering additional features that are particularly well-suited for more generic, multi-model settings, as explained below.

As the adoption of FL continues to grow, it is anticipated that cloud-based servers offering Secure Aggregation as a Service (SAaaS) will become common. In such scenarios, clients may want to train multiple models over time, in collaboration with the server and various sets of clients. Our \mmhtlp can replace the aforementioned secure aggregation mechanisms while enabling asynchronous multi-model training. Specifically, in this setting, each client can submit to the server its (initial) parameters for different models in a single step, well before the start of each model training round. 

\mmhtlp supports (a) asynchronous client participation, allowing clients to create their puzzles and join the computation of linear combinations at different times without waiting for others, and (b) a dynamic client base, enabling new clients to independently join the system, prepare, and submit their puzzles. Consequently, in the context of SAaaS, \mmhtlp facilitates asynchronous and dynamic parameter submissions, ensuring that the server can efficiently and securely aggregate parameters for each model at the appropriate time. 

%

Including the notion of time-lock can also enhance the privacy of FL. By requiring the server to learn the aggregated result after a certain time, the system can allow a sufficient number of clients to submit their model parameters to the server. This is important when there is a possibility of collusion between the clients and the server. Given the result of the secure aggregation of model parameters: $agr=\sum\limits^{\st n}_{\st i=1}a_{\st i}$, corrupt parties that collude with each other can always deduce the linear combination of honest parties' inputs from $agr$. Thus, by defining a time window and assuming enough honest clients submit their updates in this period, we can reduce the chance of colluding adversaries to link input to an honest party.


%

\subsubsection{Sequential Verifiable E-Voting and Sealed-Bid Auction Systems.}

E-voting and sealed-bid auction systems are critical applications where maintaining data integrity, confidentiality, and system scalability is essential. Researchers have suggested utilizing homomorphic TLPs in these systems, to facilitate secure computations while preserving the privacy of each individual vote or bid \cite{MalavoltaT19}.

 \mmhtlp can be applied more broadly to e-voting and sealed-bid auction systems, where clients participate in multiple instances of voting or bidding managed sequentially by a server. This method minimizes the server's workload by avoiding the need to simultaneously handle each voter’s or bidder’s puzzle created for a different instance. Additionally, it allows for public verification of the computations, therefore ensuring the result's correctness in the process, in the case where a malicious server may attempt to alter the result.


%% file: cost-table.tex

\vspace{-5mm}
\begin{table*}[!htb]
\begin{center}
\caption{Asymptotic costs of existing (partially) homomorphic TLPs.  In the figure, $n$ is the number of clients, \tl is the number of leaders, $ \bar\Delta_{\st j, \prt}$ is the period between the discovery of two consecutive solutions of client $\prtt_{\st \prt}$, $\mxsqr$ is the maximum number of squaring that the strongest server \srv can perform per second,  $z$ is the total number of puzzles given by a client to \srv, $\Delta_{\st j, \prt}$ is the period when the puzzle of $\prtt_{\st \prt}$ must remain concealed after the related puzzle is given to \srv, and $Y=\mxsqr\cdot \Delta$ is the period between granting the computation and when a linear combination of solutions is learned by \srv.}
\label{complexity-of-our-schemes}

\renewcommand{\arraystretch}{.96}
\scalebox{1}{
\begin{tabular}{|cc|c|c|c|c|c|c|c|c|c|c|} 
   \hline
   
   &&&&\\
   

 \multirow{-2}{*}{\scriptsize Schemes}&&\multirow{-2}{*}{\scriptsize Parties}&
\multirow{-2}{*}{\scriptsize Computation Cost}&\multirow{-2}{*}{\scriptsize Communication Cost}\\

\hline

&&{\rotatebox[origin=c]{0}{\scriptsize }}\scriptsize Client&\cellcolor{gray!20} \scriptsize $O(z)$&\cellcolor{gray!20} \scriptsize $O(z)$ \\

        \cline{3-5}
   
   \multirow{-2}{*}{\rotatebox[origin=c]{20}{\scriptsize }}&&{\rotatebox[origin=c]{0}{\scriptsize  }} \scriptsize Verifier&\cellcolor{gray!20}\scriptsize $O(z)$&\cellcolor{gray!20}\scriptsize $-$\\
  
   \cline{3-5}

  \multirow{-3}{*}{\rotatebox[origin=c]{0}{\scriptsize \text{{\mhtlp}}}}&&{\rotatebox[origin=c]{0}{\scriptsize  }}\scriptsize Server&\cellcolor{gray!20}\scriptsize $O(z+Y+\mxsqr\cdot\sum\limits^{\st z}_{\st j=1}\bar{\Delta}_{\st j, \prt})$&\cellcolor{gray!20}\scriptsize $O(z)$\\      
            
\hline


&&{\rotatebox[origin=c]{0}{\scriptsize }}\scriptsize Client&\cellcolor{gray!50} \scriptsize $O(\tl\cdot (z+n))$&\cellcolor{gray!50} \scriptsize $O((\tl+n)\cdot z)$ \\

        \cline{3-5}
   
   \multirow{-2}{*}{\rotatebox[origin=c]{0}{\scriptsize }}&&{\rotatebox[origin=c]{0}{\scriptsize  }} \scriptsize Verifier&\cellcolor{gray!50}\scriptsize $O(\tl^{\st 2}+\tl+z)$&\cellcolor{gray!50}\scriptsize $-$\\
  
   \cline{3-5}

  \multirow{-3}{*}{\rotatebox[origin=c]{0}{\scriptsize \text{{\mmhtlp}}}}&&{\rotatebox[origin=c]{0}{\scriptsize  }}\scriptsize Server&\cellcolor{gray!50}\scriptsize $O(\tl\cdot (z+n+Y+\tl)+\mxsqr\cdot \sum\limits_{\st j=1}^{\st z}\bar{\Delta}_{\st j, \prt})$&\cellcolor{gray!50}\scriptsize $O(\tl\cdot n+z)$\\     

\hline
 
&& {\rotatebox[origin=c]{0}{\scriptsize }}\scriptsize Client&\cellcolor{gray!20} \scriptsize $O(\tl\cdot n)$&\cellcolor{gray!20} \scriptsize $O(\tl\cdot n)$ \\

      \cline{3-5}
   
 &&{\rotatebox[origin=c]{0}{\scriptsize  }} \scriptsize Verifier&\cellcolor{gray!20}\scriptsize $O(\tl^{\st 2}+\tl +z)$&\cellcolor{gray!20}\scriptsize $-$\\
  
  \cline{3-5}

 \multirow{-3}{*}{\rotatebox[origin=c]{0}{\scriptsize  \cite{tempora-fusion}}}&&{\rotatebox[origin=c]{0}{\scriptsize  }}\scriptsize Server&\cellcolor{gray!20}\scriptsize $O(\tl\cdot (n+Y+\tl)+\mxsqr\cdot \sum\limits_{\st j=1}^{\st z}\Delta_{\st j, \prt})$&\cellcolor{gray!20}\scriptsize $O(\tl\cdot n)$\\

\hline
 
&& {\rotatebox[origin=c]{0}{\scriptsize }}\scriptsize Client&\cellcolor{gray!50} \scriptsize $O(1)$&\cellcolor{gray!50} \scriptsize $O(1)$ \\

      \cline{3-5}
   
 &&{\rotatebox[origin=c]{0}{\scriptsize  }} \scriptsize Verifier&\cellcolor{gray!50}\scriptsize $-$&\cellcolor{gray!50}\scriptsize $-$\\
  
  \cline{3-5}

 \multirow{-3}{*}{\rotatebox[origin=c]{0}{\scriptsize  \cite{MalavoltaT19}}}&&{\rotatebox[origin=c]{0}{\scriptsize  }}\scriptsize Server&\cellcolor{gray!50}\scriptsize $O(n+\mxsqr\cdot \sum\limits_{\st j=1}^{\st z}\Delta_{\st j, \prt})$&\cellcolor{gray!50}\scriptsize $O(1)$\\

\hline
 
&& {\rotatebox[origin=c]{0}{\scriptsize }}\scriptsize Client&\cellcolor{gray!20} \scriptsize $O(1)$&\cellcolor{gray!20} \scriptsize $O(1)$ \\

      \cline{3-5}
   
 &&{\rotatebox[origin=c]{0}{\scriptsize  }} \scriptsize Verifier&\cellcolor{gray!20}\scriptsize $-$&\cellcolor{gray!20}\scriptsize $-$\\
  
  \cline{3-5}

 \multirow{-3}{*}{\rotatebox[origin=c]{0}{\scriptsize  \cite{liu2022towards}}}&&{\rotatebox[origin=c]{0}{\scriptsize  }}\scriptsize Server&\cellcolor{gray!20}\scriptsize $O(n+\mxsqr\cdot(\frac{ \Delta}{\log(\mxsqr\cdot \Delta)}+ \sum\limits_{\st j=1}^{\st z}\Delta_{\st j, \prt}))$&\cellcolor{gray!20}\scriptsize $O(1)$\\

\hline
 
&& {\rotatebox[origin=c]{0}{\scriptsize }}\scriptsize Client&\cellcolor{gray!50} \scriptsize $O(1)$&\cellcolor{gray!50} \scriptsize $O(1)$\\

      \cline{3-5}
   
 &&{\rotatebox[origin=c]{0}{\scriptsize  }} \scriptsize Verifier&\cellcolor{gray!50}\scriptsize $-$&\cellcolor{gray!50}\scriptsize $-$\\
  
  \cline{3-5}

 \multirow{-3}{*}{\rotatebox[origin=c]{0}{\scriptsize  \cite{dujmovic2023time}}}&&{\rotatebox[origin=c]{0}{\scriptsize  }}\scriptsize Server&\cellcolor{gray!50}\scriptsize $O(n^{\st 2}+\mxsqr\cdot(\Delta+ \sum\limits_{\st j=1}^{\st z}\Delta_{\st j, \prt}))$&\cellcolor{gray!50}\scriptsize $O(1)$\\
 

 \hline
 

\end{tabular}
}
\end{center}
\end{table*}

%% file: feature-table.tex

\vspace{-6mm}

\begin{table*}[!htb]
\begin{center}
\caption{Comparing features of different (partially) homomorphic TLPs.}
\label{feature-of-our-schemes}

\renewcommand{\arraystretch}{.91}
\scalebox{1}{
\begin{tabular}{|cc|c|c|c|c|c|c|c|c|c|c|} 
   \hline
   
   &&&&\scriptsize Without&\multicolumn{2}{c|}{\scriptsize Supporting Verification}&\scriptsize Flexible&\\
   
   \cline{6-7}

 \multirow{-2}{*}{\scriptsize Schemes}&&\multirow{-2}{*}{\scriptsize Multi-Instance}&\multirow{-2}{*}{\scriptsize Multi-Client}&\multirow{-2}{*}{\scriptsize  Trusted-Setup}&{\scriptsize Client's Solution}&{\scriptsize Linear Combination}&\multirow{-2}{*}{\scriptsize Time Paramters}&\multirow{-2}{*}{\scriptsize Batch Solving}\\

\hline


 {\rotatebox[origin=c]{0}{\scriptsize \text{{\mhtlp}}}}&&\cellcolor{gray!20}{{\scriptsize \checkmark}}&\cellcolor{gray!20}{\rotatebox[origin=c]{0}{\scriptsize $\times$}}&\cellcolor{gray!20}{{\scriptsize \checkmark}}&\cellcolor{gray!20}{{\scriptsize \checkmark}}&\cellcolor{gray!20}{{\scriptsize \checkmark}}&\cellcolor{gray!20}{{\scriptsize \checkmark}}&\cellcolor{gray!20}{{\scriptsize $\times$}}\\      
            
\hline



  {\rotatebox[origin=c]{0}{\scriptsize \text{{\mmhtlp}}}}&&\cellcolor{gray!50}{\rotatebox[origin=c]{0}{\scriptsize \checkmark}}&\cellcolor{gray!50}{\rotatebox[origin=c]{0}{\scriptsize \checkmark}}&\cellcolor{gray!50}{{\scriptsize \checkmark}}&\cellcolor{gray!50}{{\scriptsize \checkmark}}&\cellcolor{gray!50}{{\scriptsize \checkmark}}&\cellcolor{gray!50}{{\scriptsize \checkmark}}&\cellcolor{gray!50}{{\scriptsize $\times$}}\\     

\hline


{\rotatebox[origin=c]{0}{\scriptsize  \cite{tempora-fusion}}}&&\cellcolor{gray!20}{\rotatebox[origin=c]{0}{\scriptsize $\times$}}&\cellcolor{gray!20}{\rotatebox[origin=c]{0}{\scriptsize \checkmark}}&\cellcolor{gray!20}{{\scriptsize \checkmark}}&\cellcolor{gray!20}{{\scriptsize \checkmark}}&\cellcolor{gray!20}{{\scriptsize \checkmark}}&\cellcolor{gray!20}{{\scriptsize \checkmark}}&\cellcolor{gray!20}{{\scriptsize $\times$}}\\

\hline


{\rotatebox[origin=c]{0}{\scriptsize  \cite{MalavoltaT19}}}&&\cellcolor{gray!50}{\rotatebox[origin=c]{0}{\scriptsize $\times$}}&\cellcolor{gray!50}{\rotatebox[origin=c]{0}{\scriptsize \checkmark}}&\cellcolor{gray!50}{{\scriptsize $\times$}}&\cellcolor{gray!50}{{\scriptsize $\times$}}&\cellcolor{gray!50}{{\scriptsize $\times$}}&\cellcolor{gray!50}{{\scriptsize $\times$}}&\cellcolor{gray!50}{{\scriptsize $\times$}}\\

\hline


{\rotatebox[origin=c]{0}{\scriptsize  \cite{liu2022towards}}}&&\cellcolor{gray!20}{\rotatebox[origin=c]{0}{\scriptsize $\times$}}&\cellcolor{gray!20}{\rotatebox[origin=c]{0}{\scriptsize \checkmark}}&\cellcolor{gray!20}{{\scriptsize $\times$}}&\cellcolor{gray!20}{{\scriptsize $\checkmark$}}&\cellcolor{gray!20}{{\scriptsize $\times$}} &\cellcolor{gray!20}{{\scriptsize $\times$}}&\cellcolor{gray!20}{{\scriptsize $\times$}}\\

\hline


{\rotatebox[origin=c]{0}{\scriptsize  \cite{dujmovic2023time}}}&&\cellcolor{gray!50}{\rotatebox[origin=c]{0}{\scriptsize $\times$}}&\cellcolor{gray!50}{\rotatebox[origin=c]{0}{\scriptsize \checkmark}}&\cellcolor{gray!50}{{\scriptsize $\times$}}&\cellcolor{gray!50}{{\scriptsize $\times$}}&\cellcolor{gray!50}{{\scriptsize $\times$}}&\cellcolor{gray!50}{{\scriptsize $\times$}}&\cellcolor{gray!50}{{\scriptsize $\checkmark$}}\\
 

 \hline
 

\end{tabular}
}
\end{center}
\end{table*}

%% file: Puzzle-literature-review.tex


\section{Related Work}\label{Related-Work}

%
Initially, the idea of sending information into the future, i.e., time-lock puzzle/encryption was proposed by Timothy C. May \cite{TimothyMay1993}. 
A basic property of a time-lock scheme is that generating a puzzle takes less time than solving it. 
May’s scheme relies on a trusted agent to release a secret at the appropriate time for a puzzle to be solved, which can be a significant assumption. To address this, Rivest \textit{et al.} \cite{Rivest:1996:TPT:888615} proposed an RSA-based TLP that eliminates the need for a trusted agent. This scheme relies on sequential modular squaring and remains secure even against a receiver with extensive computational resources running in parallel. 

Following the development of the RSA-based TLP, several new variations have emerged. For instance, Boneh \et. \cite{BonehN00} and Garay \et. \cite{DBLP:conf/fc/GarayJ02}  introduced TLPs tailored for situations where a potentially malicious client must provide zero-knowledge proofs to assure the server that the correct solution will be revealed after a predetermined time. Additionally, Baum \et. \cite{BaumDDNO21} devised a composable TLP that can be defined and validated within the universal composability framework. To ensure security against adversaries equipped with quantum computers   Agrawal \et. \cite{crypto-2024-34246} proposed a TLP based on lattices, believed to be post-quantum secure. In the remainder of this section, we will discuss two variants that are most closely related to our schemes.

\subsection{Multi-instance Time Lock Puzzle}

To enhance the scalability of TLPs, researchers have examined scenarios where a server receives multiple puzzles or instances from a client simultaneously and needs to solve all of them. The TLPs proposed in \cite{ChvojkaJSS21,Abadi-C-TLP} address this setting. 
The scheme in \cite{ChvojkaJSS21} relies on an asymmetric-key encryption method, unlike the symmetric-key approach used in traditional TLPs, and it lacks verification capabilities. To overcome these limitations, a TLP is introduced in \cite{Abadi-C-TLP}. This TLP employs an efficient hash-based verification, enabling the server to prove the correctness of the solutions. However, this scheme is effective only for time intervals of identical size. Later, a multi-instance TLP in \cite{abadi2023delegated} addressed this limitation. Nevertheless, none of these multi-instance TLPs support  (verifiable) homomorphic computations on puzzles.  Their functionality is limited to solving and verifying puzzles.


\subsection{Homomorphic Time-lock Puzzles} 

The notion of homomorphic TLPs was introduced by  Malavolta and Thyagarajan \et. \cite{MalavoltaT19}. It allows arbitrary functions to be applied to puzzles before they are solved. This scheme incorporates the RSA-based TLP and fully homomorphic encryption.
To enhance efficiency, partially homomorphic TLPs have been developed, including those that support homomorphic linear combinations or multiplications of puzzles \cite{MalavoltaT19,liu2022towards}. Unlike fully homomorphic TLPs, partially homomorphic TLPs do not depend on fully homomorphic encryption, leading to more efficient implementations. In contrast to the partially homomorphic TLP described in \cite{MalavoltaT19}, the TLPs presented in \cite{liu2022towards} offer additional features. Firstly, they enable a verifier to confirm that puzzles have been correctly generated. Secondly, they allow verification of the server’s solution for a single client’s puzzle, though they do not support the verification of solutions related to homomorphic computations. 

Later, Srinivasan \et. \cite{SrinivasanLMNPT23} noted that existing homomorphic TLPs are limited in their ability to handle a high number of puzzles in a batch, as solving one puzzle leads to discovering all solutions in the batch. To address this limitation, they proposed a scheme based on indistinguishability obfuscation and puncturable pseudorandom functions. 
To improve the efficiency of this scheme,  Dujmovic \et. \cite{dujmovic2023time} introduced a new method that does not rely on indistinguishability obfuscation but instead uses pairings and learning with errors. Both schemes assume that all initial puzzles'  solutions will be discovered simultaneously.

All the homomorphic TLPs discussed, except the scheme proposed in \cite{SrinivasanLMNPT23}, require a fully trusted setup and assume the server will act honestly during computation. Incorporating a trusted third party undermines the core objective of RSA-based TLPs, which is to eliminate the need for such a party. Otherwise, this third party could simply hold the secret and release it to the recipient at the designated time, thereby compromising the design's intent. 
Recently, Abadi \cite{tempora-fusion} introduced a partially homomorphic TLP that enables anyone to verify the correctness of the server’s solution, whether for a single client’s puzzle or the computation itself. This scheme operates without relying on a trusted third party. However, all of the above schemes assume each client has only a single puzzle. They lack support for the multi-instance setting. If they are used directly in this setting, they would impose a high computation cost on the server.



%
%
%
%
%

%% file: preliminaries.tex


\section{Preliminaries}\label{sec::prelminaries}

\subsection{Notations and Assumptions}\label{notations}

We consider the case where a server \srv is given multiple instances of a puzzle (or multiple puzzles) by a single client at once, to solve them.  We consider a generic case where client \prtt has a vector of messages: ${\vec{m}}=[m_{\st 1},\ldots, $ $m_{\st z}]$. It wants \srv to learn each message $m_{\st i}$ at time $time_{\st i}\in {\vec{time}}$, where $ {\vec{time}}=[time_{\st 1},\ldots, time_{\st z}]$ and $ time_{\st j-1}< time_{\st j}$. We define a time interval between two consecutive points in time as $\bar\Delta_{\st j}=time_{\st j}-time_{\st j-1}$. We define a vector of time intervals as $ {\vec{\Delta}}=[\bar\Delta_{\st 1},\ldots, \bar\Delta_{\st z}]$.

The standard parameter $\mxsqr$ denotes the maximum number of squaring that a solver (with the highest level of computation resources) can perform per second. We define a time parameter $T_{\st j}=\mxsqr\cdot \bar\Delta_{\st j}$ and $1\leq j \leq z$. Time points  $time_{\st 0}$ and $time'_{\st 0}$ refer to the only times when \prtt is available and online, where $time_{\st 0}, time'_{\st 0}< time_{\st 1}$. We define $\Delta_{\st j}$ the time period where the $j$-th message remains hidden; therefore,  $\Delta_{\st j}$ can be written as $\Delta_{\st j}=\sum\limits^{\st j}_{\st i=1}\bar{\Delta}_{\st i}$. Server \srv must learn the computation result (i.e., a linear combination of messages) after a certain time $\Delta$, where $\Delta<\bar\Delta_{\st 1}$.

We define $U$ as the universe of a solution $m_{\st j}$. We denote by $\lambda\in \mathbb{N}$ the security parameter. 
For certain system parameters, we use polynomial $poly(\lambda)$ to state the parameter is a polynomial function of  $\lambda$. 
%
%
We define a public vector $\vec{x}$ as $\vec{x}=[x_{\st 1}, \cdots, x_{\st m}]$, where $x_{\st i}\neq x_{\st j}$, $x_{\st i}\neq 0$, and $x_{\st i}\notin U$. 

We define a hash function $\g: \{0,1\}^{\st *} \rightarrow  \{0,1\}^{\st poly(\lambda)}$, that maps arbitrary-length message to a message of length $poly(\lambda)$. 
We denote a null value or set by $\bot$. By $||v||$ we mean the bit-size of $v$ and by $||\vec{v}||$ we mean the total bit-size of elements of $\vec{v}$. 
We denote by $\prm$ a large prime number, where $\log_{\st 2}(\prm)$ is the security parameter, e.g., $\log_{\st 2}(\prm)=128$. A set of leaders are involved in \mmhtlp,  we denote the total number of leaders with \tl. We set $\cor=\tl+2$. For a value $v$ defined over a finite field $\mathbb{F}_{\st \prm}$ (of prime order $\prm$), by $v^{\st -1}$ we mean the multiplicative inverse of $v$, i.e., $v\cdot v^{\st -1}\equiv 1 \bmod \prm$.

To ensure generality, we adopt notations from zero-knowledge proof systems \cite{BlumSMP91,FeigeLS90}. Let $R_{\st \cmd}$ be an efficient binary relation that consists of pairs of the form $(stm_{\st \cmd}, wit_{\st \cmd})$, where $stm_{\st \cmd}$ is a statement and $wit_{\st \cmd}$ is a witness. Let $\mathcal{L}_{\st \cmd}$ be the language (in $\mathcal{NP}$) associated with $R_{\st \cmd}$, i.e., $\mathcal{L}_{\st \cmd}=\{stm_{\st \cmd}|\ \exists wit_{\st \cmd}  \text{ s.t. }$ $ R(stm_{\st \cmd}, wit_{\st \cmd})=1 \}$. A (zero-knowledge) proof for $\mathcal{L}_{\st \cmd}$ allows a prover to convince a verifier that $stm_{\st \cmd}\in \mathcal{L}_{\st \cmd}$ for a common
input $stm_{\st \cmd}$ (without revealing $wit_{\st \cmd}$). In this paper, two main types of verification occur (1) verification of a single client's puzzle solution, in this case, $\cmd=\scp$, and (2)  verification of a linear combination, in this case, $\cmd=\ep$. 
In this work, we assume parties interact through a secure channel. Moreover, we consider a malicious (or active) adversary that will corrupt the server.



\subsection{Pseudorandom Function}\label{sec::prf}

Informally, a pseudorandom function is a deterministic function that takes a key of length $\lambda$ and an input; and outputs a value. Informally, the security of \prf states that the output of \prf is indistinguishable from that of a truly random function.  In this paper, we use pseudorandom functions:   $\mathtt {PRF}:  \{0,1\}^{\st *}  \times\{0,1\}^{\st poly(\lambda)} \rightarrow  \mathbb{F}_{\st \prm}$. 
In practice, a pseudorandom function can be obtained from an efficient block cipher \cite{DBLP:books/crc/KatzLindell2007}. In this work, we use $\mathtt {PRF}$ to derive pseudorandom values to blind (or encrypt) secret messages.

\subsection{Oblivious Linear Function Evaluation}\label{sec::OLE-plus}

Oblivious Linear function Evaluation (\ole) is a two-party protocol that involves a sender and receiver. In this setting,  the sender has two inputs  $a, b\in \mathbb{F}_{\st \prm}$ and the receiver has a single input, $c \in \mathbb{F}_{\st \prm}$.  The protocol enables the receiver to learn only $s = a\cdot c + b \in \mathbb{F}_{\st \prm}$, while the sender learns nothing. Ghosh \textit{et al.} \cite{GhoshNN17} proposed an efficient \ole that has $O(1)$ overhead and mainly uses symmetric-key operations.\footnote{The scheme uses an Oblivious Transfer (OT) extension as a subroutine. However, the OT extension requires only a constant number of public-key-based OT invocations. The rest of the OT invocations are based on symmetric-key operations. The exchanged messages in the OT extension are defined over a small-sized field, e.g.,  a field of size $128$-bit \cite{AsharovL0Z13}.}

Later, in \cite{GhoshN19} an enhanced \ole, denoted as $\ole^{\st +}$, was introduced. $\ole^{\st +}$ ensures that the receiver cannot learn anything about the sender's inputs,  even if it sets its input to $0$. $\ole^{\st +}$ is also accompanied by an efficient symmetric-key-based verification mechanism that enables a party to detect its counterpart's misbehavior during the protocol's execution. In this paper, we use $\ole^{\st +}$ to securely switch the blinding factors of secret messages (encoded in the form of puzzles) held by a server.  We refer readers to Appendix \ref{apndx:F-OLE-plus}, for the construction of $\ole^{\st +}$.

\subsection{Polynomial Representation of a Message}\label{sec:Polynomial-Representation-of-Message}
Encoding a message $m$ as a polynomial $\bm{\pi}(x)$ allows imposing a certain structure on the message. Polynomial representation has previously been used in different contexts, for example in secret sharing \cite{Shamir79},  private set intersection \cite{DBLP:conf/crypto/KissnerS05}, or error-correcting codes \cite{reed1960polynomial-}. 
To encode  a message $m$ in a polynomial $\bm{\pi}(x)$, we can use one of the following approaches,  (1) setting $m$ as the constant terms of $\bm{\pi}(x)$, e.g.,  $m+\sum\limits^{\st n}_{\st j=1}x^{\st j}\cdot a_{\st j}\bmod \prm$ 
%
%
and (2) setting $m$ as the root of $\bm{\pi}(x)$, e.g.,  $\bm{\pi}(x)=(x-m)\cdot \bm\tau(x)\bmod \prm$. In this paper, we utilize both approaches. 
The former approach allows us to perform a linear combination of the constant terms of different polynomials. 
Furthermore, we use the latter approach to insert a secret random root (or a trap) into the polynomials encoding the messages. Hence, the resulting polynomial representing the linear combination maintains this root, allowing the verification of the correctness of the computation.



In general, polynomials can be represented in the point-value form, in the following way. A polynomial $\bm\pi(x)$ of degree $n$ can be represented as a set of $l$ ($l>n$) point-value pairs $\{(x_{\st 1},\pi_{\st 1}),\ldots,$ $(x_{\st l},\pi_{\st l})\}$ such that all $x_{\st i}$ are distinct  non-zero points and $\pi_{\st i}=\bm\pi(x_{\st i})$ for all $i$, $1\le i\le l$.  A polynomial
in this form can be converted into coefficient form via polynomial interpolation, e.g., via Lagrange interpolation~\cite{aho19}.  

 Arithmetic of polynomials in point-value representation can be done by adding or multiplying the corresponding $y$-coordinates of polynomials. 
Let  $a$ be a scalar and $\{(x_{\st 1},\pi_{\st 1}),\ldots,$ $(x_{\st l},\pi_{\st l})\}$  be $(y,x)$-coordinates of a polynomial $\bm{\pi}(x)$. Then, the polynomial $\bm{\theta}$ interpolated from  $\{(x_{\st 1}, a\cdot\pi_{\st 1}),\ldots,$ $(x_{\st l},a\cdot\pi_{\st l})\}$ is the product of $a$ and polynomial $\bm{\pi}(x)$, i.e., $\bm{\theta}(x)=a\cdot \bm{\pi}(x)$.

\subsection{Unforgeable Encrypted Polynomial with a Hidden Root}\label{Unforgeable-Encrypted-Polynomial}

The idea of the ``unforgeable encrypted polynomial with a hidden root'' was introduced in \cite{AbadiTD16}. Informally, it can be explained as follows. Let us consider a polynomial $\bm{\pi}(x)\in \mathbb{F}_{\st \prm}[x]$ with a random secret root $\beta$. We can represent $\bm{\pi}(x)$ in the point-value form and then encrypt its $y$-coordinates. We give all the $x$-coordinates and encrypted $y$-coordinates to an adversary and we locally delete all the $y$-coordinates. The adversary may modify any subset of the encrypted $y$-coordinates and send back to us the encrypted $y$-coordinates. If we decrypt all the $y$-coordinates sent by the adversary and then interpolate a polynomial $\bm{\pi}'(x)$, the probability that $\bm{\pi}'(x)$ has $\beta$ has a root is negligible in the security parameter $\lambda=\log_{\st 2}(\prm)$. Below, it is formally stated. 


%

\begin{theorem}[Unforgeable Encrypted Polynomial with a Hidden Root]\label{theorem::Unforgeable-Encrypted-Polynomial}
Let $\bm{\pi}(x)$ be a polynomial of degree $n$ with a random root $\beta$, and  $\{(x_{\st 1},\pi_{\st 1}),\ldots,$ $(x_{\st l},\pi_{\st l})\}$ be point-value representation of 
$\bm{\pi}(x)$, where $l>n$, $\prm$ denotes a large prime number,  $\log_{\st 2}(\prm)=\lambda'$ is the security parameter, $\bm{\pi}(x)\in \mathbb{F}_{\st \prm}[x]$, and $\beta\stackrel{\st\$}\leftarrow \mathbb{F}_{\st \prm}$. Let $o_{\st i}=w_{\st i}\cdot (\pi_{\st i}+z_{\st i}) \bmod \prm$ be the encryption of each $y$-coordinate $\pi_{\st i}$ of $\bm{\pi}(x)$, using values  $w_{\st i}$ and $r_{\st i}$ chosen uniformly at random from $\mathbb{F}_{\st \prm}$. Given $\{(x_{\st 1}, o_{\st 1}), \ldots,$ $(x_{\st l}, o_{\st l})\}$, the probability that an adversary (which does not know $(w_{\st 1}, r_{\st 1}),\ldots,(w_{\st l}, r_{\st l}),  \beta$) can forge $[o_{\st 1}, \ldots, o_{\st l}]$ to arbitrary $\vec{\ddot{o}}=[\ddot{o}_{\st 1}, \ldots, \ddot{o}_{\st l}]$, such that: (i) $\exists \ddot{o}_{\st i}\in \vec{\ddot{o}}, \ddot{o}_{\st i}\neq {o}_{\st i}$, and (ii) the polynomial $\bm{\pi}'(x)$ interpolated from unencrypted $y$-coordinates $\{\big(x_{\st 1}, (w_{\st 1}\cdot \ddot{o}_{\st 1})-z_{\st l}\big), \ldots,$ $(x_{\st l}, \big(w_{\st l}\cdot \ddot{o}_{\st l})-z_{\st l}\big)\}$ will have root $\beta$ is negligible in $\lambda'$, i.e.,  
$Pr[\bm{\pi}'(\beta)\bmod \prm=0]\leq \mu(\lambda')$. 
\end{theorem}

We refer readers to \cite{AbadiTD16} for the proof of Theorem \ref{theorem::Unforgeable-Encrypted-Polynomial}. This paper uses the unforgeable encrypted polynomial with a hidden root concept to detect a server's misbehaviors, a technique also used in \cite{tempora-fusion}.

\subsection{Commitment Scheme}\label{subsec:commit}

A commitment scheme involves a sender and a receiver, and it includes two phases:  commitment and opening. During the commitment phase, the sender commits to a message, using algorithm $\comcom(m,r)=com$, where $m$ is the message and $r$ is a secret value randomly chosen from $ \{0,1\}^{\st poly(\lambda)}$.  Once the commitment phase concludes, the sender forwards the commitment, $com$, to the receiver. During the opening phase, the sender transmits the pair $\hat{m}:=(m, r)$ to the receiver. The receiver verifies the correctness of the pair using the algorithm $\comver({com},\hat{m})$. It accepts the pair if the output of the verification algorithm is $1$. 

A commitment scheme offers two main properties hiding and binding. Hiding ensures that it is computationally infeasible for an adversary, i.e., the receiver, to gain any knowledge about the committed message $m$ before $com$ is opened. Binding ensures that it is computationally infeasible for an adversary, i.e., the sender, to open $com$ to different values $\hat{m}':=(m',r')$ than the ones originally used during the commit phase. Thus, it should be infeasible to find an alternative pair $\hat{m}'$ such that $\comver({com},\hat{m})=\comver({com},\hat{m}')=1$, where $\hat{m}\neq \hat{m}'$. 

There is a standard efficient hash-based commitment scheme that involves computing $\mathtt{Q}(m||r)={com}$ during the commitment.  The verification step requires checking $\mathtt{Q}(m||r)\stackrel{\st ?}={com}$. Note that $\mathtt{Q}:\{0,1\}^{\st *}\rightarrow \{0,1\}^{\st poly(\lambda)}$ is a collision-resistant hash function, where the probability of finding two distinct values $m$ and $m'$ such that $\mathtt{Q}(m)=\mathtt{Q}(m')$ is negligible with regard to the security parameter $\lambda$. In this paper, we use this hash-based commitment scheme to identify a server's misbehaviors.



\

\subsection{Time-Lock Puzzle}\label{sec::Time-lock-Encryption} 
In this section, we restate the TLP's formal definition. 

\begin{definition}\label{Def::Time-lock-Puzzle} A TLP consists of three algorithms: $(\mathsf{Setup_{\st TLP},}$ $\mathsf{GenPuzzle_{\st TLP}}$ $\mathsf{, Solve_{\st TLP}})$. It meets completeness and efficiency properties. TLP involves a client \prtt and a server \srv.

\begin{itemize}[leftmargin=.37cm]
\item \textit{Algorithms}:
\begin{itemize}
\item$\mathsf{Setup_{\st TLP}}(1^{\st\lambda},\Delta, \mxsqr)\rightarrow (pk,sk)$. A probabilistic algorithm run by \prtt. It takes as input a security parameter, $1^{\st \lambda}$, time parameter $\Delta$ (in seconds) that specifies how long a message must remain hidden, and time parameter $\mxsqr$ which is the maximum number of squaring that a solver (with the highest level of computation resources) can perform per second. It outputs pair $(pk,sk)$ that contains public and private keys. 

\item $\mathsf{GenPuzzle_{\st TLP}}(m, pk, sk)\rightarrow {o}$. A probabilistic algorithm executed by \prtt. It takes as input a solution $m$ and $(pk,sk)$. It outputs a puzzle $o$.

\item $\mathsf{Solve_{\st TLP}}(pk, {o})\rightarrow s$. A deterministic algorithm run by \srv. It takes as input  $pk$ and $ {o}$. It outputs a solution $s$.

\end{itemize}

\item \textit{Completeness}. For any honest \cl and \se, it always holds that $\mathsf{Solve_{\st TLP}}(pk, o)=m$.

\item \textit{Efficiency}. The run-time of  $\mathsf{Solve_{\st TLP}}(pk, {o})$ is upper-bounded by $\bar{poly}(\Delta,\lambda)$, with a fixed polynomial $\bar{poly}$. 
\end{itemize}
\end{definition}

The security of a TLP requires that the puzzle's solution remains confidential from all adversaries running in parallel within the time period, $\Delta$. It also requires that an adversary cannot extract a solution in time $\delta(\Delta)<\Delta$, using a polynomial number of processors  ${poly}(\Delta)$ that run in parallel and after a large amount of pre-computation. 

\vspace{-1.4mm}
\begin{definition} A TLP is secure if for all $\lambda$ and $\Delta$, all probabilistic polynomial time (PPT) adversaries $\mathcal{A}:=(\mathcal{A}_{\st 1},\mathcal{A}_{\st 2})$ where $\mathcal{A}_{\st 1}$ runs in total time $O(poly(\Delta,\lambda))$ and $\mathcal{A}_{\st 2}$ runs in time $\delta(\Delta)<\Delta$ using at most ${poly}(\Delta)$ parallel processors, there is a negligible function $\mu()$, such that: 
\vspace{-1mm}
$$ Pr\left[
  \begin{array}{l}
\mathsf{Setup_{\st TLP}}(1^{\st\lambda},\Delta, \mxsqr)\rightarrow (pk, sk)\\
\mathcal{A}_{\st 1}(1^{\st\lambda},pk, \Delta)\rightarrow (m_{\st 0},m_{\st 1},\text{state})\\
b\stackrel{\st\$}\leftarrow \{0,1\}\\
\mathsf {GenPuzzle_{\st TLP}}(m_{\st b}, pk, sk)\rightarrow  {o}\\
\hline
\mathcal{A}_{\st 2}(pk,  {o},\text{state})  \rightarrow b \\

\end{array}
\right]\leq \frac{1}{2}+\mu(\lambda)$$
where $\delta(\Delta)=(1-\epsilon)\Delta$ and $\epsilon<1$, according to \cite{BonehBBF18}.
\end{definition}

Note that, in the literature, the TLP definition includes two adversaries $\mathcal{A}_{\st 1}$ and $\mathcal{A}_{\st 2}$ because their run-times are different. We refer readers to Appendix \ref{sec::RSA-based-TLP} for the description of the original RSA-based TLP, which is the core of the majority of TLPs. 
By definition, TLPs are sequential functions.  Their construction requires that a sequential function, such as modular squaring, is invoked iteratively a fixed number of times. The sequential function and iterated sequential functions notions, in the presence of an adversary possessing a polynomial number of processors, are formally defined in~\cite{BonehBBF18}. We restate the definitions in Appendix \ref{sec::equential-squering}. 

\subsection{\tf: A TLP with Efficient Verifiable Homomorphic Linear Combination}\label{sec::prelimi-tempora-fusion}
Recently, Abadi \cite{tempora-fusion} proposed a TLP, called \tf, that supports an efficient verifiable homomorphic linear combination of puzzles, where which belongs to a different client. Since our schemes are built on top of \tf, we briefly describe it in the remainder of this section. \tf mainly relies on (1) a polynomial representation of a message, (2)  an unforgeable encrypted polynomial, (3) oblivious linear function evaluation, and (4) the original RSA-based TLP initially introduced by Rivest \et. \cite{Rivest:1996:TPT:888615}.

They consider a malicious server and semi-honest clients. In their scheme, the malicious server is considered strong in the sense that it can act arbitrarily and access the secret keys and parameters of a subset of clients. The scheme designates at random a subset of clients as leaders. Let \idx be this subset, containing $\tl$ leaders. The scheme allows the malicious server to gain access to the secret keys and parameters of some of these leaders. Specifically,  for a set $P=\{\srv, \prtt_{\st 1},\dots, \prtt_{\st n}\}$ containing all the parties involved, the scheme allows the adversary adaptively corrupt a subset $\corupt$ of $P$. It will fully corrupt \srv and act arbitrarily on its behalf. It can retrieve the secret keys of a subset of clients in $P$. They define a threshold $t$ and require the number of non-corrupted leaders (i.e., the parties in \idx) to be at least $t$. This setting allows a high number of clients to be corrupted by the adversary, under a constraint.
For instance, when $|P|=500$, and the total number of leaders is $10$ (i.e., $\tl=10$), and $t=2$, then the adversary may corrupt $498$ parties in $P$ (i.e., $|\corupt|=98$), as long as at most $8$ parties from $\mathcal{I}$ are in \corupt, i.e., $|\corupt\ \cap\ \idx|\leq t$. 

\subsubsection{Outline of \tf.} Initially, each client creates a puzzle that encodes a secret solution and sends it to a server \srv. Upon receiving each puzzle, \srv processes it to determine the solution within a specified timeframe. The time at which each client's solution is found can be distinct from others. Upon finding a solution, \srv generates a proof that asserts the solution's correctness. 
Later, possibly after a certain amount of time has passed since sending their puzzles to  \srv, some clients whose puzzles remain undiscovered may collaborate and request  \srv to homomorphically combine their puzzles. The combined puzzles will encode a linear combination of their solutions. The plaintext result of this computation will be determined by \srv after a specific timeframe, potentially before any of their puzzles are solved.

To enable \srv to impose a specific structure (that can be considered as a hidden ``trap'') on its outsourced puzzles and homomorphically combine them, each client interacts with \srv and also communicates a message to other clients. After a certain period, \srv can find the solution to the puzzle encoding the computational result. It generates proof asserting the correctness of the solution, which anyone can efficiently verify to ensure that the result maintains the intended structure. Eventually, \srv discovers the solution to each client's puzzle as well. It then publishes both the solution and a proof that enables anyone to verify the validity of the solution too. Appendix \ref{sec::tf-protocol} presents a detailed description of \tf. 


%% file: definition-of-HTLP.tex

\section{Definition of Multi-Instance Verifiable Partially Homomorphic TLP}\label{sec::definition}

\input{definition}

\subsection{Syntax of \vhtlp}

\begin{definition}[Syntax] A Multi-Instance Verifiable Homomorphic Linear Combination TLP (\vhtlp) scheme consists of six algorithms: $\vhtlp=(\ssetup, \csetup, \pgen, \eval, \solv, \ver)$, defined below.

\begin{itemize}[label=$\bullet$]

\item \underline{$\ssetup(1^{\st \lambda})\rightarrow \kp_{\st \srv}$}. It is an algorithm run by the server \srv. It takes a security parameter $1^{\st \lambda}$. It generates a pair $\kp_{\st \srv}:=(sk_{\st \srv}, pk_{\st \srv})$, that includes a set of secret  parameters $sk_{\st \srv}$ and  a set of public parameters $pk_{\st \srv}$. 
It returns $\kp_{\st \srv}$. Server \srv publishes $pk_{\st \srv}$. 

\item \underline{$\csetup(1^{\st \lambda})\rightarrow \kp$}.  It is a probabilistic algorithm run by a client $\prtt$. It takes security parameter $1^{\st \lambda}$ as input. It returns a pair $\kp:=(sk, pk)$ of secret key $sk$ and public key $pk$. Client $\prtt$ publishes $pk$.


\item \underline{$\pgen(\vec{m}, \kp, pk_{\st \srv}, \vec\Delta, \mxsqr)\rightarrow(\vec{o}, prm)$}. 
 %
It is an algorithm run by $\prtt$. It takes as input a vector $\vec{m}=[m_{\st 1},\ldots, m_{\st z}]$ of plaintext solutions, key pair $\kp$, server's public parameters set $pk_{\st \srv}$, a vector $\vec\Delta=[\Delta_{\st 1},\ldots,$ $ \Delta_{\st z}]$ of time parameters, where each $\Delta_{\st j}$ determines the period in which $m_{\st j}$ should remain private, and the maximum number $\mxsqr$ of sequential steps  (e.g., modular squaring) per second that the strongest solver can perform. 
It returns (i) a vector $\vec{o}=[\vec{o}_{\st 1}, \ldots, \vec{o}_{\st z}]$, where each  $\vec{o}_{\st j}$ represents a single puzzle corresponding to $m_{\st j}$, and (ii) a pair  $prm:= (sp, pp)$ of secret parameter $sp$ and public  parameter $pp$ of the puzzles.  Client $\prtt$ publishes $(\vec{o}, pp)$.

\item \underline{$\eval(\langle \srv(\vec{o}, \Delta, $ $ \mxsqr, {pp}, {pk}, pk_{\st \srv}), \prtt(\Delta, \mxsqr, \kp, prm, q_{\st 1}, pk_{\st \srv}), \ldots, \prtt(\Delta, $ $\mxsqr, $ $\kp, prm, q_{\st z},$  $pk_{\st \srv}) \rangle)\rightarrow$}\\ \underline{$(\vec{g}, \vec{pp}^{\st(\text{Evl})})$}. 
%
%
It is an (interactive) algorithm run by  \srv (and client $\prtt$). When no interaction between \srv and \prtt is required, the client's inputs will be null $\bot$. Server \srv takes as input vector $\vec{o}$ of $z$ puzzles, time parameter $\Delta$ within which the evaluation result should remain private (where $\Delta<min(\Delta_{\st 1},\ldots, \Delta_{\st z})$),   \mxsqr, ${pp}$,  ${pk}$, and  $pk_{\st \srv}$. For each puzzle $\vec{o}_{\st j}$, client $\prtt$ takes as input $\Delta, \mxsqr,  \kp, prm$, coefficient $q_{\st j}$, and $pk_{\st \srv}$. The algorithm returns a puzzle vector  $\vec{g}$, representing a single puzzle and a vector of public parameters $\vec{pp}^{\st (\text{Evl})}$.  Server \srv publishes $\vec{g}$ and client \prtt publish $\vec{pp}^{\st (\text{Evl})}$. 

\item \underline{$\solv(\vec{o}, pp, \vec{g}, \vec{pp}^{\st (\text{Evl})}, pk, pk_{\st \srv},\cmd)\rightarrow(\vec{m}, \vec\zeta)$}. It is a deterministic algorithm run by \srv. It takes as input client $\prtt$'s puzzle vector $\vec{o}$, the puzzle's public parameter $pp$, a vector $ \vec{g}$ representing the puzzle that encodes evaluation of the client's puzzles, a vector of public parameter $\vec{pp}^{\st (\text{Evl})}$, $pk$, $pk_{\st \srv}$, and a command \cmd, where $cmd \in \{\text{\scp, \ep}\}$. When $\cmd=\scp$, it solves each puzzle $\vec{o}_{\st j}$ in $\vec{o}$ (which is an output of $\pgen()$), this yields a solution  $m_{\st j}$. It appends each $m_{\st j}$ to $\vec{m}$. In this case, input $\vec{g}$ can be $\bot$. When $\cmd=\ep$, it solves puzzle $\vec{g}$ (which is an output of $\eval()$), which results in a solution $m$. It appends $m$ to $\vec{m}$. In this case, input $\vec{o}$ can be $\bot$. Depending on the value of \cmd, it generates a proof vector $\vec\zeta$ (asserting that $m \textnormal{ or } m_{\st j}\in \mathcal{L}_{\st \cmd}$ for each solution). It outputs plaintext solution(s) $\vec{m}$ and proof $\vec\zeta$. Server \srv publishes $(\vec{m}, \vec\zeta)$.

\item \underline{$\ver(m, \zeta, \vec{o}_{\st j}, pp, \vec{g}, \vec{pp}^{\st (\text{Evl})}, pk_{\st \srv}, \cmd)\rightarrow \ddot{v}\in\{0,1\}$}. It is a deterministic algorithm run by a verifier. It takes as input a plaintext solution $m$, proof $\zeta$, puzzle $\vec{o}_{\st j}$, 
public parameters  $pp$ of  $\vec{o}_{\st j}$,  a puzzle $\vec{g}$ for a linear combination of the puzzles, public parameters $\vec{pp}^{\st (\text{Evl})}$ of $\vec{g}$, server's public key $pk_{\st \srv}$, and command \cmd that determines whether the verification corresponds to \prtt's single puzzle or linear combination of the puzzles. It returns $1$ if it accepts the proof. It returns $0$ otherwise.  
\end{itemize}


\end{definition}

In the above syntax, the prover is required to generate a witness/proof $\zeta$ for the language $\mathcal{L}_{\st\cmd}=\{stm_{\st\cmd}=(\bar{pp}, m)|\ R_{\st\cmd}(stm_{\st\cmd},$ $ \zeta)=1\}$, where if $\cmd=\scp$, then $\bar{pp}=pp$ and if $\cmd=\ep$,  then $\bar{pp}=\vec{pp}^{\st (\text{Evl})}$.

%


\subsection{Security Model of \vhtlp}\label{sec::Security-Model}

A \vhtlp scheme must satisfy \textit{security} (i.e., privacy and solution-validity), \textit{completeness}, \textit{efficiency}, and \textit{compactness} properties. 
%
%
Each security property of a \vhtlp scheme is formalized through a game between a challenger \chal that plays the role of honest parties and an adversary \adv that controls the corrupted parties. In this section, initially, we define a set of oracles that will strengthen the capability of \adv. After that, we will provide formal definitions of the properties of \vhtlp.

\subsubsection{Oracles.} To allow \adv to adaptively select plaintext solutions and corrupt parties, inspired by \cite{tempora-fusion}, we provide \adv with access to oracles: puzzle generation $\opgen()$ and evaluation $\oeval()$. Moreover, to enable \adv to have access to the messages exchanged between the corrupt and honest parties during the execution of $\eval()$, we define an oracle called $\orev()$. Below, we define these oracles.

\begin{itemize}[label=$\bullet$]

\item $\opgen(st_{\st \chal}, GeneratePuzzle, \vec{m}, \vec\Delta)\rightarrow(\vec{o}, pp)$. It is executed by the challenger \chal.  It receives a query $(GeneratePuzzle, \vec{m}, \vec\Delta)$, where $GeneratePuzzle$ is a string, $\vec{m}=[m_{\st 1},\ldots, m_{\st z}]$ is a vector of plaintext solutions, and $\vec\Delta=[\Delta_{\st 1},\ldots,$ $ \Delta_{\st z}]$ is a vector of time parameters, where each $\Delta_{\st j}$ determines the period within which $m_{\st j}$ should remain private.  \chal retrieves $(\kp, pk, \mxsqr)$ from its state $st_{\st \chal}$ and then executes $\pgen(\vec{m}, \kp, pk_{\st \srv},$ $ \vec\Delta, \mxsqr)\rightarrow(\vec{o}, prm)$, where $prm:= (sp, pp)$. It returns $(\vec{o}, pp)$ to \adv.

 \item $\oeval(st_{\st \chal}, evaluate)\rightarrow (\vec{g}, \vec{pp}^{\st(\text{Evl})})$. It is executed interactively between the challenger \chal and the adversary to run 
 $
\eval(\langle \adv(input_{\st \srv}),  \prtt(\Delta, \mxsqr, \kp, prm, q_{\st 1}, pk_{\st \srv}), \ldots, $ $ \prtt(\Delta,\mxsqr,$ $
 \kp, prm, q_{\st z}, pk_{\st \srv}) \rangle)\rightarrow(\vec{g}, \vec{pp}^{\st (\text{Evl})})$. The inputs of honest parties are extracted by \chal from $st_{\st \chal}$. The adversary may provide arbitrary inputs $input_{\st \srv}$ or the input of an honest \srv (in this case $input_{\st \srv}$ is set to $(\vec{o}, \Delta, \mxsqr, {pp}, {pk}, pk_{\st \srv})$) during the execution of \eval. 
\chal returns $(\vec{g}, \vec{pp}^{\st(\text{Evl})})$ to \adv.

\item $\orev(st_{\st \chal}, reveal^{\st (\text{Evl})}, \vec{g}, \vec{pp}^{\st (\text{Evl})})\rightarrow transcript^{\st (\text{Evl})}$. It is run by $\chal$ which is provided with a corrupt party \srv and a state $st_{\st \chal}$. It receives a query $(reveal^{\st (\text{Evl})}, \vec{g}, \vec{pp}^{\st (\text{Evl})})$, where $Reveal^{\st (\text{Evl})}$ is a string, and pair $( \vec{g}, \vec{pp}^{\st (\text{Evl})})$ is an output pair (previously) returned by an instance of $\eval()$. If the pair  $ (\vec{g}, \vec{pp}^{\st (\text{Evl})})$ was never generated, then the challenger sets $transcript^{\st (\text{Evl})} $ to $\emptyset$. Otherwise, the challenger retrieves from $st_{\st \chal}$ a set of messages that honest parties sent to the corrupt \srv while executing the specific instance of $\eval()$. It appends these messages to $transcript^{\st (\text{Evl})}$ and returns $transcript^{\st (\text{Evl})}$ to \adv.

\end{itemize}

\subsubsection{Properties.} We proceed to formally define the main properties of a \vhtlp scheme. Initially, we will focus on the \textit{privacy} property. 
Informally, {privacy} states that the $j$-th solution $m_{\st j}$ to the $j$-th puzzle $\vec{o}_{\st j}$ must remain concealed for a predetermined period from any adversaries equipped with a polynomial number of processors. Specifically, an adversary with a running time of $\delta(\sum\limits_{\st i=1}^{\st j} \bar{\Delta}_{\st i})<\sum\limits_{\st i=1}^{\st j} \bar{\Delta}_{\st i}$ using at most $poly(Max(\bar{\Delta}_{\st 1}$ $,\ldots, \bar{\Delta}_{\st j}))$ parallel processors is unable to discover a message significantly earlier than $\delta(\sum\limits_{\st i=1}^{\st j} \bar{\Delta}_{\st i})$. This requirement is also applicable to the evaluation result. Specifically, an adversary with a running time of $\delta({\Delta})<\Delta$ using a polynomial number of parallel processors cannot learn the linear combination of messages significantly earlier than $\delta({\Delta})$.

To capture privacy, we define an experiment $\mathsf{Exp}_{\st \textnormal{prv}}^{\st\adv}(1^{\st\lambda}, z)$, in Figure \ref{fig:exp-prv}. This experiment involves a challenger \chal which plays honest parties' roles and a pair of adversaries $\adv=(\adv_{\st 1}, \adv_{\st 2})$. This experiment considers an adversary that has access to the oracles: (i) puzzle generation $\opgen()$, (ii) evaluation $\oeval()$, and (iii) transcript revelation $\orev()$. The adversary initially outputs a key pair and a state (line \ref{expr-prv:keyGen}). Next,  \chal generates a pair of secret and public keys (line \ref{expr-prv:keyGen-cln}). Given the  \chal's public key, and having access to $\opgen()$ and  $\oeval()$, $\adv_{\st 1}$ outputs $z$ pairs of messages $[(m_{\st 0, 1}, m_{\st 1, 1}),\ldots, (m_{\st 0, z}, m_{\st 1, z})]$ (line \ref{expr-prv:pick-messages}).

After that, \chal for each pair of messages, provided by $\adv_{\st 1}$, selects a random bit $b_{\st j}$ and generates a puzzle and related public parameter for the message with index $b_{\st j}$ (lines \ref{expr-prv:puzzle-gen-start}--\ref{expr-prv:puzzle-gen-end}). Given these puzzles and the corresponding public parameters, and having access to $\opgen()$ and $\oeval()$, $\adv_{\st 1}$ outputs a state (line \ref{expr-prv:puzzle-gen-state}). Using this state as input,  $\adv_{\st 2}$ guesses the value of bit $b_{\st j}$ for its selected puzzle (line \ref{expr-prv:puzzle-gen-state-guess}). The adversary wins the game (and accordingly, the experiment outputs 1) if its guess is correct (line \ref{expr-prv:chalenger-check-single-puzzle}).

Next, the adversary $\adv_{\st 1}$ and \chal interactively execute $\eval()$ (line \ref{expr-prv:evaluate}). The experiment enables $\adv_{\st 1}$ to learn about the messages exchanged between the honest parties and the corrupt party, by providing $\adv_{\st 1}$ with access oracle  \orev(). Having access to this transcript and the output of $\eval()$,  $\adv_{\st 1}$ outputs a state (line \ref{expr-prv:reveal}).   
Given this state and the output of $\eval()$, adversary $\adv_{\st 2}$ guesses the value of bit $b_{\st j}$ for its chosen puzzle (line \ref{expr-prv:adv-2-guess-again}). The adversary wins the game and the experiment returns 1 if its guess is correct (line \ref{expr-prv:2nd-wining-condition}).


\begin{definition}[Privacy]\label{def:privacy-vh-tlp} A \vhtlp scheme is privacy-preserving if for any security parameter $\lambda$, any plaintext input solution $m_{\st 1}, \ldots, m_{\st z}$ and coefficient $q_{\st 1}, \ldots, q_{\st z}$ (where each $m_{\st j}$ and $q_{\st j}$ belong to the plaintext universe $U$), any time interval  $\bar{\Delta}_{\st i}$ between two consecutive points in time when two solutions are found, any difficulty parameter $T=\Delta_{_{\st l}}\cdot \mxsqr$ (where $\Delta_{_{\st l}}\in\{ D=\{\Delta_{{\st 1}}, \ldots, \Delta_{{\st z}}\}\cup\{\Delta\}\}$ is the period, polynomial in $\lambda$, within which a message must remain hidden, $\Delta_{\st j}\in D$, $\Delta_{\st j}=\sum\limits^{\st j}_{\st i=1}\bar{\Delta}_{\st i}$, and \mxsqr is a constant in $\lambda$), and any polynomial-size adversary $\adv:=(\adv_{\st 1}, \adv_{\st 2})$, where $\adv_{\st 1}$ runs in time $O(poly( T,\lambda))$ and $\adv_{\st 2}$ runs in time $\delta(T)<T$ using at most $\bar{poly}(T)$ parallel processors, there exists a negligible function $\mu()$ such that for any experiment $\mathsf{Exp}_{\st \textnormal{prv}}^{\st\adv}(1^{\st\lambda}, z)$, presented in Figure \ref{fig:exp-prv}, it holds that: $Pr[\mathsf{Exp}_{\st \textnormal{prv}}^{\st\adv}(1^{\st\lambda}, z)\rightarrow 1]\leq \frac{1}{2}+\mu(\lambda)$. 

\input{privacy-def}
\end{definition}


Informally, for a solution to be considered \textit{valid}, it must be infeasible for a probabilistic polynomial time (PPT) adversary to generate an invalid solution that can still pass the verification process. This holds true whether the solution is for an individual puzzle posed by the client or for a puzzle encoding a linear combination of the client's messages. 
To capture solution validity, we define an experiment $\mathsf{Exp}_{\st \textnormal{val}}^{\st\adv}(1^{\st\lambda}, z)$, presented in Figure \ref{fig:exp-vld}. It involves  \chal which plays honest parties' roles and an adversary $\adv$.

The adversary initially outputs a key pair and a state (line \ref{exp-val:pick-adv-output-key-xxxx}). Next, \chal generates a pair of secret
and public keys (line \ref{expr-val:genKey-cln}). Given the state and the public key, as well as having access to $\opgen()$ and  $\oeval()$, $\adv$ outputs a message vector $\vec{m}= [m_{\st  1},\ldots, m_{\st  z}]$  (line \ref{expr-val:pick-messages}). Next, \chal generates a puzzle for each message that \adv selected (line \ref{expr-vld:puzzle-gen-end}).  
The experiment allows \adv and \chal to interactively execute $\eval()$ (line \ref{expr-val:evaluate}). 
Given the output of $\eval()$ which is itself a puzzle, \chal solves the puzzle and outputs the solution and associated proof (line \ref{expr-val:solvPul-2}).

The experiment lets $\adv$ learn about the messages exchanged between the corrupt party and \chal (acting as honest parties) during the execution of $\eval()$, by providing $\adv$ with access to  \orev(). Given this transcript, the output of $\eval()$, the plaintext solution, and the proof,  $\adv$ outputs a solution and proof  (line \ref{expr-val:reveal}).  
After that, \chal verifies the solution and proof provided by \adv. The experiment outputs 1 and \adv wins if \adv convinces \chal to accept an invalid evaluation result (line \ref{expr-val:ver-2}). 
\chal solves every puzzle of the client (line \ref{expr-val:end-solvPul}). Given the puzzles and the solutions, and having access to oracles $\opgen()$ and $\oeval()$, \adv provides a solution and proof for its selected puzzle (line \ref{expr-val:guess-1}). Next, \chal checks the validity of the solution and proof provided by \adv. The experiment outputs 1 (and \adv wins) if \adv persuades \chal to accept an invalid message for a puzzle of the client (line \ref{expr-val:ver-1}).

\begin{definition}[Solution-Validity]\label{def:validity-vh-tlp} A \vhtlp scheme preserves a solution validity,  if for any security parameter $\lambda$, any plaintext input solution $m_{\st 1}, \ldots, m_{\st z}$ and coefficient $q_{\st 1}, \ldots, q_{\st z}$ (where each $m_{\st j}$ and $q_{\st j}$ belong to the plaintext universe $U$), any time interval  $\bar{\Delta}_{\st i}$ between two consecutive points in time when two solutions are found, any difficulty parameter $T=\Delta_{_{\st l}}\cdot \mxsqr$ (where $\Delta_{_{\st l}}\in\{ D=\{\Delta_{{\st 1}}, \ldots, \Delta_{{\st z}}\}\cup\{\Delta\}\}$ is the period, polynomial in $\lambda$, within which a message must remain hidden, $\Delta_{\st j}\in D$, $\Delta_{\st j}=\sum\limits^{\st j}_{\st i=1}\bar{\Delta}_{\st i}$, and \mxsqr is a constant in $\lambda$), and any polynomial-size adversary $\adv$ that runs in time $O(poly( T,\lambda))$, there is a negligible function $\mu()$ such that for any experiment $\mathsf{Exp}_{\st \textnormal{val}}^{\st\adv}(1^{\st\lambda}, z)$ (presented in Figure \ref{fig:exp-vld}), it holds that: $Pr[\mathsf{Exp}_{\st \textnormal{val}}^{\st\adv}(1^{\st\lambda}, z)\rightarrow 1]\leq \mu(\lambda)$. 




\input{solution-validity-def}

\end{definition}



Informally, \textit{completeness} examines the behavior of algorithms when honest parties are involved. It asserts that $\solv()$ will always retrieve (1) a correct solution for a puzzle related to a linear combination and (2) a set of correct solutions for $z$ puzzles. It also asserts that $\ver()$ will always return 1 when given an honestly generated solution. 
%
%
Because algorithm $\solv()$ is used for two cases:  (a) to find a solution for a single puzzle generated by a client and (b) to discover a solution for a puzzle that encodes the linear evaluation of messages, we will state the correctness of this algorithm separately for each case. Similarly, we will state the correctness of $\ver()$ for each scenario. 
In the following definitions, because the experiments’ descriptions are relatively short, we integrate the experiment $\mathsf{Exp}$ into the probability notation. Thus, we use the notation $\Pr\left[\begin{array}{c}\mathsf{Exp}\\
\hline\mathsf{Cond}\\\end{array}\right]$, where $\mathsf{Exp}$ is an experiment, and $\mathsf{Cond}$ represents the set of the conditions under which the property must hold.

\begin{definition}[Completeness]\label{def:completeness-vh-tlp} A \vhtlp  is correct if for any security parameter $\lambda$, any plaintext input message $m_{\st 1}, \ldots, m_{\st z}$ and coefficient $q_{\st 1}, \ldots, q_{\st z}$ (where each $m_{\st j}$ and $q_{\st j}$ belong to the plaintext universe $U$), any difficulty parameter $T=\Delta_{_{\st l}}\cdot \mxsqr$ (where $\Delta_{_{\st l}}$ is the period, polynomial in $\lambda$, within which a message  must remain hidden and \mxsqr is a constant in $\lambda$), the following conditions are satisfied.

\begin{enumerate}

\item\label{completeness-of-single-puzzle} $\solv(\vec{o}, pp, ., ., pk, pk_{\st \srv}, \cmd)$, that takes a vector  $\vec{o}$ of puzzles  each encoding a plaintext solution $m_{\st j}$ and its related parameters, always returns a vector $\vec{m}=[{m}_{\st 1}, \ldots, {m}_{\st z}]$ of solutions:
{\small{
\[
Pr\left[
\begin{array}{l}
\ssetup(1^{\st \lambda})\rightarrow \kp_{\st \srv}\\
\csetup(1^{\st \lambda})\rightarrow \kp\\
\pgen(\vec{m}, \kp, pk_{\st \srv}, \vec\Delta, \mxsqr)\rightarrow(\vec{o}, prm)\\
\hline
\solv(\vec{o}, pp, ., ., pk, pk_{\st \srv}, \cmd)\rightarrow(\vec{m}, .)\\
\end{array}
\right]=1\]
}}

where $pp\in prm$ and  $\cmd=\scp$.

\item\label{completeness-of-linear-comb-puzzle} $\solv(., .,\vec{g},  \vec{pp}^{\st (\text{Evl})},  pk, pk_{\st \srv}, \cmd)$, that takes (i) a puzzle $\vec{g}$ encoding linear combination $\sum\limits_{\st j=1}^{\st z} q_{\st j}\cdot m_{\st j}$ of $z$ messages, where each  $m_{_{\st j}}$ is a plaintext message and $q_{\st j}$ is a coefficient and (ii) their related parameters, always returns $\sum\limits_{\st j=1}^{\st z} q_{\st j}\cdot m_{\st j}$:

{\small{
\[
Pr\left[
\begin{array}{l}
\ssetup(1^{\st \lambda})\rightarrow \kp_{\st \srv}\\
\csetup(1^{\st \lambda})\rightarrow \kp\\
\pgen(\vec{m}, \kp, pk_{\st \srv},  \vec\Delta, \mxsqr)\rightarrow(\vec{o}, prm)\\
\eval\big(\langle \srv(\vec{o}, \Delta, \mxsqr, {pp}, {pk}, pk_{\st \srv}), \prtt(\Delta, \mxsqr, \kp, prm,\\ q_{\st 1}, pk_{\st \srv}), \ldots, \prtt(\Delta, \mxsqr, \kp, prm, q_{\st z}, pk_{\st \srv}) \rangle\big)\rightarrow(\vec{g}, \vec{pp}^{\st (\text{Evl})})\\
\hline
\solv(., .,\vec{g}, \vec{pp}^{\st (\text{Evl})},  pk, pk_{\st \srv}, \cmd)\rightarrow(\sum\limits_{\st j=1}^{\st z} q_{\st j}\cdot m_{\st j},.)\\
%
\end{array}
\right]=1\]
}}

where  
$\vec{o}=[\vec{o}_{\st {1}},\dots, \vec{o}_{\st {z}}]$, 
%
%
%
$pk\in  \kp$, 
$pk_{\st \srv}\in \kp_{\st \srv}$, 
and 
$\cmd=\ep$.

\item\label{correctness-verify-single-puzzle}

$\ver(m, \zeta, \vec{o}_{\st j}, pp, ., ., pk_{\st \srv}, \cmd)$, that takes a solution for a puzzle, related proof, and public parameters, always returns 1:

{\small{
\[
Pr\left[
\begin{array}{l}
\ssetup(1^{\st \lambda})\rightarrow \kp_{\st \srv}\\
\csetup(1^{\st \lambda})\rightarrow \kp\\

\pgen(\vec{m}, \kp, pk_{\st \srv},  \vec\Delta, \mxsqr)\rightarrow(\vec{o}, prm)\\
%
%
\solv(\vec{o}, pp,., ., pk, pk_{\st \srv}, \cmd)\rightarrow (\vec{m}, \vec{\zeta})\\
\hline
\ver(m, \zeta, \vec{o}_{\st j}, pp, ., ., pk_{\st \srv}, \cmd)\rightarrow 1
\end{array}
\right]=1\]
}}

where $\cmd=\scp$.

\item\label{completeness-verification-for-HLC} $\ver(m, \zeta, .,., \vec{g}, \vec{pp}^{\st (\text{Evl})}, pk_{\st \srv}, \cmd)$, that takes a solution for a puzzle that encodes a linear combination of $z$ messages, related proof, and public parameters, always returns 1:

{\small{
\[
Pr\left[
\begin{array}{l}
\ssetup(1^{\st \lambda})\rightarrow \kp_{\st \srv}\\
%
%
\csetup(1^{\st \lambda})\rightarrow \kp\\

\pgen(\vec{m}, \kp, pk_{\st \srv},  \vec\Delta, \mxsqr)\rightarrow(\vec{o}, prm)\\

\eval(\langle \srv(\vec{o}, \Delta,   \mxsqr, {pp}, {pk}, pk_{\st
 \srv}), c(\Delta, \mxsqr, \kp, prm, q_{\st 1},\\ pk_{\st \srv}), \ldots, c(\Delta,  \mxsqr,  \kp, prm, q_{\st z},  pk_{\st \srv}) \rangle)\rightarrow(\vec{g}, \vec{pp}^{\st(\text{Evl})})\\

%
%
\solv(., ., \vec{g}, \vec{pp}^{\st (\text{Evl})}, pk, pk_{\st \srv}, \cmd)\rightarrow(\vec{m}, \vec{\zeta})\\
\hline
\ver(m, \zeta, .,., \vec{g}, \vec{pp}^{\st (\text{Evl})}, pk_{\st \srv}, \cmd)\rightarrow 1
\end{array}
\right]=1\]
}}

where $\cmd=\ep$.

\end{enumerate}

\end{definition}

Now, we move on to the \textit{efficiency} property, and along the way, we formally define the multiple instance notion for the first time. Intuitively, {efficiency} states that (1) the running time of dealing with multiple instances of a puzzle through multi-instance TLP is faster than via traditional TLPs,  (2) $\solv()$ returns a solution in polynomial time, i.e., polynomial in the time parameter $T$, (3) $\pgen()$ generates a puzzle faster than solving it, with a running time of at most logarithmic in $T$, and (4) the running time of $\eval()$ is faster than solving any puzzle involved in the evaluation, that should be at most logarithmic in $T$ \cite{dujmovic2023time}.

\begin{definition}[Efficiency]\label{def:efficiency-vh-tlp} A \vhtlp is efficient if for any security parameter $\lambda$, any plaintext input message $m_{\st 1}, \ldots, m_{\st z}$ and coefficient $q_{\st 1}, \ldots, q_{\st z}$ (where each $m_{\st j}$ and $q_{\st j}$ belong to the plaintext universe $U$), any time interval  $\bar{\Delta}_{\st i}$ between two consecutive points in time when two solutions are found, any difficulty parameter $T=\Delta_{_{\st l}}\cdot \mxsqr$ (where $\Delta_{_{\st l}}\in\{ D=\{\Delta_{{\st 1}}, \ldots, \Delta_{{\st z}}\}\cup\{\Delta\}\}$ is the period, polynomial in $\lambda$, within which a message must remain hidden, $\Delta_{\st j}\in D$, $\Delta_{\st j}=\sum\limits^{\st j}_{\st i=1}\bar{\Delta}_{\st i}$, and \mxsqr is a constant in $\lambda$), the following conditions are satisfied.

\begin{enumerate}


\item\label{efficiency-multi-instance}\underline{\textit{Multi-instance}}: Let $\Psi_{\st single}$ be the time complexity to solve a single traditional time-lock puzzle, e.g., \cite{Rivest:1996:TPT:888615}. It holds that the time complexity of a solving algorithm in traditional time-lock puzzles to deal with $z$ puzzle is $\Psi_{\st trad}(z)= z\cdot \Psi_{\st single}$. Let $\Psi_{\st multi}(z)$ be the time complexity to solve $z$ puzzles using a multi-instance algorithm. The solving algorithm of a multi-instance scheme is efficient if:  $\Psi_{\st trad}(z)-\Psi_{\st multi}(z)\geq {poly}(z,\mxsqr,  \Delta_{\st 1},\ldots, \Delta_{\st z})$, for a fixed polynomial  ${poly}$.


\item\label{efficiency-Polynomial-time-solving}\underline{\textit{Polynomial-time solving}}: The running time of $\solv(\vec{o}, pp, \vec{g}, \vec{pp}^{\st (\text{Evl})}, pk, pk_{\st \srv}, \cmd)$ is upper bounded by $ \hat{poly}(z, $ $T_{\st max}, \lambda)$, where $\hat{poly}()$ is a fixed polynomial and $T_{\st max}=Max(\Delta_{{\st 1}}, \ldots, \Delta_{{\st z}}, \Delta)\cdot \mxsqr$.

\item\label{efficiency-requirement-genpuz}\underline{\textit{Faster puzzle generation}}: The running time of  $\pgen(\vec{m}, \kp, pk_{\st \srv}, \vec\Delta, \mxsqr)$ is upper bounded by $poly'(z,$ $ \log(T_{\st max}), \lambda)$, where $poly'()$ is a fixed polynomial. 

\item\label{efficiency-requirement-Faster-puzzles-evaluation}\underline{\textit{Faster puzzles evaluation}}: The running time of  $\eval(\langle \srv(\vec{o}, \Delta,   \mxsqr, {pp}, {pk}, pk_{\st
 \srv}), c(\Delta, \mxsqr, \kp, prm,$ $ q_{\st 1}, pk_{\st \srv}), \ldots, c(\Delta,  \mxsqr,  \kp,$ $ prm, q_{\st z},  pk_{\st \srv}) \rangle)\rightarrow(\vec{g}, \vec{pp}^{\st(\text{Evl})})$ is upper bounded by $poly''\Big(\log(T), \lambda, \func \big((q_{\st 1}, m_{\st 1}),$ $\ldots,(q_{\st z}, m_{\st z}) \big)\Big)$, where $poly''()$ is a fixed polynomial and $\func()$ is the functionality that computes a linear combination of messages (as stated in Relation \ref{equ::PLC}). 

\end{enumerate}

\end{definition}

Next, we consider  \textit{compactness} which requires that the size of evaluated ciphertexts is independent of the complexity of the evaluation function \func.

\begin{definition}[Compactness]\label{def:compactness-vh-tlp}  A \vhtlp is compact if for any security parameter $\lambda$, any difficulty parameter $T=\Delta_{_{\st l}}\cdot \mxsqr$, any plaintext input solution $m_{\st 1}, \ldots, m_{\st z}$ and coefficient $q_{\st 1}, \ldots, q_{\st z}$ (where each $m_{\st j}$ and $q_{\st j}$ belong to the plaintext universe $U$), any time interval  $\bar{\Delta}_{\st i}$ between two consecutive points in time when two solutions are found, any difficulty parameter $T=\Delta_{_{\st l}}\cdot \mxsqr$ (where $\Delta_{_{\st l}}\in\{ D=\{\Delta_{{\st 1}}, \ldots, \Delta_{{\st n}}\}\cup\{\Delta\}\}$ is the period, polynomial in $\lambda$, within which a message $m$ must remain hidden, $\Delta_{\st j}\in D$, $\Delta_{\st j}=\sum\limits^{\st j}_{\st i=1}\bar{\Delta}_{\st i}$, and \mxsqr is a constant in $\lambda$), always $\eval()$ outputs a puzzle (representation) whose bit-size is independent of \func's complexity $O(\func)$:

{\small{
\[
Pr\left[
\begin{array}{l}
\ssetup(1^{\st \lambda})\rightarrow \kp_{\st \srv}\\
\csetup(1^{\st \lambda})\rightarrow \kp\\
\pgen(\vec{m}, \kp, pk, \vec\Delta, \mxsqr)\rightarrow(\vec{o}, prm)\\
\hline

\eval(\langle \srv(\vec{o}, \Delta,   \mxsqr, {pp}, {pk}, pk_{\st
 \srv}), c(\Delta, \mxsqr, \kp, prm,\\ q_{\st 1}, pk_{\st \srv}), \ldots, c(\Delta,  \mxsqr,  \kp,$ $ prm, q_{\st z},  pk_{\st \srv}) \rangle)\rightarrow(\vec{g}, \vec{pp}^{\st(\text{Evl})})\\
\text{s.t.} \\
 ||\vec{g}||= poly\Big(\lambda, ||\func \big((q_{\st 1}, m_{\st 1}),\ldots,(q_{\st z}, m_{\st z}) \big)||\Big)\\
%
\end{array}
\right]=1\]
}}

\end{definition}


\begin{definition}[Security]\label{def:sec-def-vh-tlp}  A \vhtlp is secure if it satisfies privacy and solution validity as outlined in Definitions \ref{def:privacy-vh-tlp} and \ref{def:validity-vh-tlp}. 

\end{definition}

%% file: definition.tex
\subsection{Private Linear Combination} The basic functionality \func that any $z$-input Private Linear Combination (PLC) computes takes as input a pair of coefficient $q_{\st j}$ and plaintext value $m_{\st j}$  (for every $j, 1\leq j \leq n$), and returns their linear combination $\sum\limits^{\st n}_{\st j=1}q_{\st j}\cdot m_{\st j}$, as stated in \cite{tempora-fusion}. More formally, \func is defined as:
 \begin{equation}\label{equ::PLC}
\func \big((q_{\st 1}, m_{\st 1}),\ldots,(q_{\st z}, m_{\st z}) \big)\rightarrow \sum\limits_{\st j=1}^{\st z} q_{\st j}\cdot  m_{\st j}
%
\end{equation}

We proceed to present a formal definition of Muti-Instance Verifiable Homomorphic Linear Combination TLP (\vhtlp), by presenting the syntax followed by the security and correctness definitions.

%% file: privacy-def.tex


\begin{figure}[h]
\centering
{{
\begin{tcolorbox}[colback=white, width=150mm, colframe=black!!black,title={{$\mathsf{Exp}_{\st \textnormal{prv}}^{\st\adv}(1^{\st\lambda}\text{, }z)$}},fonttitle=\bfseries]
$
  \begin{array}{l}
\setcounter{equation}{0}

\lnum\label{expr-prv:keyGen} \adv_{\st 1}(1^{\st \lambda})\rightarrow (\kp_{\st \srv}:=(sk_{\st \srv}, pk_{\st \srv}), state)\\
\lnum\label{expr-prv:keyGen-cln}\csetup(1^{\st \lambda})\rightarrow \kp:=(sk, pk)\\
\lnum \vec{b}\leftarrow \bf{0}\\
\lnum\label{expr-prv:pick-messages}  \adv_{\st 1}(state, pk, \opgen(), \oeval())\rightarrow \vec{m}=[(m_{\st 0, 1}, m_{\st 1, 1}),\ldots, (m_{\st 0, z}, m_{\st 1, z})]\\
\lnum\label{expr-prv:puzzle-gen-start} \textnormal{For }  j=1,\ldots, z \textnormal{ do}: \\
\lnum \quad b_{\st j}\stackrel{\st \$}\leftarrow\{0, 1\}\\

\lnum \quad \vec{b}[j]\leftarrow b_{\st j}\\
\lnum\label{expr-prv:puzzle-gen-end}   \pgen(\vec{m}'=[m_{\st b_{\st 1}},\ldots, m_{\st b_{\st z}}], \kp, pk_{\st \srv}, \vec\Delta, \mxsqr)\rightarrow(\vec{o}, prm)\\
\lnum\label{expr-prv:puzzle-gen-state} \adv_{\st 1}(state, pk, \opgen(), \oeval(), \vec{o},  {pp})\rightarrow   state    \\ 
\lnum\label{expr-prv:puzzle-gen-state-guess} \adv_{\st 2}(\vec{o},  {pp},  state)\rightarrow (b'_{\st j}, j)\\
\lnum\label{expr-prv:chalenger-check-single-puzzle} \textnormal{If } b'_{\st j}=\vec{b}[j], \textnormal{then return } 1\\
\lnum\label{expr-prv:evaluate} \eval(\langle \adv_{\st 1}(\vec{o}, \Delta, \mxsqr, {pp}, {pk}, pk_{\st \srv}, state),  {\prtt}(\Delta, \mxsqr, \kp, prm, q_{\st 1}, pk_{\st \srv}), \ldots, {\prtt}(\Delta,\mxsqr,\\
\hspace{4.2mm} \kp, prm, q_{\st z}, pk_{\st \srv}) \rangle)\rightarrow(\vec{g}, \vec{pp}^{\st (\text{Evl})})\\
\lnum\label{expr-prv:reveal} \adv_{\st 1}(state, \orev(),  \vec{g}, \vec{pp}^{\st (\text{Evl})})\rightarrow   state    \\ 
\lnum\label{expr-prv:adv-2-guess-again} \adv_{\st 2}(\vec{o}, {pp}, \vec{g}, \vec{pp}^{\st (\text{Evl})}, state)\rightarrow (b'_{\st i}, i)\\
\lnum\label{expr-prv:2nd-wining-condition} \textnormal{If } b'_{\st i}=\vec{b}[i], \textnormal{then return } 1,
\textnormal{else return }0
   \end{array} 
$
\end{tcolorbox}
}
}
\caption{The $\mathsf{Exp}_{\st \textnormal{prv}}^{\st\adv}$ experiment.}
\label{fig:exp-prv}
\end{figure}

%% file: solution-validity-def.tex


\begin{figure}[h]
\centering
{{
\begin{tcolorbox}[colback=white!,colframe=black!!black,title={{$\mathsf{Exp}_{\st \textnormal{val}}^{\st\adv}(1^{\st\lambda}\text{, }z)$}},fonttitle=\bfseries]
$
  \begin{array}{l}
\setcounter{equation}{0}
\lnum\label{exp-val:pick-adv-output-key-xxxx} \adv_{\st 1}(1^{\st \lambda})\rightarrow (\kp_{\st \srv}:=(sk_{\st \srv}, pk_{\st \srv}), state)\\

%

%
%
\lnum\label{expr-val:genKey-cln} \csetup(1^{\st \lambda})\rightarrow \kp:=(sk, pk)\\
%
%
%
\lnum\label{expr-val:pick-messages}  \adv(state, pk, \opgen(), \oeval())\rightarrow \vec{m}=[m_{\st 1}, \ldots, m_{\st z}]\\
 %
 %
%
\lnum\label{expr-vld:puzzle-gen-end}   \pgen(\vec{m}, \kp, pk_{\st \srv}, \vec\Delta, \mxsqr)\rightarrow(\vec{o}, prm)\\

\lnum \label{expr-val:evaluate} \eval(\langle \adv_{\st 1}(\vec{o}, \Delta, \mxsqr, {pp}, {pk}, pk_{\st \srv}, state),  {\prtt}(\Delta, \mxsqr, \kp, prm, q_{\st 1}, pk_{\st \srv}), \ldots, {\prtt}(\Delta,\mxsqr,\\
\hspace{3.5mm} \kp, prm, q_{\st z}, pk_{\st \srv}) \rangle)\rightarrow(\vec{g}, \vec{pp}^{\st (\text{Evl})})\\

\lnum\label{expr-val:solvPul-2} \solv(., ., \vec{g}, \vec{pp}^{\st (\text{Evl})}, pk, pk_{\st \srv}, \ep)\rightarrow(\vec{m}, \vec\zeta)\\

\lnum\label{expr-val:reveal} \adv(state, pk_{\st \srv}, \orev(), \vec{m}, \vec\zeta,  \vec{o}, \vec{g}, pp, \vec{pp}^{\st (\text{Evl})})  \rightarrow  (m', \zeta') \\

\lnum \label{expr-val:ver-2}\textnormal{If } \ver(m', \zeta', .,., \vec{g},  \vec{pp}^{\st (\text{Evl})}, pk_{\st \srv}, \ep)\rightarrow 1, \text{s.t.}\   m'\notin \mathcal{L}_{\ep}, \textnormal{then return } 1 \\

\lnum \label{expr-val:end-solvPul} \solv(\vec{o}, pp, ., ., pk, pk_{\st \srv}, \scp)\rightarrow(\vec{\hat{m}}, \vec{\hat{\zeta}})\\

\lnum\label{expr-val:guess-1}  \adv(state, pk_{\st \srv}, \opgen(), \oeval(), \vec{m}, \vec{\zeta},  \vec{\hat{m}}, \vec{\hat\zeta},  \vec{o},  
{pp}, \vec{pp}^{\st (\text{Evl})})  \rightarrow (m'_{\st j}, \zeta'_{\st j}, j)\\
\lnum\label{expr-val:ver-1}  \textnormal{If } 
\ver(m'_{\st j}, \zeta'_{\st j}, \vec{o}_{\st j}, {pp}, .,., pk_{\st \srv}, \scp)\rightarrow 1, \text{s.t.}\   m'_{\st j}\notin \mathcal{L}_{\scp},   \textnormal{then return } 1; \textnormal{otherwise, return } 0\\
   \end{array} 
$
\end{tcolorbox}
}}
\caption{The $\mathsf{Exp}_{\st \textnormal{val}}^{\st\adv}(1^{\st\lambda}, z)$ experiment.}
\label{fig:exp-vld}
\end{figure}

%% file: V-HLC.tex

\input{V-H-C-TLP-protocol}

%% file: V-H-C-TLP-protocol.tex

\section{Construction of \vhtlp: Multi-instance Verifiable Partially Homomorphic TLP  (\mhtlp)}\label{sec::Multi-Instance-TF}

\subsection{An Overview} We develop \mhtlp on top of \tf, discussed in Section \ref{sec::prelimi-tempora-fusion}. Inspired by previous multi-instance TLPs proposed in \cite{Abadi-C-TLP}, we use the chaining technique, to chain different puzzles, such that when \srv solves one puzzle it will obtain enough information to work on the next puzzle.  Therefore, it can solve puzzles \textit{sequentially} rather than dealing with all of them at once. However, we develop a new technique for chaining puzzles to support homomorphic linear combinations as well. 

In the multi-instance TLPs introduced in  \cite{Abadi-C-TLP}, during the puzzle generation phase (and chaining process) the base $r_{\st j+1}$ for the $(j+1)$-th puzzle is concatenated and concealed with the  $j$-th solution, $m_{\st j} $. Consequently,  the $j$-th puzzle is created on solution $m_{\st j}||r_{\st j+1}$.  The issue with this approach is that it cannot support homomorphic operations on the puzzle's actual solution $m_{\st j}$, because the solution is now $m_{\st j}||r_{\st j+1}$. 

To address this issue, we take a different approach. In short, we require the client to derive the next puzzle's base from the current puzzle's master key. 
We briefly explain how it works. Recall that in \tf, each $j$-th puzzle is associated with a master key $mk_{\st j}$, found when the puzzle is solved. Using our new technique, during the puzzle generation, when $j=1$, (similar to \tf) client \prtt picks a random base $r_{\st j}$, sets $a_{\st j}=2^{\st T_{\st j}}\bmod \phi(N)$, and then sets master key $mk_{\st j}$ as $mk_{\st j}= r^{\st a_{{\st j}}}_{\st j}\bmod N$, with a difference being $T_{\st j}$ now determines how long a solution $m_{\st j}$ must remain concealed after the previous solution $m_{\st j-1}$ is discovered. Next, as in \tf, \prtt derives some pseudorandom values from $ mk_{\st j}$ and encrypts the $y$-coordinates of a polynomial representing $m_{\st j}$.

However, when $j>1$, client \prtt derives a fresh base $r_{\st j} $ from the \textit{previous master key}  as $r_{\st j}=\prf(j||0, mk_{\st j-1})$. As before, it  sets $a_{\st j}=2^{\st T_{\st j}}\bmod \phi(N)$ and then sets the current master key $mk_{\st j}$ as $mk_{\st j}= r^{\st a_{{\st j}}}_{\st j}\bmod N$.  It derives some pseudorandom values from $ mk_{\st j}$, and encrypts the $y$-coordinates of a polynomial representing $m_{\st j}$. It repeats this process until it creates a puzzle for the last solution $m_{\st z}$.  The client hides all bases except the first one $r_{\st 1}$ which is made public.

Clearly, this new approach does not require \prtt to modify each solution $m_{\st j}$. Thus, we can now use the techniques developed for \tf in \cite{tempora-fusion} to allow \srv to perform homomorphic linear combinations on the puzzles. 
In this setting, \srv has to solve puzzles in ascending order, starting from the first puzzle, to find a solution and a base for the next puzzle. Given this base, it starts repeated modular squaring to find the next solution and base until it finds the last puzzle's solution.

 We proceed to briefly explain how  \mhtlp operates. We highlight that there are overlaps between the design of \mhtlp and \tf. However, there are major differences as well. We will shortly discuss these differences.
 
 \begin{enumerate}

 \item \textit{Server-Side Setup.} Initially, \srv generates and publishes a set of public parameters, without requiring it to generate any secret parameter. The public parameters include a sufficiently large prime number $p$ and a vector $\vec{x}=[x_{\st 1}, x_{\st 2}, x_{\st 3}]$ of non-zero elements. The elements in $\vec{x}$ can be considered as $x$-coordinates and will help each client to represent its message as a polynomial in point-value form.

\item \textit{Client-side Key Generation.}  A client $\prtt$ independently generates a set of public and private keys, which includes an RSA-public modulus $N$. The client publishes its public key.

\item\textit{Client-side Puzzle Generation.} Client $\prtt$ uses its secret key and time parameter $T_{\st j}$   that determines how long a solution $m_{\st j}$ must remain concealed after the previous solution $m_{\st j-1}$ is discovered, to generate a set of master keys $mk_{\st 1},\ldots, mk_{\st z}$, using the chaining technique described above. $\prtt$ uses each master key $mk_{\st j}$, that is associated with the $j$-th message, to derive pseudorandom values $(z_{\st i, j}, w_{\st i, j})$ for each element $x_{\st i}$ of $\vec{x}$. The client represents its secret solution $m_{\st j}$ as (a constant term of) a polynomial and then represents the polynomial in the point-value form. This results in a vector of $y$-coordinates: $[\pi_{\st 1, j},\pi_{\st 2, j}, \pi_{\st 3, j}]$. It encrypts each $y$-coordinate using the related pseudorandom values:  $o_{\st i, j} = w_{\st i, j}\cdot(\pi_{\st i, j} +  z_{\st i,j}) \bmod \prm$. Each vector of encrypted $y$-coordinates $\vec{o}_{\st j}=[o_{\st 1, j}, o_{\st 2, j}, o_{\st 3, j}]$ represents a puzzle for $j$-th solution $m_{\st j}$. Note that polynomial representation is used to allow a homomorphic linear combination of solutions and efficient verification of the computation. $\prtt$ publishes these puzzles along with a set of public parameters. Given these public parameters, by solving the related puzzle, anyone can sequentially find each master key $mk_{\st j}$ and remove the blinding factors to extract the corresponding solution. To support the public verifiability of a solution that will be found by \srv, $\prtt$ commits $com_{\st j} = \comcom(m_{\st j}, mk_{\st j})$ to each message $m_{\st j}$ using the related $mk_{\st j}$ as the commitment's randomness. It publishes each $com_{\st j}$. Anyone who solves a puzzle encoding solution $m_{\st j}$ can find $mk_{\st j}$ and prove $(m_{\st j}, mk_{\st j})$ matches $com_{\st j}$.

 \item\textit{Linear Combination.} In this phase, similar to \tf,  client $\prtt$ produces certain messages that allow server \srv to find a linear combination of its plaintext
solutions after a certain time. $\prtt$, using its secret key and time parameter $Y$ (which determines how long the result of the computation must remain private), generates a {temporary} master key $tk$ and public parameters ${pp}^{\st(\text{Evl})}$.  Anyone who solves a puzzle for the computation will be able to find $tk$, after a certain time. $\prtt$ uses $tk$ to derive new pseudorandom values $(z'_{\st i}, w'_{\st i})$ for each element $x_{\st i}$ of $\vec{x}$.  
It also uses its secret key to {regenerate} the pseudorandom values $(z_{\st i, j}, w_{\st i, j})$ used to encrypt each $y$-coordinate related to its solution $m_{\st j}$.  The client selects a single random root: $root$. It commits to this root, using $tk$ as the randomness:  $com'=\comcom(\rt, tk)$.  $\prtt$ represents $root$ as a polynomial in point-value form. This yields a vector of $y$-coordinates: $[\gamma_{\st 1},\gamma_{\st 2}, \gamma_{\st 3}]$. 
%
%
The client will insert this random root into its outsourced puzzle, to give a certain structure to the result of the homomorphic linear combination, which will ultimately allow future verification.

For each solution $m_{\st 2},\ldots, m_{\st z}$ it possesses,  the client selects a fresh key $f_{\st {\st l}}$. This key is used by the client to generate zero-sum pseudorandom values for (each $y$-coordinate related to) each solution. These values are generated such that if are summed, they will cancel out each other. They are used to ensure that \srv learns only the linear combination of all the solutions that belong to $\prtt$, and \srv cannot learn the combination of a subset of these solutions. 
$\prtt$ for each $y$-coordinate of each $j$-th puzzle participates in an instance of $\ole^{\st +}$ with \srv. At a high level, $\prtt$'s input includes the $y$-coordinate $\gamma_{\st i}$ of the random root, the new pseudorandom values $(z'_{\st i}, w'_{\st i})$, the inverse of the old pseudorandom values $(w_{\st i, j})^{\st -1}$, its coefficient $q_{\st j}$, and the pseudorandom values $y_{\st i,j}$ derived from $f_{\st {\st j}}$. The input of \srv is the client's $j$-th puzzle. Each instance of $\ole^{\st +}$ returns to \srv an encrypted $y$-coordinate. At the end of this process, $\prtt$ publishes its public parameters ${pp}^{\st(\text{Evl})}$.

Consequently,  \srv sums the outputs of $\ole^{\st +}$ component-wise, resulting in  a vector of encrypted $y$-coordinates, $\vec{g}=[g_{\st 1}, g_{\st 2}, g_{\st 3}]$. Then, \srv publishes $\vec{g}$. Each $g_{\st i}$ has a layer of blinding factor, inserted by $\prtt$ during the invocation of $\ole^{\st +}$. Anyone who solves the related puzzle can regenerate the temporary master key $tk$ and remove the blinding factor.

 \item \textit{Solving a Puzzle}. To find the solution encoding the result of the computation,  \srv operates as follows. 
Given public parameters previously released by $\prtt$, \srv solves the related puzzle to find the temporary master key $tk$. Using $tk$,  \srv removes the blinding factors from each $g_{\st i}$, that results in three plaintext $y$-coordinates. 
\srv uses these $y$-coordinates and the $x$-coordinates in $\vec{x}$ to interpolate a polynomial $\bm{\theta}(x)$. It finds the roots of $\bm{\theta}(x)$. It publishes the root and $tk$ that match the published commitment $com'$. It also retrieves from $\bm{\theta}(x)$ and publishes the linear combination of $\prtt$'s solutions.

To find a solution for a single puzzle of $\prtt$,  \srv operates as follows. 
Given public parameters and each puzzle,  \srv after a certain time retrieves a master key $mk_{\st j}$. Using $mk_{\st j}$, \srv removes the blinding factors from each  
 $o_{\st i,j}$, that yields a vector of $y$-coordinates. It uses them and $x$-coordinates in $\vec{x}$ to interpolate a polynomial $\bm{\pi}_{\st j}(x)$ and retrieves message $m_{\st j}$ from polynomial $\bm{\pi}_{\st j}(x)$. It publishes $m_{\st j}$ and $mk_{\st j}$ that match the published commitment  $com_{\st j}$. Moreover, solving the $j$-th puzzle provides \srv with enough information to begin working on the next puzzle.

 \item \textit{Verification}. To verify a solution related to the linear combination of solutions, the verifier checks if the root $\rt$ and the temporary master key $tk$ provided by \srv match the commitment $com'$. It also removes the blinding factors from $[g_{\st 1}, g_{\st 2}, g_{\st 3}]$, yielding three $y$-coordinates. It uses these $y$-coordinates to interpolate a polynomial. From this polynomial, the verifier retrives its constant term (which is the computation result). It checks whether $\rt$ is a root of this polynomial and the computation result matches the one \srv published. If all the checks pass, it accepts the result. To verify a solution $m_{\st j}$ related to a single puzzle of $\prtt$, the verifier checks if $m_{\st j}$ and the temporary master key $mk_{\st j}$, provided by \srv, match the commitment $com_{\st j}$.

 \end{enumerate}

Beyond offering the multi-instance feature, \mhtlp differs from \tf in two main ways. Firstly, since there is only one client, only one random root is selected and inserted into the outsourced polynomials (rather than inserting $|\idx|$ roots in \tf). During the verification, a verifier checks whether the evaluation of the resulting polynomial at this root is zero.  Secondly, in \mhtlp, using only three $x$-coordinates $\vec{x}=[x_{\st 1}, x_{\st 2}, x_{\st 3}]$ will suffice, because the outsourced puzzle's degree is one, and when it is multiplied by a polynomial representing a random root (while computing the linear combination), the resulting polynomial's degree will become two. Thus, three $(y, x)$-coordinates are sufficient to interpolate a polynomial of degree two. Figure \ref{fig:MH-TLP-workflow} outlines the \mhtlp workflow. 

\begin{figure}[htp]
    \centering
    \includegraphics[width=8cm]{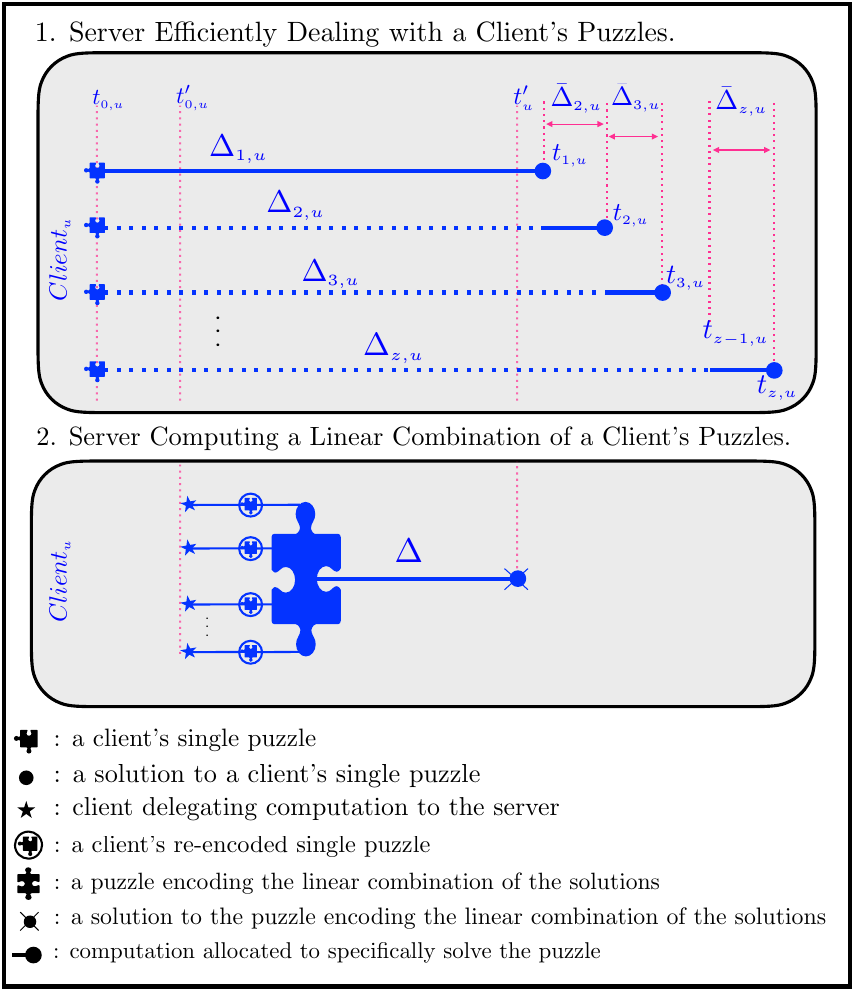}
    \caption{\mhtlp Workflow Outline.}\label{fig:MH-TLP-workflow}
\end{figure}

\subsection{Detailed Construction} Next, we present a detailed description of  \mhtlp.

\begin{enumerate}

\item\label{single-client::Setup-server} {\textbf{Setup}}. $\ssetup(1^{\st \lambda})\rightarrow (., pk_{\st \srv})$

Server \srv only once takes the following steps:

\begin{enumerate}

\item\underline{\textit{Setting a field's parameter}}:  generates a sufficiently large prime number $\prm$, where $\log_{\st 2}(\prm)$ is the security parameter.

\item  \underline{\textit{Generating public $x$-coordinates}}: sets $\vec{x}=[x_{\st 1}, x_{\st 2}, x_{\st 3}]$, where $x_{\st i}\neq x_{\st j}$, $x_{\st i}\neq 0$, and $x_{\st i}\notin U$.

\item \underline{\textit{Publishing public parameters}}: publishes $pk_{\st \srv}=(\prm, \vec{x})$. Note that, $\prtt$ itself can generate $(\prm, \vec{x})$, too.  

\end{enumerate}

\item\label{single-client-KeyGeneration} {\textbf{Key Generation}}. $\csetup(1^{\st \lambda})\rightarrow \kp$

Client $\prtt$ takes the flowing steps.

\begin{enumerate}



\item \underline{\textit{Generating RSA public and private keys}}: computes $N=\prm_{\st 1}\cdot \prm_{\st 2}$, where $\prm_{\st i}$  is a large randomly chosen prime number, e.g., $\log_{\st 2}(\prm_{\st i})\geq 2048$. Next,  it computes Euler's totient function of $N$, as: $\phi(N)=( \prm_{\st 1}-1)\cdot( \prm_{\st 2}-1)$.

\item\label{single-client::Publishing-public-parameters} \underline{\textit{Publishing public parameters}}:  locally keeps secret key $sk=\phi(N)$ and publishes public key $pk=N$. 


\end{enumerate}

\item\label{single-client::PuzzleGeneration} {\textbf{Puzzle Generation}}. $\pgen(\vec{m}, \kp, pk_{\st \srv}, \vec\Delta, \mxsqr)\rightarrow(\vec{o}, prm)$

Client \prtt takes the following steps to generate $z$ puzzles for messages $\vec{m}=[m_{\st 1},\ldots, m_{\st z}]$. It wants \srv to learn each message $m_{\st j}$ at time $time_{\st j}\in {\vec{time}}$, where $ {\vec{time}}=[time_{\st 1},\ldots, time_{\st z}]$, $\bar\Delta_{\st j}=time_{\st j}-time_{\st j-1}$, $ {\vec{\Delta}}=[\bar\Delta_{\st 1},\ldots, \bar\Delta_{\st z}]$, and $1\leq j \leq z$.

%
\begin{enumerate}

\item \underline{\textit{Checking public parameters}}: checks the bit-size of $p$ and  elements of $\vec{x}$ in $pk_{\st \srv}$, to ensure $\log_{\st 2}(p)\geq 128$, $x_{\st i}\neq x_{\st j}, x_{\st i}\neq 0$, and $x_{\st i}\notin U$. If it does not accept the parameters, it returns $(\bot, \bot)$ and does not take further action.

\item\label{sinlge-client:set-expo} \underline{\textit{Generating secret keys}}: generates a vector of master keys $\vec{mk}=[mk_{\st 1}, \ldots, mk_{\st z}]$ and two secret keys $k_{\st j}$ and  $s_{\st j}$ for each  $mk_{\st j}$ in $\vec{mk}$ as follows. It constructs an empty vector $\vec{mk}$. Then, it

\begin{enumerate}

\item\label{single-GC-TLP::compute-a-values} sets each exponent $a_{\st j}$. 
  $$\forall j, 1\leq j\leq z: \quad a_{\st j}=2^{\st T_{\st j}}\bmod \phi(N)$$
  
  where $T_{\st j}=\mxsqr\cdot\bar\Delta_{\st j}$ is the total number of squaring needed to decrypt an encrypted solution $m_{\st j}$ after previous solution $m_{\st j-1}$ is revealed. 
  %

\item\label{single-GC-TLP::compute-r-values} computes each master key $mk_{\st j}$ as follows.  For every $j$, where 
$1\leq j\leq z:$

\begin{itemize}[label=$\bullet$]

\item when $j=1:$

\begin{enumerate}
 \item selects a uniformly random base $r_{\st j}\stackrel{\st\$}\leftarrow \mathbb{Z}_{\st N}$.
 \item sets  key $mk_{\st j}$ as $mk_{\st j}= r^{\st a_{{\st j}}}_{\st j}\bmod N$. 
 \item  appends $mk_{\st j}$ to $\vec{mk}$. 
 \end{enumerate}

\item when $j>1:$

\begin{enumerate}

 \item\label{step::single-client-set-fresh-based-when-j-not-1} derives a fresh base $r_{\st j} $ from the previous master key  as $r_{\st j}=\prf(j||0, mk_{\st j-1})$.
  
 \item sets  key $mk_{\st j}$ as $mk_{\st j}= r^{\st a_{{\st j}}}_{\st j}\bmod N$.
 \item appends $mk_{\st j}$ to $\vec{mk}$.
 
 \end{enumerate}
 
\end{itemize}

\item\label{single-client:derive-keys-for-puzzle} derives two secret keys $k_{\st j}$ and  $s_{\st j}$ from each $mk_{\st j}$.  
$$\forall j, 1\leq j\leq z: \quad k_{\st j}=\prf(1, mk_{\st j}), \hspace{4mm} s_{\st j}=\prf(2, mk_{\st j})$$

\end{enumerate}

%

\item\label{single-client::Generating-blinding-factors}\underline{\textit{Generating blinding factors}}:  generates six pseudorandom values, by using $k_{\st j}$ and  $s_{\st j}$. 
$$\forall j, 1\leq j\leq z\quad \text{and} \quad \forall  i,  1\leq i \leq 3: \hspace{4mm}  z_{\st i, j}=\prf(i, k_{\st j}), \hspace{4mm} w_{\st i, j}=\prf(i, s_{\st j})$$

\item\underline{\textit{Encoding plaintext message}}:  

\begin{enumerate}
\item\label{sinlge-client:rep-message-as-poly} represents each plaintext solution $m_{\st j}$ as a polynomial, such that the polynomial's constant term is the message. 
$$\forall j, 1\leq j\leq z:\quad  \bm{\pi}_{\st j}(x) = x + m_{\st j} \bmod \prm$$

\item\label{sinlge-client:gen-y-coord-of-poly} computes three $y$-coordinates of each  $\bm{\pi}_{\st j}(x)$. 
$$\forall j, 1\leq j\leq z  \quad \text{and} \quad  \forall i, 1\leq i \leq 3: \hspace{4mm}\pi_{\st i, j}= \bm{\pi}_{\st j}(x_{\st i}) \bmod \prm$$

where $x_{\st i}\in \vec{x}$.

\end{enumerate}

\item\label{sinlge-client:enc-y-coord-of-poly}\underline{\textit{Encrypting the solution}}:  encrypts the $y$-coordinates using the blinding factors. 
$$\forall j, 1\leq j\leq z  \quad \text{and} \quad \forall i, 1\leq i \leq 3: \hspace{4mm} o_{\st i, j} = w_{\st i, j}\cdot(\pi_{\st i,j} +  z_{\st i, j}) \bmod \prm$$

\item\label{sinlge-client:commit-y-coord-of-poly}\underline{\textit{Committing to the solution}}:  commits to each plaintext message:   
$$\forall j, 1\leq j\leq z:\quad  com_{\st j} = \comcom(m_{\st j}, mk_{\st j})$$

Let $\vec{com}=[com_{\st 1},\ldots, com_{\st z}]$. 

\item\label{sinlgeclient--Managing-messages}\underline{\textit{Managing messages}}:  publishes  $\vec{o}=\Big[ [o_{\st 1, 1}, o_{\st 2, 1}, o_{\st 3, 1}], \ldots, [o_{\st 1, z}, o_{\st 2, z}, o_{\st 3, z}]\Big]$ and $pp=(\vec{com}, r_{\st 1}) $. It locally keeps secret parameters $sp=\vec{mk}$. 
It sets $prm=(sp, pp)$. It deletes everything else, including $m_{\st j}, \bm{\pi}_{\st j}(x), \pi_{\st 1, j}, \pi_{\st 2, j},$ and $\pi_{\st 3, j}$ for every $j$, $1\leq j\leq z$.

\end{enumerate}

\item\label{phase::multi-instance-TF-Linear-Combination} {\textbf{Linear Combination}}. $\eval(\langle \srv(\vec{o}, \Delta, $ $ \mxsqr, {pp}, {pk}, pk_{\st \srv}), \prtt(\Delta, \mxsqr, \kp, prm, q_{\st 1}, pk_{\st \srv}), \ldots, \prtt(\Delta, $ $\mxsqr,\\ $ $\kp, prm, q_{\st z},$  $pk_{\st \srv}) \rangle)\rightarrow(\vec{g}, {pp}^{\st(\text{Evl})})$

In this phase, client $\prtt$ produces certain messages that allow \srv to find a linear combination of its plaintext solutions after time $\Delta$.

\begin{enumerate}

\item\underline{\textit{Granting the computation}}: client $\prtt$ takes the following steps.

\begin{enumerate}

\item\label{sinlge-client:Generating-temporary-secret-keys} \underline{\textit{Generating temporary secret keys}}: generates a temporary master key $tk$ and two secret keys $k'$ and $s'$. Moreover, it generates $z-1$ secret key $[f_{\st j},\ldots,f_{\st z}]$. To do that, the following steps are taken: $\prtt$ computes the exponent: 
$$b=2^{\st Y}\bmod \phi(N)$$

where $Y=\Delta\cdot \mxsqr$. It selects a  base uniformly at random: $h\stackrel{\st\$}\leftarrow\mathbb{Z}_{\st N}$ and then sets a temporary master key $tk$:
$$tk= h^{\st b} \bmod N$$

It derives two keys from $tk$: 
$$ k'=\prf(1, tk), \hspace{4mm} s'=\prf(2, tk)$$

It picks fresh $z-1$ random keys $\vec{f}=[f_{\st 2},\ldots, f_{\st z}]$, where $f_{\st j}\stackrel{\st \$}\leftarrow \{0, 1\}^{\st poly(\lambda)}$. 

\item\label{single-client:Generating-temporary-pr-vals} \underline{\textit{Generating blinding factors}}: regenerates its original blinding factors, for each $j$-th puzzle. Specifically, for every $j$,  derives two secret keys $k_{\st j}$ and  $s_{\st j}$ from $mk_{\st j}$ as follow.  
$$\forall j, 1\leq j\leq z:  \quad k_{\st j}=\prf(1, mk_{\st j}), \hspace{4mm} s_{\st j}=\prf(2, mk_{\st j})$$

%

Then, it regenerates pseudorandom values, by using $k_{\st j}$ and  $s_{\st j}$. 
$$\forall j, 1\leq j\leq z  \quad \text{and}\quad \forall  i,  1\leq i \leq 3: \hspace{4mm}  z_{\st i, j}=\prf(i, k_{\st j}), \hspace{4mm} w_{\st i, j}=\prf(i, s_{\st j})$$

It also generates new pseudorandom values using keys $(k', s')$: 
 $$\forall i, 1\leq i \leq 3: \hspace{4mm}   z'_{\st i}=\prf(i, k'), \hspace{4mm}   w'_{\st i}=\prf(i, s')$$

It computes new sets of (zero-sum) blinding factors, using the keys in $\vec{f}$, as follows. $\forall j, 1\leq j \leq z:$

\begin{itemize}[label=$\bullet$]

\item if $j=1$: 
$$\forall i, 1\leq i \leq 3: \hspace{4mm}y_{\st i, j}=  -\sum\limits_{\st j = 2}^{\st z}\prf(i, f_{\st j})  \bmod \prm$$

\item if $j>1$: 
$$\forall i, 1\leq i \leq 3: \hspace{4mm}y_{\st i, j}= \prf(i, f_{\st j})         \bmod \prm$$
\end{itemize}
where $f_{\st j}\in \vec{f}$. 

\item\label{sinlge-client:Generating-temporary-root} \underline{\textit{Generating $y$-coordinates of a random root}}: picks a random root, $\rt\stackrel{\st\$}\leftarrow\mathbb{F}_{\st \prm}$.  It represents $\rt$ as a polynomial, such that the polynomial's root is $\rt$. Specifically, it computes polynomial $\bm{\gamma}(x)$ as:
 $$\bm{\gamma}(x) =x- \rt \bmod \prm$$

 Then, it computes three $y$-coordinates of  $\bm{\gamma}(x)$: 
$$\forall i, 1\leq i \leq 3: \hspace{4mm}\gamma_{\st i} = \bm{\gamma}(x_{\st i}) \bmod \prm$$

\item\label{sinlge-client:Committing-to-the-root}  \underline{\textit{Committing to the root}}: computes $ com'=\comcom(\rt, tk)$.

\item\label{single-client::eval:ole-detect} \underline{\textit{Re-encoding outsourced puzzle}}: participates in an instance of $\ole^{\st +}$ with \srv, for every $j$-th puzzle and  every $i$, where $1\leq j \leq z$ and $1\leq i \leq 3$.  The inputs of client \prtt to the $i$-th instance of $\ole^{\st +}$ are: 
$$e_{\st i, j} =\gamma_{\st i}\cdot q_{\st j}\cdot w'_{\st i}\cdot (w_{\st i, j})^{\st -1}\bmod \prm,\quad\quad e'_{\st i, j} = -(\gamma_{\st i}\cdot q_{\st j}\cdot w'_{\st i}\cdot  z_{\st i,j})+z'_{\st i} + y_{\st i, j}\bmod \prm$$ 

The input of \srv to the $(i, j)$-th instance of $\ole^{\st +}$ is corresponding encrypted $y$-coordinate: $e''_{\st i, j} = o_{\st i, j}$. Accordingly, the $(i, j)$-th instance of $\ole^{\st +}$ returns to \srv:
\begin{equation*}
\begin{split}
d_{\st i, j} &= e_{\st i, j}\cdot e''_{\st i, j} + e'_{\st i, j}\\ &=\gamma_{\st i}\cdot q_{\st j}\cdot w'_{\st i}\cdot  \pi_{\st i, j}+z'_{\st i } + y_{\st i, j}\bmod \prm
\end{split}
\end{equation*}

where $q_{\st j}$ is a coefficient for $j$-th solution $m_{\st j}$. 

\item \label{sinlge-client::Publishing-public-parameters--x}\underline{\textit{Publishing public parameters}}: publishes  ${pp}^{\st(\text{Evl})}=(h, com')$. 


\end{enumerate}

\item\label{sinlge-client::Computing-encrypted-linear-combination-}\underline{\textit{Computing encrypted linear combination}}:  Server \srv sums all of the outputs of $\ole^{\st +}$ instances that it has invoked as follows. $\forall i, 1\leq i \leq 3: $

\begin{equation*}
\begin{split}
g_{\st i} &= \sum\limits_{\st j=1}^{\st z} d_{\st {i, j}} \bmod \prm\\ & =    (w'_{\st i}\cdot \gamma_{\st i}  
\cdot \sum\limits_{\st j=1}^{\st z}   q_{\st j}\cdot \pi_{\st i, j} )+    z'_{\st i}  \bmod  \prm
\end{split}
\end{equation*}
Note that in  $g_{\st  i, j}$ there is no $y_{\st  i, j}$, because  $y_{\st  i, j}$ in different $d_{\st i, j}$ have canceled out each other. 

\item\label{single-client::publish-encrypted-solutions}\underline{\textit{Disseminating encrypted result}}:  server \srv publishes $\vec{g}=[g_{\st 1},g_{\st 2}, g_{\st 3}]$.

\end{enumerate}

\item\label{sinlge-client::Solving-a-Puzzle} {\textbf{Solving a Puzzle}}.  $\solv(\vec{o}, pp, \vec{g}, \vec{pp}^{\st (\text{Evl})}, pk, pk_{\st \srv},\cmd)\rightarrow(\vec{m}, \vec{\zeta})$

 Server \srv takes the following steps. 



\begin{steps}[leftmargin=15mm]

\item\label{single-client::solving-linear-combination}\hspace{-2mm}. when solving a puzzle related to the linear combination.

\begin{enumerate}

\item\label{single-client-extract-temp-key}  \underline{\textit{Finding secret keys}}: 

\begin{enumerate}

\item finds temporary key $tk$, where  $tk = h^{\st 2^{\st Y}} \bmod N $, via repeated squaring of $h$ modulo $N$.

\item derives two keys from $tk$: 
$$ k'=\prf(1, tk), \hspace{4mm} s'=\prf(2, tk)$$

\end{enumerate}

%

%

\item\label{sinlge-client::server-side-computatoin-removing-pr-vals} \underline{\textit{Removing blinding factors}}: removes the blinding factors from $[g_{\st 1}, g_{\st 2}, g_{\st 3}]\in \vec{g}$. 

$\forall i, 1\leq i \leq 3:$
\begin{equation*}
\begin{split}
\theta_{\st i}&=\underbrace{\big( \prf(i, s')\big)^{\st -1}}_{\st (w'_{\st i})^{\st -1}}\cdot\big(g_{\st i}- \overbrace{\prf(i, k'}^{\st z'_{\st i}})\big) \bmod \prm\\ &=
\gamma_{\st i}  \cdot \sum\limits_{\st j=1}^{\st z}   q_{\st j}\cdot \pi_{\st i, j} \bmod \prm
\end{split}
\end{equation*}

\item\label{step::sinlge-client-interpolate-poly}  \underline{\textit{Extracting a polynomial}}: interpolates a polynomial $\bm{\theta}$, given pairs $(x_{\st 1}, \theta_{\st 1}), (x_{\st 2}, \theta_{\st 2}), (x_{\st 3}, \theta_{\st 3})$.  Note that $\bm{\theta}$ will have the form: 
$$\bm\theta(x) =(x-\rt)\cdot \sum\limits_{\st j=1}^{\st z}q_{\st j}\cdot (x+m_{\st j}) \bmod \prm$$

We can rewrite $\bm\theta(x)$ as: 
$$\bm\theta(x) = \bm\psi(x)-\rt\cdot  \sum\limits_{\st j=1}^{\st z} q_{\st j}\cdot m_{\st j}   \bmod \prm$$

where $\bm\psi(x)$ is a polynomial of degree two with constant term being $0$.

\item\label{sinlge-client::server-side-Extracting-the-linear-combination} \underline{\textit{Extracting the linear combination}}: retrieves the result (i.e., the linear combination of $m_{\st 1},\ldots, m_{\st z}$)  from polynomial $\bm\theta(x)$'s constant term: $cons=-\rt\cdot  \sum\limits_{\st j=1}^{\st z} q_{\st j}\cdot m_{\st j}$ as follows:
\begin{equation*}
\begin{split}
m &=cons\cdot (-\rt)^{\st -1}\bmod \prm\\ &= \sum\limits_{\st j=1}^{\st z} q_{\st j}\cdot m_{\st j}
\end{split}
\end{equation*}

\item\label{sinlge-client::Extracting-valid-roots} \underline{\textit{Extracting valid roots}}: extracts the root(s) of $\bm{\theta}$. Let set $R$ contain the extracted roots. It identifies the valid root, by finding a root $\rt$ in $R$, such that $\comver(com',(\rt, tk))=1$.

\item\label{single-client-publish-lc-proof} \underline{\textit{Publishing the result}}: initiates  vectors $\vec{m}$ and $\vec\zeta$. It appends $m$ to  $\vec{m}$ and $(\rt, tk)$ to $\vec{\zeta}$. It publishes $\vec{m}$ and $\vec{\zeta}$.

\end{enumerate}

\item\label{single-client-solving-only-one-puzzle}\hspace{-2mm}. when solving each $j$-th puzzle $\vec{o}_{\st j, \prt}$ of client $\prtt_{\st \prt}$ (i.e., when $\cmd=\scp$), server \srv takes the following steps. $\forall j, 1\leq j\leq z:$

\begin{enumerate}

\item\label{single-client::find-base} \underline{\textit{Finding secret bases and keys}}: sets base $r_{\st j}$ and $mk_{\st j}$ as follows. 

\begin{itemize}[label=$\bullet$]
\item if $j=1: $ sets the base to $r_{\st 1}$, where $r_{\st 1}\in pp$. Then, it finds  $mk_{\st 1}$ where  $mk_{\st 1}=r^{\st 2^{\st T_{\st 1}}}_{\st 1}\bmod N$ through repeated squaring of $r_{\st 1}$ modulo $N$. It also initiates  vectors $\vec{m}$ and $\vec\zeta$.

\item if $j>1: $ computes base $r_{\st j}$ as   $r_{\st j}=\prf(j||0, mk_{\st j-1})$. Next, it  finds  $mk_{\st j}$ where  $mk_{\st j}=r^{\st 2^{\st T_{\st j}}}_{\st j}\bmod N$ through repeated squaring of $r_{\st j}$ modulo $N$.
\end{itemize}

It derive two keys from $mk_{\st j}$: 
$$ k_{\st j}=\prf(1, mk_{\st j}), \hspace{4mm} s_{\st j}=\prf(2, mk_{\st j})$$

\item\label{single-client::Removing-blinding-factors-server-side} \underline{\textit{Removing blinding factors}}: re-generates six pseudorandom values using $k_{\st j}$ and  $s_{\st j}$:
$$ \forall i, 1\leq i \leq 3: \hspace{4mm}  z_{\st i, j}=\prf(i, k_{\st j}), \hspace{4mm} w_{\st i, j}=\prf(i, s_{\st j})$$

Next, it uses the blinding factors to unblind $\vec{o}_{\st j} = [o_{\st 1, j}, o_{\st 2, j}, o_{\st 3, j}]$:
$$\forall i, 1\leq i \leq 3: \hspace{4mm}  \pi_{\st i, j}  = \big((w_{\st i, j})^{\st -1}\cdot o_{\st i, j}\big) -z_{\st i, j} \bmod \prm$$


\item\label{single-client::Extracting a polynomial-server-side}   \underline{\textit{Extracting a polynomial}}: interpolates a polynomial $\bm{\pi}_{\st j}$, given pairs $(x_{\st 1}, \pi_{\st 1, j}), (x_{\st 2}, $ $\pi_{\st 2, j}), $ $(x_{\st 3}, \pi_{\st 3, j})$.

\item\label{single-client-pub-solution-} \underline{\textit{Publishing the solution}}: considers the constant term of $\bm{\pi}_{\st j}$ as the plaintext message, $m_{\st j}$. It appends $(m_{\st j}, j)$ to $\vec{m}$ and $mk_{\st j}$ to $\vec{\zeta}$. If $j=z$, then it publishes $\vec{m}$ and $\vec{\zeta}$.


\end{enumerate}

\end{steps}

\item\label{sinlge-client::verification} {\textbf{Verification}}. $\ver(m, \zeta, ., pp, \vec{g}, \vec{pp}^{\st (\text{Evl})}, pk_{\st \srv}, \cmd)\rightarrow \ddot{v}\in\{0,1\}$

A verifier (that can be anyone) takes the following steps.

\begin{steps}[leftmargin=15mm]

\item\label{single-client::verifying-linear-combination}\hspace{-2mm}. when verifying a solution related to the linear combination, i.e., when $\cmd=\ep$:

\begin{enumerate}


\item\label{step-single-client-verify-opening-}  \underline{\textit{Checking the commitment's opening}}: verify the validity of  $(\rt, tk)\in \zeta$, provided by \srv in step \ref{single-client-publish-lc-proof} of \ref{single-client::solving-linear-combination}: 
$$\comver\big(com', (\rt , tk )\big)\stackrel{\st?}=1$$

If the verification passes, it proceeds to the next step. Otherwise, it returns $\ddot{v}=0$ and takes no further action. 

\item\label{sinlge-client:Checking-resulting-polynomial-valid-roots-} \underline{\textit{Checking the resulting polynomial's valid roots}}: checks if the resulting polynomial contains
 the root \rt in $\zeta$, by taking the following steps.

\begin{enumerate}

\item derives two keys from $tk$: 
$$ k'=\prf(1, tk), \hspace{4mm} s' = \prf(2, tk)$$

\item\label{sinlge-client:verification-case--removes-the-blinding-factors} removes the blinding factors from $\vec{g}=[g_{\st 1}, g_{\st 2}, g_{\st 3}]$ that were provided by server \srv in step \ref{single-client::publish-encrypted-solutions}. Specifically, for every $i$, 
$ 1\leq i \leq 3:$
\begin{equation*}
\begin{split}
\theta_{\st i}&=\underbrace{\big( \prf(i, s')\big)^{\st -1}}_{\st (w'_{\st i})^{\st -1}}\cdot\big(g_{\st i}- \overbrace{\prf(i, k'}^{\st z'_{\st i}})\big) \bmod \prm\\ &=
\gamma_{\st i}  \cdot \sum\limits_{\st j=1}^{\st z}   q_{\st j}\cdot \pi_{\st i, j} \bmod \prm
\end{split}
\end{equation*}

\item\label{step-single-client-check-roots}  interpolates a polynomial $\bm{\theta}$, given $(x_{\st 1}, \theta_{\st 1}), (x_{\st 2}, \theta_{\st 2}), (x_{\st 3}, \theta_{\st 3})$.  Note that $\bm{\theta}$ will have the form: 
\begin{equation*}
\begin{split}
\bm\theta(x)&=(x-\rt)\cdot \sum\limits_{\st j=1}^{\st z} q_{\st j}\cdot (x+m_{\st j}) \bmod \prm\\ &= \bm\psi(x)-\rt\cdot \sum\limits_{\st j=1}^{\st z}q_{\st j}\cdot m_{\st j}   \bmod \prm
\end{split}
\end{equation*}
where $\bm\psi(x)$ is a polynomial of degree $2$ whose constant term is $0$.

\item\label{single-client::eval-and-check-root} checks if $\rt$ is a root of $\bm \theta(x)$, by evaluating $\bm \theta(x)$ at $\rt$ and checking if the result is $0$, i.e., $\bm\theta(\rt)\stackrel{\st ?}=0$. It proceeds to the next step if the check passes.  It returns $\ddot{v}=0$ and takes no further action, otherwise.

\end{enumerate}

\item\label{step-sinlge-client-check-res-} \underline{\textit{Checking the final result}}: retrieves the result (which is the linear combination of $m_{\st 1},\ldots, m_{\st z}$)  from polynomial $\bm\theta(x)$'s constant term: $t=-\rt\cdot\sum\limits_{\st j=1}^{\st z}q_{\st j}\cdot m_{\st j}$ as follows:
\begin{equation*}
\begin{split}
res' &= -t\cdot \rt^{\st -1}\bmod \prm\\  &= \sum\limits_{\st j=1}^{\st z}q_{\st j}\cdot m_{\st j}
\end{split}
\end{equation*}

It checks $res' \stackrel{\st ?}=m$, where $m$ is the result that \srv sent to it, in step \ref{single-client-publish-lc-proof} of \ref{single-client::solving-linear-combination}. 

\item \underline{\textit{Accepting or rejecting the result}}: If all the checks pass, it accepts $m$ and returns $\ddot{v}=1$. Otherwise, it returns $\ddot{v}=0$.

\end{enumerate}


\item\label{sinlge-client::verifying-a-solution-of-single-puzzle}\hspace{-2mm}. when verifying the $j$-th solution of a single puzzle belonging to client $\prtt$: 

\begin{enumerate}


\item\underline{\textit{Checking the commitment' opening}}: checks whether opening $m_{\st j}\in m$ and $mk_{\st j}\in \zeta$, given by \srv in step \ref{single-client-pub-solution-} of \ref{single-client-solving-only-one-puzzle}, matches the commitment: 
$$\comver\big(com_{\st j}, (m_{\st j}, mk_{\st j})\big)\stackrel{\st?}=1$$



\item \underline{\textit{Accepting or rejecting the solution}}: accepts the solution $m_{\st j}$ and returns $\ddot{v}=1$, if the above check passes.  It rejects the solution, and it returns $\ddot{v}=0$ otherwise.

\end{enumerate}

\end{steps}

%

\end{enumerate}

\begin{remark}
In step \ref{step::single-client-set-fresh-based-when-j-not-1}, index $j$ is concatenated with $0$ to avoid any collision (i.e., generating the same pseudorandom value more than once), because $j$, as input of \prf, will be used as input in other steps. 
\end{remark}

\begin{theorem}\label{theo:security-of-MI-VH-TLP}
If the sequential modular squaring assumption holds, factoring $N$ is a hard problem, \prf, $\ole^{\st +}$,  and the commitment schemes are secure, then  \mhtlp presented above is secure, regarding Definition \ref{def:sec-def-vh-tlp}. 
\end{theorem}

\begin{theorem}\label{theo:completeness-efficiency-compactness-of-MI-VH-TLP}
The \mhtlp protocol presented above meets completeness, efficiency, and compactness, regarding Definitions \ref{def:completeness-vh-tlp}, \ref{def:efficiency-vh-tlp}, and \ref{def:compactness-vh-tlp} respectively.  
\end{theorem}

The remainder of this section provides proof for Thereoms \ref{theo:security-of-MI-VH-TLP} and \ref{theo:completeness-efficiency-compactness-of-MI-VH-TLP}.


\input{Single-client--VH-TLP--proof}

\input{correctness-proof}

\section{Multi-Instance Multi-Client Partially Homomorphic TLP}

In this section, we present Multi-Instance Multi-Client Partially Homomorphic TLP (\mmhtlp), a protocol that can be considered as a generalization of \mhtlp, presented in Section \ref{sec::Multi-Instance-TF}.  \mmhtlp is built upon \mhtlp and the \tf introduced in \cite{tempora-fusion}. It offers the features of both schemes within one unified protocol. \mmhtlp will (i) allow a client to generate multi-puzzles such that the server can solve them sequentially, (ii) enable the server to homomorphically compute a linear combination of puzzles of a single client, (iii) allow the server to homomorphically compute a linear combination of puzzles of different clients, and (iv) enable anyone to verify the correctness of each puzzle's solution and computations' outputs. 

\subsection{An Overview}

At a high level, the protocol works as follows. Initially,  \srv generates and publishes a set of public parameters, including vector $\vec{x}$ and a sufficiently large prime number $\prm$. Each client independently generates its secret and public keys. It publishes the public key. In the puzzle generation phase, each client, possessing a set of solutions, creates puzzles for these solutions using the chaining technique described in Section \ref{sec::Multi-Instance-TF}. Each client then publishes the puzzles along with some public parameters.

To enable \srv to learn a homomorphic linear combination of messages (encoded into the published puzzles) belonging to a single client, the client engages with \srv through the interactive algorithm $\eval_{\st sc}()$. Following the execution of this algorithm, the client publishes a set of public parameters, and \srv releases a puzzle encoding the computation result. 

To facilitate \srv learning a homomorphic linear combination of messages, where each message (encoded into a published puzzle) originates from a different client, the clients interact with \srv using algorithm $\eval_{\st mc}()$. Upon completing this algorithm, the clients publish a set of public parameters, and \srv publishes a puzzle encoding the computation result.  
After a certain period, \srv solves a puzzle, related to (i) the linear combination of a single client's solutions, (ii) the linear combination of multiple clients' solutions, or (iii) a single client's solution. \srv then publishes the solution and the corresponding proof. Given the public parameters and the solution, anyone can verify the proof.  Figure \ref{fig:MMH-TLP-workflow} illustrates the \mmhtlp workflow.

\vspace{-2mm}
\begin{figure}[h]
    \centering
    \includegraphics[width=6.6cm]{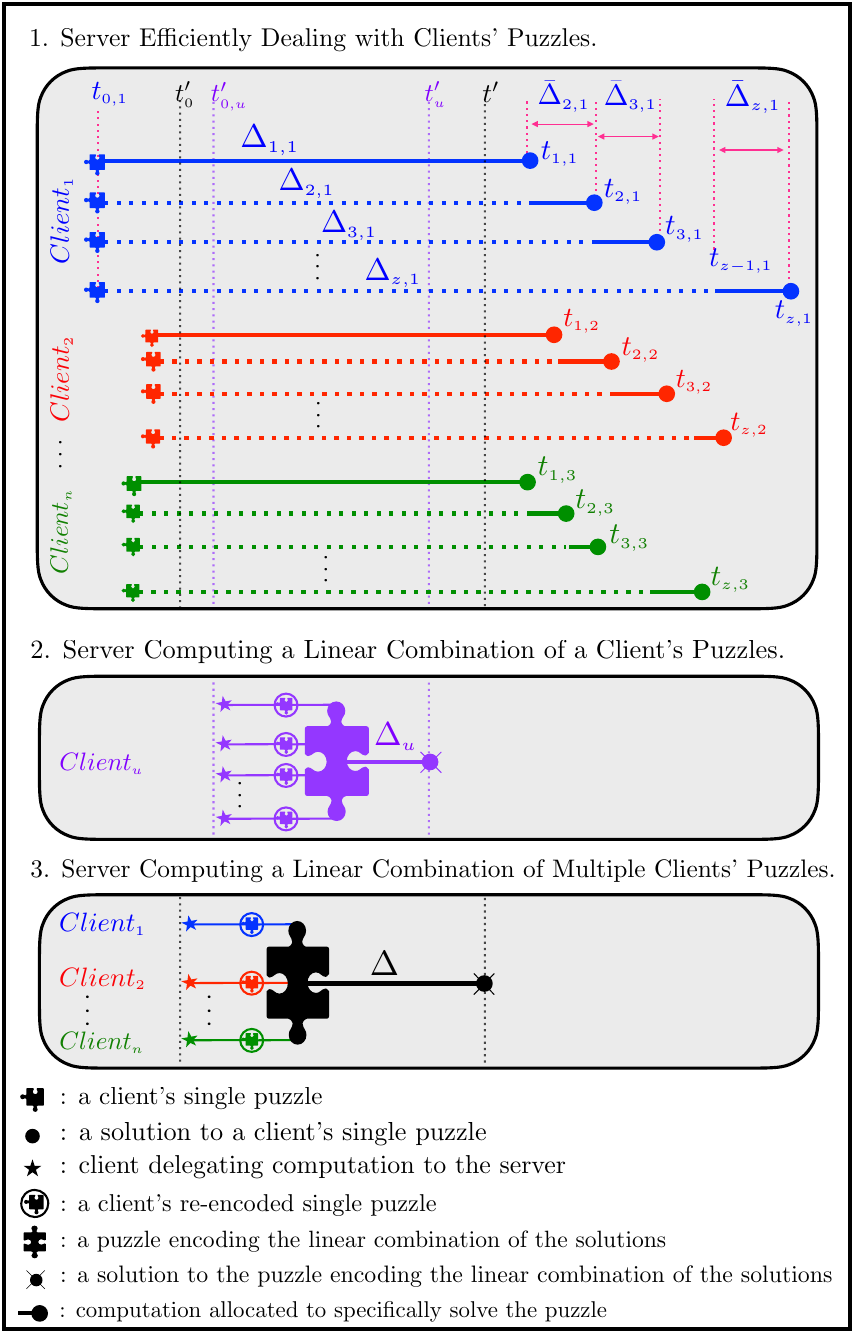}
    \caption{\mmhtlp Workflow Overview.}\label{fig:MMH-TLP-workflow}
\end{figure}

Regarding system design, \mmhtlp differs from \mhtlp (and \tf) in several ways. We briefly outline the differences.  

\begin{enumerate}


\item In \mmhtlp, the server needs to generate \cor $x$-coordinates (instead of generating only $3$ $x$-coordinates in \mhtlp) and each client should use these \cor $x$-coordinates, for the following reason. In \mmhtlp, each client's outsourced polynomial (that represents its puzzle) is of degree 1. During Phase \ref{multi-client-multi-instance-tf-LinearCombination-phase-x} (Linear Combination for Multiple Clients), this polynomial is multiplied by $\tl$ polynomials each representing a random root and is of degree $1$. The resulting polynomial will have degree $\tl+1$. Therefore, at least $\cor=\tl+2$ $(y,x)$-coordinate pairs are needed to interpolate a polynomial encoding a linear combination of the solutions. 

\item\label{diff::multi-vs-single--two-linear-comb} There will be two different algorithms for the linear combination, one algorithm, called $\eval_{\st sc}$, to perform a linear combination of a single client's puzzles, and another one, called $\eval_{\st mc}$, to perform a linear combination of  $n$ different clients' puzzles.

\item During the linear combination of different clients' puzzles,  $\eval_{\st mc}()$ takes a new input $\id_{\st \prt}$ for each client $\prtt_{\st \prt}$. 
This approach allows each $\prtt_{\st \prt}$ to specify which one of its outsourced puzzles must be used as an input to the computation.

\item\label{diff::multi-vs-single--new-inputs-three-cases} Algorithms $\solv()$ and $\ver()$ take a new input string $\hat{\cmd}\in\{\singclient, \mulclient\}$ that specifies whether the encrypted linear combination $\vec{g}$ is the result of a homomorphic linear combination of messages belonging to a single client or multiple clients. Moreover, now, these two algorithms consider three different cases (rather than two); namely, (i) when solving or verifying a puzzle related to the linear combination of messages where all messages belong to the same client, (ii) when solving or verifying a puzzle related to the linear combination of messages, where each message belong to a different client, and (iii) when solving or verifying each client's single puzzle.

\end{enumerate}

 \mmhtlp and \tf differ from  aspects \ref{diff::multi-vs-single--two-linear-comb}--\ref{diff::multi-vs-single--new-inputs-three-cases}  as well. 

%

\subsection{Detailed Construction}

\input{multi-client-multi-instance}

%% file: Single-client--VH-TLP--proof.tex

\subsection{Proof of Theorem \ref{theo:security-of-MI-VH-TLP}}\label{sec::Single-client-VH-TLP-proof}

\begin{proof}[sketch]
There will be a significant overlap between the proofs of Theorems \ref{theo:security-of-MI-VH-TLP} and \tf (i.e., Theorem \ref{theo:security-of-VH-TLP} in Appendix \ref{sec::tf-protocol}). 
The proof of Theorem \ref{theo:security-of-MI-VH-TLP} differs from that of Theorem \ref{theo:security-of-VH-TLP} from a key perspective. Namely, the former requires an additional discussion on the privacy of each base $r_{\st j'}$ before the $(j'-1)$-th puzzle is solved, for every $j'$, where $2\leq j' \leq z$.  

Briefly, the additional discussion will rely on the security of standard RSA-based TLP, the hiding property of the commitment,  and the security of \prf.

Specifically, before the $j$-th puzzle is solved the related master key $mk_{\st j}$ cannot be extracted except for a probability negligible in the security parameter, $\mu(\lambda)$, if the sequential modular squaring assumption holds, factoring $N$ is a hard problem, and the commitment scheme satisfies the hiding property. This argument holds for any $j$, where $1\leq j \leq z$. 
Therefore, with a probably at most $\mu(\lambda)$ an adversary can find a key $mk_{\st j'-1}$ of \prf to compute the base $r_{\st j'}$, which has been set as $r_{\st j'}=\prf(j' || 0, mk_{\st j'-1})$, for any $j'$, where $2\leq j' \leq z$. Furthermore, since $r_{\st j'}$ is the output of \prf, due to the security of \prf (that its output is indistinguishable from the output of a random function), the probability of correctly computing it is  $\mu(\lambda)$, without the knowledge of $mk_{\st j'-1}$. 
\hfill$\square$
\end{proof}

%% file: correctness-proof.tex

\subsection{Proof of Theorem \ref{theo:completeness-efficiency-compactness-of-MI-VH-TLP}} \label{theo:completeness-efficiency-compactness}

\begin{proof}
We begin with proving the completeness of \mhtlp. 

\begin{lemma}\label{lemma::completeness}
 \mhtlp satisfies completeness, regarding Definition \ref{def:completeness-vh-tlp}.
\end{lemma}

We start by addressing Case \ref{completeness-of-single-puzzle}, which involves solving a single puzzle. Due to the correctness and deterministic nature of the original TLP \cite{Rivest:1996:TPT:888615}, particularly through repeated squaring, a server can consistently derive the master key  $mk_{\st j}$ through a fixed number of repeated squaring.  Furthermore, the correctness and deterministic property of \prf ensures that the server can extract $r_{\st j}$ and subsequently compute keys $k_{\st j}, s_{\st j}$ and blinding factors $z_{\st i, j}$ and  $w_{\st i, j}$. Given these blinding factors, the server can remove them from each $o_{\st i,j}$, resulting in a set of $y$-coordinates. Due to the correctness (especially deterministic nature) of interpolation algorithms, such as Lagrange interpolation, the server can recover the identical polynomial that the client initially constructed, i.e., polynomial $ \bm{\pi}_{\st j}(x) = x + m_{\st j} \bmod \prm$. Given, each polynomial $ \bm{\pi}_{\st j}(x)$, the server can easily recover its constant term $m_{\st j}$, which is the solution to the related puzzle. 

Case \ref{completeness-of-linear-comb-puzzle} focuses on the correctness of solving a puzzle related to the linear combination. Similar to the above case, due to the correctness and deterministic nature of repeated squaring and \prf, the server can find the temporary key $tk$ and accordingly discover keys $k'$ and $s'$. These keys allow the server to remove the blinding factors from each masked $y$-coordinate $g_{\st i}$, yielding three $y$-coordinates $\theta_{\st 1},\theta_{\st 2},$ and $\theta_{\st 3}$. Because the interpolation algorithm is deterministic, the server will recover a polynomial of the form $\bm\theta(x) =(x-\rt)\cdot \sum\limits_{\st j=1}^{\st z}q_{\st j}\cdot (x+m_{\st j}) \bmod \prm$. The main reasons that polynomial $\bm\theta(x)$ maintains this form are (1) the correctness of $\ole^{\st +}$, (2) the properties of polynomial arithmetic, as explained in Section \ref{sec:Polynomial-Representation-of-Message}, and (3) the product of multiple polynomials preserves each individual polynomial's roots. Given $\bm\theta(x)$, one can easily retrieve its constant term and multiply it by  $(-\rt)^{\st -1}$ which yields the linear combination of solutions: $\sum\limits_{\st j=1}^{\st z} q_{\st j}\cdot m_{\st j}$. 

We proceed to Case \ref{correctness-verify-single-puzzle}, which pertains to the correctness of verification for a single puzzle. Building on the argument presented in Case \ref{completeness-of-single-puzzle}, the server can always retrieve the master key $mk_{\st j}$ and the related solution $m_{\st j}$. Assuming the commitment verification algorithm is correct, an honest server's proof  $(mk_{\st j}, m_{\st j})$ is always accepted by an honest verifier. 

Case \ref{completeness-verification-for-HLC} considers the correctness of verification for the linear combination. As discussed in Case \ref{completeness-of-linear-comb-puzzle}, an honest server can always find the temporary key $tk$ and the keys derived from it $k'$ and $s'$.  Given these keys, it can find the root $\rt$. Due to the correctness of the commitment's verification algorithm, proof  $(\rt, tk)$ is always accepted by an honest verifier. Moreover, due to the correctness and deterministic nature of \prf, a verifier can derive from $tk$ the same keys $(k',s')$ as the client used to blind its $y$-coordinates. These keys allow the verifier to unblind the $y$-coordinates to obtain $\theta_{\st 1}, \theta_{\st 2}, $ and $\theta_{\st 3}$. Because the interpolation algorithm is deterministic, the verifier will recover a polynomial of the form $\bm\theta(x) =(x-\rt)\cdot \sum\limits_{\st j=1}^{\st z}q_{\st j}\cdot (x+m_{\st j}) \bmod \prm$. Also, evaluating polynomial $\bm\theta(x)$ at $\rt$ will always result in $0$, because $\rt$ is a root of $\bm\theta(x)$. 
Given $\rt$ and the constant term $t$ of $\bm\theta(x)$, the verifier can always extract the linear combination of messages $-t\cdot \rt^{\st -1}\bmod \prm= \sum\limits_{\st j=1}^{\st z}q_{\st j}\cdot m_{\st j}$, which will be equal to the result $\sum\limits_{\st j=1}^{\st z} q_{\st j}\cdot m_{\st j}$ that the prover sends.  
\hfill$\blacksquare$

\begin{lemma}\label{lemma::efficiency}
 \mhtlp satisfies efficiency, regarding Definition \ref{def:efficiency-vh-tlp}.
\end{lemma}

We will initially focus on Condition \ref{efficiency-multi-instance}: multi-instance. The time complexity of (traditional and multi-instance) TLPs is dominated by the number of modular squaring. When a traditional TLP encounters $z$ instances of a puzzle at once, it must deal with each puzzle individually. Therefore, in this setting, its time complexity is $\Psi_{\st trad}(z)= \mxsqr\cdot \sum\limits_{\st j=1}^{\st z}\Delta_{\st j}$. However, in \mhtlp, the puzzles are solved sequentially, resulting in the time complexity of $\Psi_{\st multi}(z)= \mxsqr\cdot\big(\Delta_{\st 1} + \sum\limits_{\st j=2}^{\st z}(\Delta_{\st j}-\Delta_{\st j-1})\big)$.  Therefore, the 
difference between the time complexity of the TLPs in these two settings is:  
\begin{equation*}
\begin{split}
\Psi_{\st trad}(z)-\Psi_{\st multi}(z)&=
 \mxsqr\cdot\big((\sum\limits_{\st j=1}^{\st z}\Delta_{\st j})-(\Delta_{\st 1} + \sum\limits_{\st j=2}^{\st z}(\Delta_{\st j}-\Delta_{\st j-1})\big)\\
&=\mxsqr\cdot\sum\limits_{\st j=1}^{\st z-1}\Delta_{\st j}
\end{split}
\end{equation*}

Thus, $\Psi_{\st trad}(z)-\Psi_{\st multi}(z)= {poly}(z, \mxsqr, \Delta_{\st 1},\ldots, \Delta_{\st z})$, for a fixed polynomial $poly$, meeting the criteria set out in Condition \ref{efficiency-multi-instance}. 

We proceed to Condition \ref{efficiency-Polynomial-time-solving}: polynomial-time solving. The core primitive upon which \mhtlp and in particular the algorithm $\solv()$ relies to solve a puzzle is the standard sequential squaring. The complexity of $\solv()$ is 
$\mxsqr\cdot\big(\Delta_{\st 1} + \sum\limits_{\st j=2}^{\st z}(\Delta_{\st j}-\Delta_{\st j-1})\big)$, which itself can be represented as $\hat{poly}(z, T_{\st max}, \log(N))$, where $\hat{poly}$ is a fixed polynomial and $T_{\st max}=\mxsqr\cdot \Delta_{\st z}$.

We move on to Condition \ref{efficiency-requirement-genpuz}: faster puzzle generation property. The main operation in the algorithm $\pgen()$  that generates puzzles is generating each value $mk_{\st j} = r_{\st j}^{\st a_{\st j}}\bmod N$, where $a_{\st j}=2^{\st T_{\st j}}\bmod \phi(N)$. The complexity of generating each $a_{\st j}$ is $O(\log_{\st 2}(T)\cdot \log_{\st 2}^{\st 2}(\phi(N)))\approx O(\log_{\st 2}(T)\cdot \log_{\st 2}^{\st 2}(N))$, while the complexity of generating each $mk_{\st j}$ is $O(log_{\st 2}(a_{\st j})\cdot \log_{\st 2}^{\st 2}(N))\approx O(\log_{\st 2}(N)\cdot log_{\st 2}^{\st 2}(N))$. Thus, the total complexity is $O((\log_{\st 2}(T)+\log_{\st 2}(N))\cdot z\cdot \log_{\st 2}^{\st 2}(N))$, which can be represented as $poly'(z, \log(T), \log(N))$, where $poly'$ is a fixed polynomial and $T$ is the maximum time paramter.

Next, we turn our attention to Condition \ref{efficiency-requirement-Faster-puzzles-evaluation}: faster puzzle evaluation. The $\eval()$ algorithm involves generating a temporary key $tk=h^{\st b}\bmod (N)$, where $b= 2^{\st Y}\bmod \phi(N)$. Based on the above analysis, these operations' total complexity is $O((\log_{\st 2}(T)+\log_{\st 2}(N))\cdot \log_{\st 2}^{\st 2}(N))$, where $T=\Delta\cdot \mxsqr$. The $\eval()$ algorithm also involves operations to compute a linear combination of solutions $m_{\st 1},\ldots, m_{\st z}$ using the coefficients $q_{\st 1}, \ldots, q_{\st z}$ to realize the functionality $\func$. The operation to complete the linear combination is linear with the total number of puzzles. Hence, the complexity of  $\eval()$ can be represented as $poly''(\log(T), \log(N), \func \big((q_{\st 1}, m_{\st 1}),$ $ \ldots, (q_{\st z}, m_{\st z}))$, for a fixed polynomial $poly''$. 
\hfill$\blacksquare$

\begin{lemma}\label{lemma::compactness}
 \mhtlp satisfies compactness, regarding Definition \ref{def:compactness-vh-tlp}. 
\end{lemma}
The algorithm $\eval()$ outputs a vector of three elements as a puzzle $\vec{g}=[g_{\st 1}, g_{\st 2},  g_{\st 3}]$ along with a small set of public parameters $\vec{pp}^{\st(\text{Evl})}$. The bit size of each element  $g_{\st i}$ of the puzzle vector $\vec{g}$ is $\log_{\st 2}(\prm)$. Thus, the bit size of $\vec{g}$ can be represented as $||\vec{g}||= poly\Big(\log(\prm), ||\func \big((q_{\st 1}, m_{\st 1}),\ldots,(q_{\st z}, m_{\st z}) \big)||\Big)$, for a fixed polynomial $poly$. 
\hfill$\blacksquare$

This concludes the proof of Theorem \ref{theo:completeness-efficiency-compactness-of-MI-VH-TLP} as we have proved the completeness (Lemma \ref{lemma::completeness}), efficiency (Lemma \ref{lemma::efficiency}), and compactness (Lemma \ref{lemma::compactness}). 
\hfill$\square$
\end{proof}

%% file: multi-client-multi-instance.tex

In this section, we present a detailed description of  \mmhtlp.

\begin{enumerate}

\item\label{multi-instance-multi-client-::Setup-server} \underline{\textbf{Setup}}. $\ssetup(1^{\st \lambda}, \tl, t)\rightarrow (., pk_{\st \srv})$

 The server \srv (or any party) only once takes the following steps:

\begin{enumerate}

\item \underline{\textit{Setting a field's parameter}}: generates a sufficiently large prime number $\prm$, where $\log_{\st 2}(\prm)$ is a security parameter, e.g.,  $\log_{\st 2}(\prm)\geq 128$.

\item  \underline{\textit{Generating public $x$-coordinates}}: let $\tl$ be the total number of leader clients. It sets $\cor=\tl+2$ and $\vec{x}=[x_{\st 1}, \ldots, x_{\st \cor}]$, where $x_{\st i}\neq x_{\st j}$, $x_{\st i}\neq 0$, and $x_{\st i}\notin U$.

\item \underline{\textit{Publishing public parameters}}: publishes $pk_{\st \srv}=(\prm, \vec{x}, t)$. 

\end{enumerate}

\item\label{multi-instance-multi-client-::key-gen} \underline{\textbf{Key Generation}}. $\csetup(1^{\st \lambda})\rightarrow \kp_{\st \prt}$

Each party $\prtt_{\st\prt}$ in $C=\set$ takes the following steps:

\begin{enumerate}


\item \underline{\textit{Generating RSA public and private keys}}: computes $N_{\st \prt}=\prm_{\st 1}\cdot \prm_{\st 2}$, where $\prm_{\st i}$  is a large randomly chosen prime number, where $\log_{\st 2}(\prm_{\st i})$ is a security parameter, e.g., $\log_{\st 2}(\prm_{\st i})\geq 2048$. Next,  it computes Euler's totient function of $N_{\st \prt}$, as: $\phi(N_{\st \prt})=( \prm_{\st 1}-1)\cdot( \prm_{\st 2}-1)$.



\item\label{multi-client::Publishing-public-parameters} \underline{\textit{Publishing public parameters}}: locally keeps secret key $sk_{_{\st\prt}}=\pnp$ and publishes public key $pk_{\st\prt}=N_{\st \prt}$. 
\end{enumerate}

\item\label{multi-instance-multi-client-tf-Puzzle-Generation-phase} \underline{\textbf{Puzzle Generation}}. $\pgen(\vec{m}_{\st \prt}, \kp_{\st \prt}, pk_{\st \srv}, \vec\Delta_{\st \prt}, \mxsqr)\rightarrow(\vec{o}_{\st \prt}, prm_{\st \prt})$

Each client $\prtt_{\st\prt}$ takes the following steps to generate $z$ puzzles for messages $\vec{m}_{\st \prt}=[m_{\st 1, \prt},\ldots, m_{\st z, \prt}]$ and wants \srv to learn each message $m_{\st j, \prt}$ at time $time_{\st j, \prt}\in \vec{time}_{\st \prt}$, where $ \vec{time}_{\st \prt}=[time_{\st 1, \prt},\ldots, time_{\st z, \prt}]$, $\bar\Delta_{\st j, \prt}=time_{\st j, \prt}-time_{\st j-1, \prt}$, $ \vec{\Delta}_{\st \prt}=[\bar\Delta_{\st 1, \prt},\ldots, \bar\Delta_{\st z, \prt}]$, $1\leq j \leq z$, and $1\leq \prt\leq n$.

\begin{enumerate}

\item \underline{\textit{Checking public parameters}}: checks the bit-size of $p$ and  elements of $\vec{x}$ in $pk_{\st \srv}$, to ensure $\log_{\st 2}(p)\geq 128$, $x_{\st i}\neq x_{\st j}, x_{\st i}\neq 0$, and $x_{\st i}\notin U$. If it does not accept the parameters, it returns $(\bot, \bot)$ and does not take further action.

\item\label{multi-client:Generating-secret-keys} \underline{\textit{Generating secret keys}}: generates a vector of master keys $\vec{mk}_{\st \prt}=[mk_{\st \prt, 1}, \ldots, mk_{\st \prt, z}]$ and two secret keys $k_{\st \prt, j}$ and  $s_{\st \prt, j}$ for each master key $mk_{\st \prt, j}$ in $\vec{mk}_{\st \prt}$ as follows. It constructs an empty vector $\vec{mk}_{\st \prt}$. Then, it

\begin{enumerate}

\item\label{multi-client-multi-instance-GC-TLP::compute-a-values} sets each exponent $a_{\st  j, \prt}$. 
  $$\forall j, 1\leq j\leq z: \quad a_{\st j, \prt}=2^{\st T_{\st j}}\bmod \phi(N_{\st \prt})$$
  
  where $T_{\st j}=\mxsqr\cdot\bar\Delta_{\st j, \prt}$ is the total number of squaring needed to decrypt an encrypted solution $m_{\st j, \prt}$ after the previous solution $m_{\st j-1, \prt}$ is revealed. 
  %

\item\label{multi-client-multi-instance-GC-TLP::compute-r-values} computes each master key $mk_{\st j, \prt}$ as follows.  For every $j$, where 
$1\leq j\leq z:$

\begin{itemize}[label=$\bullet$]

\item when $j=1:$

\begin{enumerate}
 \item picks a uniformly random base $r_{\st j}\stackrel{\st\$}\leftarrow \mathbb{Z}_{\st N_{\st \prt}}$.
 \item sets  key $mk_{\st j, \prt}$ as $mk_{\st j, \prt}= r^{\st a_{{\st j}}}_{\st j}\bmod N_{\st \prt}$. 
 \item  appends $mk_{\st j, \prt}$ to $\vec{mk}_{\st \prt}$. 
 \end{enumerate}

\item when $j>1:$

\begin{enumerate}

 \item\label{step::single-client-set-fresh-based-when-j-not-1} derives a fresh base $r_{\st j} $ from the previous master key  as $r_{\st j}=\prf(j||0, mk_{\st  j-1, \prt})$.
  
 \item sets  key $mk_{\st  j, \prt}$ as $mk_{\st  j, \prt}= r^{\st a_{{\st j}}}_{\st j}\bmod N_{\st \prt}$.
 \item appends $mk_{\st j, \prt}$ to $\vec{mk}_{\st \prt}$.
 
 \end{enumerate}
 
\end{itemize}

\item\label{multi-client-multi-instance:derive-keys-for-puzzle} derives two secret keys $k_{\st j, \prt}$ and  $s_{\st j, \prt}$ from each $mk_{\st j, \prt}$.  
$$\forall j, 1\leq j\leq z: \quad k_{\st j, \prt}=\prf(1, mk_{\st j, \prt}), \hspace{4mm} s_{\st j, \prt}=\prf(2, mk_{\st j, \prt})$$

\end{enumerate}


%

\item\label{multi-client-multi-instance:derive-PR-values}\underline{\textit{Generating blinding factors}}:  generates $2\cdot \cor$ pseudorandom values for each $j$, by using $k_{\st j, \prt}$ and  $s_{\st j, \prt}$. 
$$\forall j, 1\leq j\leq z\quad \text{and} \quad \forall  i,  1\leq i \leq \cor: \hspace{4mm}  z_{\st i, j, \prt}=\prf(i, k_{\st j, \prt}), \hspace{4mm} w_{\st i, j, \prt}=\prf(i, s_{\st j, \prt})$$

\item\underline{\textit{Encoding plaintext messages}}:  

\begin{enumerate}
\item\label{multi-client-multi-instance:rep-message-as-poly} represents each plaintext solution $m_{\st j, \prt}$ as a polynomial, such that the polynomial's constant term is the message. 
$$\forall j, 1\leq j\leq z:\quad  \bm{\pi}_{\st j, \prt}(x) = x + m_{\st j, \prt} \bmod \prm$$

\item\label{multi-client-multi-instance:gen-y-coord-of-poly} computes \cor $y$-coordinates of  each $\bm{\pi}_{\st j, \prt}(x)$: 
$$\forall i, 1\leq j \leq z  \hspace{4mm} \text{and} \hspace{4mm} \forall i, 1\leq i \leq \cor: \hspace{4mm}\pi_{\st i, j, \prt} = \bm{\pi}_{\st j, \prt}(x_{\st i}) \bmod \prm$$

\end{enumerate}

\item\label{multi-client-multi-instance:enc-y-coord-of-poly} \underline{\textit{Encrypting the messages}}:  encrypts the $y$-coordinates using the blinding factors as follows: 
$$\forall j, 1\leq j\leq z  \quad \text{and} \quad \forall i, 1\leq i \leq \cor: \hspace{4mm} o_{\st i, j, \prt} = w_{\st i, j, \prt}\cdot(\pi_{\st i, j, \prt} +  z_{\st i, j, \prt}) \bmod \prm$$

\item\label{multi-client:commit-y-coord-of-poly}\underline{\textit{Committing to the message}}:  commits to the plaintext messages:   
$$ com_{\st j,\prt} = \comcom(m_{\st j, \prt}, mk_{\st  j, \prt})$$

Let $\vec{com}_{\st \prt}=[com_{\st 1, \prt},\ldots, com_{\st z, \prt}]$.

\item\label{tf-Managing-messages}\underline{\textit{Managing messages}}: publishes  $\vec{o}_{\st \prt}=\Big[ [o_{\st 1, 1, \prt}, \ldots, o_{\st \cor, 1, \prt}], \ldots, [o_{\st 1, z, \prt}, \ldots, o_{\st \cor, z, \prt}]\Big]$ and $pp_{\st \prt}=(\vec{com}_{\st \prt}, r_{\st 1}) $. It locally keeps secret parameters $sp_{\st \prt}=\vec{mk}_{\st \prt}$. 
It sets $prm_{\st \prt}=(sp_{\st \prt}, pp_{\st \prt})$. It deletes everything else, including each $m_{\st j, \prt}$ and $\bm{\pi}_{\st j, \prt}(x)$.

\end{enumerate}

\item\label{multi-client-multi-instance-phase::multi-instance-TF-Linear-Combination} {\textbf{Linear Combination for a Single Client}}. $\eval_{\st sc}(\langle \srv(\vec{o}_{\st \prt}, \Delta_{\st \prt},  \mxsqr, {pp}_{\st \prt}, {pk}_{\st \prt}, pk_{\st \srv}), \prtt_{\st \prt}(\Delta_{\st \prt}, \mxsqr, \kp_{\st \prt}, $ $ prm_{\st \prt}, q_{\st 1, \prt}, pk_{\st \srv}), \ldots, \prtt_{\st \prt}(\Delta_{\st \prt}, $ $\mxsqr, $ $\kp_{\st \prt}, prm_{\st \prt}, q_{\st z, \prt},$  $pk_{\st \srv}) \rangle)\rightarrow(\vec{g}_{\st \prt}, {pp}^{\st(\text{Evl})}_{\st \prt})$

In this phase, a client $\prtt_{\st \prt}$ produces certain messages that allow \srv to find a linear combination of its plaintext solutions after time $\Delta_{\st \prt}$.

\begin{enumerate}

\item\underline{\textit{Granting the computation}}: client $\prtt_{\st \prt}$ takes the following steps.

\begin{enumerate}

\item\label{multi-client-multi-instance:Generating-temporary-secret-keys} \underline{\textit{Generating temporary secret keys}}: generates a temporary master key $tk$ and two secret keys $k'$ and $s'$.  It also computes $z-1$ secret key $[f_{\st j},\ldots,f_{\st z}]$.  To generate them, it takes the following steps. It computes an exponent: 
$$b=2^{\st Y}\bmod \phi(N_{\st \prt})$$

where $Y=\Delta_{\st \prt}\cdot \mxsqr$ and $\Delta_{\st \prt}$ is the period after which the solution must be discovered. It selects a  base uniformly at random: $h\stackrel{\st\$}\leftarrow\mathbb{Z}_{\st N_{\st \prt}}$ and then sets a temporary master key $tk$:
$$tk= h^{\st b} \bmod N_{\st \prt}$$

It derives two keys from $tk$: 
$$ k'=\prf(1, tk), \hspace{4mm} s'=\prf(2, tk)$$

It picks fresh $z-1$ random keys $\vec{f}=[f_{\st 2},\ldots, f_{\st z}]$, where $f_{\st j}\stackrel{\st \$}\leftarrow \{0, 1\}^{\st poly(\lambda)}$. 

\item\label{multi-client-multi-instance:Generating-temporary-pr-vals} \underline{\textit{Generating blinding factors}}: regenerates its original blinding factors, for each $j$-th puzzle. Specifically, for every $j$,  derives two secret keys $k_{\st j, \prt}$ and  $s_{\st j, \prt}$ from $mk_{\st j, \prt}\in \kp_{\st \prt}$ as follow.  
$$\forall j, 1\leq j\leq z:  \quad k_{\st j, \prt}=\prf(1, mk_{\st j, \prt}), \hspace{4mm} s_{\st j, \prt}=\prf(2, mk_{\st j, \prt})$$

%

It regenerates $z\cdot \cor$ pseudorandom values, by using $k_{\st j, \prt}$ and  $s_{\st j, \prt}$. 
$$\forall j, 1\leq j\leq z  \quad \text{and}\quad \forall  i,  1\leq i \leq \cor: \hspace{4mm}  z_{\st i, j, \prt}=\prf(i, k_{\st j, \prt}), \hspace{4mm} w_{\st i, j, \prt}=\prf(i, s_{\st j, \prt})$$

It also generates new $2\cdot \cor$ pseudorandom values using keys $(k', s')$. 
 $$\forall i, 1\leq i \leq \cor: \hspace{4mm}   z'_{\st i}=\prf(i, k'), \hspace{4mm}   w'_{\st i}=\prf(i, s')$$

It computes new sets of (zero-sum) blinding factors, using each key $f_{\st j}\in\vec{f}$, as follows. $\forall j, 1\leq j \leq z:$

\begin{itemize}[label=$\bullet$]

\item if $j=1$: 
$$\forall i, 1\leq i \leq \cor: \hspace{4mm}y_{\st i, j}=  -\sum\limits_{\st j = 2}^{\st z}\prf(i, f_{\st j})  \bmod \prm$$

\item if $j>1$: 
$$\forall i, 1\leq i \leq \cor: \hspace{4mm}y_{\st i, j}= \prf(i, f_{\st j})         \bmod \prm$$
\end{itemize}

\item\label{multi-client-multi-instance:Generating-temporary-root} \underline{\textit{Generating $y$-coordinates of a random root}}: picks a random root, $\rt\stackrel{\st\$}\leftarrow\mathbb{F}_{\st \prm}$.  It represents $\rt$ as a polynomial $\bm{\gamma}(x)$, where $\rt$ is the polynomial's root, as setting $\bm{\gamma}(x) =x- \rt \bmod \prm$. It generates \cor $y$-coordinates of  $\bm{\gamma}(x)$: 
$$\forall i, 1\leq i \leq \cor: \hspace{4mm}\gamma_{\st i} = \bm{\gamma}(x_{\st i}) \bmod \prm$$

\item\label{multi-client-multi-instance:Committing-to-the-root}  \underline{\textit{Committing to the root}}: computes $ com'=\comcom(\rt, tk)$.

\item\label{multi-client-multi-instance::eval:ole-detect} \underline{\textit{Re-encoding outsourced puzzle}}: participates in an instance of $\ole^{\st +}$ with \srv, for every $j$-th puzzle and  every $i$, where $1\leq j \leq z$ and $1\leq i \leq \cor$.  The inputs of client $\prtt_{\st \prt}$ to the $i$-th instance of $\ole^{\st +}$ are: 
$$e_{\st i, j} =\gamma_{\st i}\cdot q_{\st j, \prt}\cdot w'_{\st i}\cdot (w_{\st i, j, \prt})^{\st -1}\bmod \prm,\quad\quad e'_{\st i, j} = -(\gamma_{\st i}\cdot q_{\st j, \prt}\cdot w'_{\st i}\cdot  z_{\st i,j, \prt})+z'_{\st i} + y_{\st i, j}\bmod \prm$$ 

The input of \srv to the $(i, j)$-th instance of $\ole^{\st +}$ is the corresponding encrypted $y$-coordinate: $e''_{\st i, j} = o_{\st i, j, \prt}$. Accordingly, the $(i, j)$-th instance of $\ole^{\st +}$ returns to \srv:
\begin{equation*}
\begin{split}
d_{\st i, j} &= e_{\st i, j}\cdot e''_{\st i, j} + e'_{\st i, j}\\ &=\gamma_{\st i}\cdot q_{\st j, \prt}\cdot w'_{\st i}\cdot  \pi_{\st i, j, \prt}+z'_{\st i } + y_{\st i, j}\bmod \prm
\end{split}
\end{equation*}

where $q_{\st j, \prt}$ is a coefficient for $j$-th solution $m_{\st j, \prt}$. If client $\prtt_{\st\prt}$ detects misbehavior during the execution of  $\ole^{\st +}$, it outputs a special symbol $\bot$ and halts.

\item \label{sinlge-client::Publishing-public-parameters--}\underline{\textit{Publishing public parameters}}: publishes  ${pp}^{\st(\text{Evl})}_{\st \prt}=(h, com')$. 


\end{enumerate}

\item\label{multi-client-multi-instance::Computing-encrypted-linear-combination}\underline{\textit{Computing encrypted linear combination}}:  Server \srv sums all of the outputs of $\ole^{\st +}$ instances that it has invoked. $\forall i, 1\leq i \leq \cor: $
\begin{equation*}
\begin{split}
g_{\st i} &= \sum\limits_{\st j=1}^{\st z} d_{\st {i, j}} \bmod \prm\\ & =    (w'_{\st i}\cdot \gamma_{\st i}  
\cdot \sum\limits_{\st j=1}^{\st z}   q_{\st j, \prt}\cdot \pi_{\st i, j, \prt} )+    z'_{\st i}  \bmod  \prm
\end{split}
\end{equation*}


\item\label{single-client-publish-encrypted-solutions}\underline{\textit{Disseminating encrypted result}}:  server \srv publishes $\vec{g}=[g_{\st 1}, \ldots, g_{\st \cor}]$.

\end{enumerate}

\item\label{multi-client-multi-instance-tf-LinearCombination-phase-x} \underline{\textbf{Linear Combination for Multiple Clients}}. $\eval_{\st mc}(\langle \srv(\vec{o}, \Delta, $ $ \mxsqr, \vec{pp}, \vec{pk}, pk_{\st \srv}), \prtt_{\st 1}(\Delta, \mxsqr, \kp_{\st {\st 1}},$ $ prm_{\st {\st 1}}, q_{\st 1}, pk_{\st \srv}, \id_{\st 1}), \ldots, \prtt_{\st n}(\Delta, $ $\mxsqr, $ $\kp_{\st {\st n}}, prm_{\st {\st n}}, q_{\st n},$  $pk_{\st \srv},\id_{\st n}) \rangle)\rightarrow(\vec{g}, \vec{pp}^{\st(\text{Evl})})$

In this phase, the clients interact with  \srv to compute certain messages that enable \srv to find a linear combination of the clients' plaintext messages after time $\Delta$.

\begin{enumerate}


\item\label{multi-client:Randomly-selecting-leaders} \underline{\textit{Randomly selecting leaders}}: all parties in $C=\set $ agree on a random key $\hat{r}$, e.g., through a coin tossing protocol. Each  $\prtt_{\st\prt}$ deterministically identifies indices of $\tl$ leader clients: $\forall j, 1\leq j\leq \tl: idx_{\st j}=\g(j||\hat{r})$. Let \idx be a vector containing these $\tl$ clients.


\item\label{multi-client::Granting-the-computation}\underline{\textit{Granting the computation by each leader client}}:  each leader client $\prtt_{\st\prt}$  in \idx takes the following steps. 



\begin{enumerate}

\item\label{multi-client:Generating-temporary-secret-keys} \underline{\textit{Generating temporary secret keys}}: generates a temporary master key $tk_{\st\prt}$ and two secret keys $k'_{\st \prt}$ and $s'_{\st \prt}$ for itself. Moreover, it generates a secret key $f_{\st {\st l}}$ for each client. To do that, it takes the following steps. It computes an exponent: 
$$b_{\st \prt}=2^{\st Y}\bmod \pnp$$

where $Y=\Delta\cdot \mxsqr$, $\Delta$ is the period after which the solution representing the linear combination of the messages must be discovered, and $N_{\st \prt} \in \vec{pk}=[N_{\st 1},\ldots, N_{\st n}]$.
 It selects a  base uniformly at random: $h_{\st \prt}\stackrel{\st\$}\leftarrow\mathbb{Z}_{\st N_{\st \subprt}}$ and then sets a temporary master key $tk_{\st\prt}$: 
$$tk_{\st\prt} = h_{\st \prt}^{\st b_{\st\subprt}} \bmod N_{\st \prt}$$

It derives two keys from $tk_{\st \prt}$: 
$$ k'_{\st \prt}=\prf(1, tk_{\st \prt}), \hspace{4mm} s'_{\st \prt}=\prf(2, tk_{\st \prt})$$

It picks a random key $f_{\st {\st l}}$ for each client $\cl_{\st l}$ excluding itself, i.e.,  $f_{\st {\st l}}\stackrel{\st\$}\leftarrow\{0, 1\}^{\st poly(\lambda)}$, where $\cl_{\st l}\in C\setminus\prtt_{\st\prt}$. It sends $f_{\st {\st l}}$ to each $\cl_{\st l}$.

\item\label{multi-client:Generating-temporary-pr-vals} \underline{\textit{Generating temporary blinding factors}}: derives \cor pseudorandom values from $s'_{\st \prt}$: 
 $$\forall i, 1\leq i \leq \cor:  \hspace{4mm}  w'_{\st i,\prt}=\prf(i, s'_{\st \prt})$$



\item\label{multi-client:Generating-temporary-root}  \underline{\textit{Generating an encrypted random root}}:  picks a random root: $\rt_{\st \prt}\stackrel{\st\$}\leftarrow\mathbb{F}_{\st \prm}$. 
It represents $\rt_{\st \prt}$ as a polynomial  $\bm{\gamma}_{\st \prt}(x) =x- \rt_{\st \prt} \bmod \prm$, such that the polynomial's root is $\rt_{\st \prt}$. 
 
It computes \cor $y$-coordinates of  $\bm{\gamma}_{\st \prt}(x)$: 
$$\forall i, 1\leq i \leq \cor: \hspace{4mm}\gamma_{\st i, {\st \prt}} = \bm{\gamma}_{\st \prt}(x_{\st i}) \bmod \prm$$


It encrypts each $y$-coordinate   $\gamma_{\st i, _{\st \prt}}$ using blinding factor $w'_{\st i,\prt}$: 
$$\forall i, 1\leq i \leq \cor: \hspace{4mm}  \gamma'_{\st i, {\st \prt}}=\gamma_{\st i, {\st \prt}}\cdot w'_{\st i,\prt} \bmod \prm$$

It sends $\vv{ \gamma}'_{\st \prt}=[ \gamma'_{\st 1, {\st \prt}},\ldots,  \gamma'_{\st \cor, {\st \prt}}]$ to the rest of the clients.



\item\label{multi-client:Generating-blinding-factors}  \underline{\textit{Generating blinding factors}}: receives $(\bar f_{\st {\st l}}, \vv{ \gamma}'_{\st{l}})$ from every other client in \idx.




Let $j=\id_{\st \prt}$ be the index of one of its own outsourced puzzles, which it wants to use as an input for the linear combination. It regenerates its original blinding factors for its $j$-th solution: 

$$k_{\st j, \prt}=\prf(1, mk_{\st j, \prt}),\hspace{4mm} s_{\st j, \prt}=\prf(2, mk_{\st j, \prt})$$

 $$\forall i, 1\leq i \leq \cor: \hspace{4mm}   z_{\st i,j, \prt}=\prf(i, k_{\st j, \prt}), \hspace{4mm}   w_{\st i,j, \prt}=\prf(i, s_{\st j, \prt})$$

 where  $mk_{\st j, \prt}\in {prm}_{\st \prt}$. It also generates new ones: 
 $$\forall i, 1\leq i \leq \cor: \hspace{4mm}   z'_{\st i,\prt}=\prf(i, k'_{\st \prt})$$

It sets  values $v_{\st i, \prt}$ and $y_{\st i, \prt}$ as follows. $\forall i, 1\leq i \leq \cor:$
\begin{equation*}
\begin{split}
 v_{\st i, \prt} &=\gamma'_{\st i,\prt}\cdot  \prod\limits_{\st \forall \cl_{_{\st l}}\in \idx \setminus\prtt_{\st\prt}}
\gamma'_{\st i,{\st l}} \bmod \prm\\
y_{\st i, \prt}&=  -\sum\limits_{\st \forall \cl_{_{\st l}}\in C\setminus\prtt_{\st\prt}}\prf(i, f_{\st {l}}) +  \sum\limits_{\st \forall \cl_{_{\st l}}\in \idx\setminus\prtt_{\st\prt}}\prf(i, \bar{f}_{\st {l}})      \bmod 
\end{split}
\end{equation*}

where $\prtt_{\st\prt}\in \idx$.

\item\label{multi-instance-multi-client-eval:ole-detect} \underline{\textit{Re-encoding outsourced puzzle}}: 
obliviously prepares the puzzle (held by \srv) for the computation. To do that, it participates in an instance of $\ole^{\st +}$ with \srv, for every $i$, where $1\leq i \leq \cor$.  The inputs of $\prtt_{\st\prt}$ to the $i$-th instance of $\ole^{\st +}$ are: 
\begin{equation*}
\begin{split}
e_{\st i, \prt} &=q_{\st \prt}\cdot  v_{\st i, \prt} \cdot (w_{\st i, j, \prt})^{\st -1}\bmod \prm\\
 e'_{\st i, \prt} &= -(q_{\st \prt} \cdot  v_{\st i, \prt}   \cdot  z_{\st i, j, \prt}) + z'_{\st i,\prt} + y_{\st i, \prt}\bmod \prm
\end{split}
\end{equation*}

The input of \srv to $i$-th instance of $\ole^{\st +}$ is $\prtt_{\st\prt}$'s encrypted $y$-coordinate of its $j$-th puzzle: $e''_{\st i, \prt} = o_{\st i,j, \prt}$ (where $o_{\st i,j, \prt}\in \vec{o}$). Accordingly, the $i$-th instance of $\ole^{\st +}$ returns to \srv:
\begin{equation*}
\begin{split}
d_{\st i, \prt}& = e_{\st i, \prt}\cdot e''_{\st i, \prt} + e'_{\st i, \prt}\\ 
&=q_{\st \prt}  \cdot v_{\st i, \prt} \cdot  \pi_{\st i,j, \prt}+z'_{\st i, \prt}+y_{\st i, \prt}\bmod \prm\\ &= q_{\st \prt}\cdot \gamma_{\st i, {\st \prt}}\cdot w'_{\st i, \prt}
\cdot (\prod\limits_{\st \forall \cl_{_{\st l}}\in \idx\setminus\prtt_{\st\prt}}
\gamma_{\st i,{\st l}}\cdot w'_{\st i, {\st l}})
\cdot  \pi_{\st i,j, \prt}+z'_{\st i,\prt}+y_{\st i, \prt}\bmod \prm
\end{split}
\end{equation*}


where $q_{\st \prt}$ is the party's coefficient. If $\prtt_{\st\prt}$ detects misbehavior during the execution of  $\ole^{\st +}$, it sends a special symbol $\bot$ to all parties and halts.

\item\label{multi-client:Committing-to-the-root} \underline{\textit{Committing to the root}}: computes $com'_{\st \prt}=\comcom(\rt_{\st \prt}, tk_{\st \prt})$.

\item\label{multi-client::Publishing-public-parameters--} \underline{\textit{Publishing public parameters}}: publishes ${pp}_{\st \prt}^{\st(\text{Evl})}=(h_{\st \prt}, com'_{\st \prt}, Y)$.  Note that every leader client $\prtt_{\st\prt} \in \idx$  uses identical $Y$.  Let $\vec{pp}^{\st(\text{Evl})}$ contain all the  triples ${pp}_{\st \prt}^{\st(\text{Evl})}$ published by $\prtt_{\st\prt}$, where $\prtt_{\st\prt} \in \idx$.


\end{enumerate}

\item\underline{\textit{Granting the computation by each non-leader client}}: 
each non-leader client $\prtt_{\st\prt}$ takes the following steps.


\begin{enumerate}
\item \underline{\textit{Generating blinding factors}}: receives $(\bar{f}_{\st {l}}, \vv{ \gamma}'_{\st {l}})$  from every other leader client  in $\idx$.

As before, let $j=\id_{\st \prt}$ be the index of one of client $\prtt_{\st\prt}$ outsourced puzzles that it wants to use as an input for the linear combination. It regenerates its original blinding factors: 

$$k_{\st j, \prt}=\prf(1, mk_{\st j, \prt}),\hspace{4mm} s_{\st j, \prt}=\prf(2, mk_{\st j, \prt})$$
 $$\forall i, 1\leq i \leq \cor: \hspace{4mm}   z_{\st i,j, \prt}=\prf(i, k_{\st j, \prt}), \hspace{4mm}   w_{\st i,j, \prt}=\prf(i, s_{\st j, \prt})$$

It set  values $v_{\st i, \prt}$ and $y_{\st i, \prt}$ as follows. $\forall i, 1\leq i \leq \cor:$
\begin{equation*}
\begin{split}
 v_{\st i,\prt} &= \prod\limits_{\st \forall \cl_{_{\st l}}\in \idx}
\gamma'_{\st i,{\st l}} \bmod \prm\\
y_{\st i, \prt}&=    \sum\limits_{\st \forall \cl_{_{\st l}}\in \idx}\prf(i, \bar{f}_{\st {l}})         \bmod \prm
\end{split}
\end{equation*}

\item\label{eval:ole-detect-2} \underline{\textit{Re-encoding outsourced puzzle}}: participates in an instance of $\ole^{\st +}$ with the server \srv, for every $i$, where $1\leq i \leq \cor$.  The inputs of $\prtt_{\st\prt}$ to the $i$-th instance of $\ole^{\st +}$ are: 
\begin{equation*}
\begin{split}
e_{\st i, \prt} &= q_{\st \prt}\cdot  v_{\st i, \prt}\cdot   (w_{\st i,j, \prt})^{\st -1}\bmod \prm\\
e'_{\st i, \prt} &= -(q_{\st \prt}\cdot v_{\st i, \prt}\cdot    z_{\st i,j, \prt}) + y_{\st i, \prt}\bmod \prm
\end{split}
\end{equation*}

The input of \srv to the $i$-th instance of $\ole^{\st +}$ is $\prtt_{\st\prt}$'s encrypted $y$-coordinate: $e''_{\st i} = o_{\st i,j, \prt}$. Accordingly, the $i$-th instance of $\ole^{\st +}$ returns to \srv:
\begin{equation*}
\begin{split}
 d_{\st i, \prt}& = e_{\st i, \prt} \cdot e''_{\st i, \prt} +  e'_{\st i, \prt}\\ 
 &=  q_{\st \prt}\cdot v_{\st i, \prt} \cdot \pi_{\st i,j, \prt} + y_{\st i, \prt}\bmod \prm\\ 
 &= q_{\st \prt}\cdot (\prod\limits_{\st \forall \cl_{_{\st l}}\in \idx\setminus\prtt_{\st\prt}}
\gamma_{\st i, {\st l}}\cdot w'_{\st i, {\st l}})
\cdot  \pi_{\st i, j, \prt}+y_{\st i, \prt}\bmod \prm
\end{split}
\end{equation*}

where $q_{\st \prt}$ is the party's coefficient. If $\prtt_{\st\prt}$ detects misbehavior during the execution of  $\ole^{\st +}$, it sends a special symbol $\bot$ to all parties and halts.



\end{enumerate}

\item\label{multi-instance-multi-client-::Computing-encrypted-linear-combination}\underline{\textit{Computing encrypted linear combination}}:  server \srv sums all of the outputs of $\ole^{\st +}$ instances that it has invoked, $\forall i, 1\leq i \leq \cor:$ 
\begin{equation*}
\begin{split}
g_{\st i} &= \sum\limits_{\st \forall \prtt_{\st\prt}\in C} d_{\st {i, \prt}} \bmod \prm\\ &=    (
\prod\limits_{\st \forall \prtt_{\st\prt}\in  \idx} \underbrace{\gamma_{\st i, \prt}\cdot w'_{\st i, \prt}}_{\st v_{\st i, \prt}}
\cdot \sum\limits_{\st \forall \prtt_{\st\prt}\in C}   q_{\st \prt}\cdot \pi_{\st i,j, \prt} )+    \sum\limits_{\st \forall \prtt_{\st\prt}\in \idx}  z'_{\st i,\prt}  \bmod  \prm
\end{split}
\end{equation*}


\item\label{multi-client-publish-encrypted-solutions}\underline{\textit{Disseminating encrypted result}}:  server \srv publishes $\vec{g}=[g_{\st 1}, \ldots, g_{\st \cor}]$.

\end{enumerate}

\item\label{multi-instance-multi-client-::Solving-a-Puzzle} \underline{\textbf{Solving a Puzzle}}. $\solv(\vec{o}_{\st \prt}, pp_{\st \prt}, \vec{g}, \vec{pp}^{\st (\text{Evl})}, {pp}_{\st \prt}^{\st (\text{Evl})}, \vec{pk},  pk_{\st \srv},\cmd,\hat{\cmd})\rightarrow(\vec{m}, \vec{\zeta})$

Server \srv takes the following steps.

\begin{steps}[leftmargin=15mm]


\item\label{multi-instance-multi-client-solving-linear-combination-x}\hspace{-2mm}. when solving a puzzle corresponding to the linear combination of messages (i.e., when $\cmd=\ep$), where all messages belong to the same client $\prtt_{\st \prt}$, i.e., when $\hat{\cmd}=\singclient$. Note that in this case, $\vec{o}_{\st \prt}$ and $\vec{pp}^{\st (\text{Evl})}$ can be null. 

\begin{enumerate}

\item\label{multi-instance-multi-client-extract-temp-key-x}  \underline{\textit{Finding secret keys}}: 

\begin{enumerate}

\item finds temporary key $tk$, where  $tk = h^{\st 2^{\st Y}} \bmod N_{\st \prt}$, via repeated squaring of $h$ modulo $N_{\st \prt}$, where $h\in {pp}_{\st \prt}^{\st (\text{Evl})}, N_{\st \prt}\in \vec{pk}$.

\item derives two keys from $tk$: 
$$ k'=\prf(1, tk), \hspace{4mm} s'=\prf(2, tk)$$

\end{enumerate}

\item\label{multi-instance-multi-client-::server-side-computatoin-removing-pr-vals-x} \underline{\textit{Removing blinding factors}}: removes the blinding factors from $[g_{\st 1}, \ldots, g_{\st \cor}]\in \vec{g}$. 

$\forall i, 1\leq i \leq \cor:$
\begin{equation*}
\begin{split}
\theta_{\st i}&=\underbrace{\big( \prf(i, s')\big)^{\st -1}}_{\st (w'_{\st i})^{\st -1}}\cdot\big(g_{\st i}- \overbrace{\prf(i, k'}^{\st z'_{\st i}})\big) \bmod \prm\\ &=
\gamma_{\st i}  \cdot \sum\limits_{\st j=1}^{\st z}   q_{\st j, \prt}\cdot \pi_{\st i, j, \prt} \bmod \prm
\end{split}
\end{equation*}

\item\label{multi-instance-multi-client--interpolate-poly-x}  \underline{\textit{Extracting a polynomial}}: interpolates a polynomial $\bm{\theta}(x)$, given pairs $(x_{\st 1}, \theta_{\st 1}), \ldots, (x_{\st \cor}, \theta_{\st \cor})$.  Note that $\bm{\theta}(x)$ will have the form: 
$$\bm\theta(x) =(x-\rt)\cdot \sum\limits_{\st j=1}^{\st z}q_{\st j, \prt}\cdot (x+m_{\st j, \prt}) \bmod \prm$$

We can rewrite $\bm\theta(x)$ as: 
$$\bm\theta(x) = \bm\psi(x)-\rt\cdot  \sum\limits_{\st j=1}^{\st z} q_{\st j, \prt}\cdot m_{\st j, \prt}   \bmod \prm$$

where $\bm\psi(x)$ is a polynomial of degree two with constant term being $0$.

\item\label{multi-instance-multi-client-::server-side-Extracting-the-linear-combination-x} \underline{\textit{Extracting the linear combination}}: retrieves the result (i.e., the linear combination of $m_{\st 1, \prt},\ldots, $ $m_{\st z, \prt}$)  from polynomial $\bm\theta(x)$'s constant term: $cons=-\rt\cdot  \sum\limits_{\st j=1}^{\st z} q_{\st j, \prt}\cdot m_{\st j, \prt}$ as follows:
\begin{equation*}
\begin{split}
m &=cons\cdot (-\rt)^{\st -1}\bmod \prm\\ &= \sum\limits_{\st j=1}^{\st z} q_{\st j, \prt}\cdot m_{\st j, \prt}
\end{split}
\end{equation*}

\item\label{multi-instance-multi-client-::Extracting-valid-roots-x} \underline{\textit{Extracting valid roots}}: extracts the root(s) of polynomial $\bm{\theta}(x)$. Let set $R$ contain the extracted roots. It identifies the valid root, by finding a root $\rt$ in $R$, where it can pass the verification of the commitment scheme, i.e., $\comver(com',(\rt, tk))=1$.

\item\label{single-client-publish-lc-proof-x} \underline{\textit{Publishing the result}}: initiates  vectors $\vec{m}$ and $\vec\zeta$. It appends $m$ to  $\vec{m}$ and $(\rt, tk)$ to $\vec{\zeta}$. It publishes $\vec{m}$ and $\vec{\zeta}$.

\end{enumerate}



\item\label{multi-instance-multi-client-solving-linear-combination}\hspace{-2mm}. when solving a puzzle related to the linear combination of messages (i.e., when $\cmd=\ep$), where each message belongs to a different client, i.e., when $\hat{\cmd}=\mulclient$. In this case, $\vec{o}_{\st \prt}$ and ${pp}_{\st \prt}^{\st (\text{Evl})}$ can be null. 

\begin{enumerate}

\item\label{multi-instance-multi-client-find-master-key} \underline{\textit{Finding secret keys}}: for  each leader client $\prtt_{\st\prt}\in \idx$:

\begin{enumerate}
\item\label{extract-temp-key} finds $tk_{\st \prt}$ (where  $tk_{\st \prt}=h^{\st 2^{\st Y}}_{\st \prt}\bmod N_{\st \prt}$) through repeated squaring of $h_{\st \prt}$ modulo $N_{\st \prt}$. Note that $(h_{\st \prt}, Y, N_{\st \prt})\in \vec{pp}^{\st (\text{Evl})}$.

\item derives two keys from $tk_{\st \prt}$: 
$$ k'_{\st \prt}=\prf(1, tk_{\st \prt}), \hspace{4mm} s'_{\st \prt}=\prf(2, tk_{\st \prt})$$

\end{enumerate}

\item\label{multi-instance-multi-client-::server-side-computatoin-removing-pr-vals} \underline{\textit{Removing blinding factors}}:  removes the blinding factors from $[g_{\st 1}, \ldots, g_{\st \cor}]\in \vec{g}$. 

$\forall i, 1\leq i \leq \cor:$
\begin{equation*}
\begin{split}
\theta_{\st i}&=\big(\prod\limits_{\st \forall \prtt_{\st\prt}\in \idx}\underbrace{\prf(i, s'_{\st \prt})}_{\st w'_{\st i, \prt}}\big)^{\st -1}\cdot\big(g_{\st i}-\sum\limits_{\st \forall \prtt_{\st\prt}\in \idx} \overbrace{\prf(i, k'_{\st \prt}}^{\st z'_{\st i, \prt}})\big) \bmod \prm\\ &=
(\prod\limits_{\st \forall \prtt_{\st\prt}\in\idx} \gamma_{\st i, \prt}) \cdot \sum\limits_{\st \forall \prtt_{\st\prt}\in C}   q_{\st \prt}\cdot \pi_{\st i,j, \prt} \bmod \prm
\end{split}
\end{equation*}

\item\label{multi-instance-multi-client-step::multi-client-interpolate-poly} \underline{\textit{Extracting a polynomial}}: interpolates a polynomial $\bm{\theta}(x)$, given pairs $(x_{\st 1}, \theta_{\st 1}),\ldots, (x_{\st \cor}, \theta_{\st \cor})$.  Polynomial $\bm{\theta}(x) $ will have the following form: 
\begin{equation*}
\bm\theta(x) =\prod\limits_{\st \forall \prtt_{\st\prt}\in  \idx} (x-\rt_{\st \prt})\cdot \sum\limits_{\st \forall \prtt_{\st\prt}\in C}q_{\st \prt}\cdot (x+m_{\st j, \prt}) \bmod \prm
\end{equation*}

Note that $j=\id_{\st \prt}$ and may have different value for different client $\prtt_{\st \prt}$. It is possible to rewrite $\bm\theta(x)$ as: 
$$\bm\theta(x) = \bm\psi(x)+\prod\limits_{\st \forall \prtt_{\st\prt}\in  \idx}(-\rt_{\st \prt})\cdot \sum\limits_{\st \forall\prtt_{\st\prt}\in C}q_{\st \prt}\cdot m_{\st j, \prt}   \bmod \prm$$

with $\bm\psi(x)$ being a polynomial of degree $\tl+1$ that has constant term $0$.

\item\label{multi-instance-multi-client-::server-side-Extracting-the-linear-combination} \underline{\textit{Extracting the linear combination}}: retrieves the final result, that is the linear combination of the messages $m_{\st j, 1},\ldots, m_{\st j, n}$,  from polynomial $\bm\theta(x)$'s constant term: $cons=\prod\limits_{\st \forall \prtt_{\st\prt}\in  \idx}(-\rt_{\st \prt})\cdot\sum\limits_{\st \forall\prtt_{\st\prt}\in C}q_{\st \prt}\cdot m_{\st j, \prt}$ as follows:
\begin{equation*}
\begin{split}
m &=cons\cdot (\prod\limits_{\st \forall \prtt_{\st\prt}\in  \idx}(-\rt_{\st \prt}))^{\st -1}\bmod \prm\\ &= \sum\limits_{\st \forall\prtt_{\st\prt}\in C}q_{\st \prt}\cdot m_{\st j, \prt}
\end{split}
\end{equation*}

where each $j=\id_{\st \prt}$.



\item\label{multi-instance-multi-client-:Extracting-valid-roots} \underline{\textit{Extracting valid roots}}: retreives the roots of polynomial $\bm{\theta}(x)$. Let set $R$ contain the extracted roots. It identifies the valid roots, by finding every  $\rt_{\st \prt}$ in $R$, such that it passes the commitment's verification: $\comver(com'_{\st \prt},(\rt_{\st \prt}, tk_{\st \prt}))=1$. This check is performed for every $\prtt_{\st\prt}$ in $\idx$. 

%

\item\label{multi-client-publish-lc-proof} \underline{\textit{Publishing the result}}: initiates empty vectors $\vec{m}$ and $\vec{\zeta}$. It appends $m$ to $\vec{m}$. Also, for every $\prtt_{u}$ in  $\idx$, it appends $(\rt_{\st \prt}, tk_{\st \prt})$ to $\vec{\zeta}$. It publishes $\vec{m}$ and $\vec{\zeta}$.



\end{enumerate}


\item\label{multi-instance-multi-client--solving-only-one-puzzle-x} when solving each $j$-th puzzle $\vec{o}_{\st j}$ in $\vec{o}$ of client \prtt (i.e., when $\cmd=\scp$), server \srv takes the following steps.  Note that in this case, $\vec{pp}^{\st (\text{Evl})}$  and ${pp}_{\st \prt}^{\st (\text{Evl})}$ and can be null. $\forall j, 1\leq j\leq z:$

\begin{enumerate}

\item\label{multi-instance-multi-client-::find-base-x} \underline{\textit{Finding secret bases and keys}}: sets base $r_{\st j}$ and $mk_{\st j, \prt}$ as follows. 

\begin{itemize}[label=$\bullet$]
\item if $j=1: $ sets the base to $r_{\st 1}$, where $r_{\st 1}\in pp_{\st \prt}$. Then, it finds  $mk_{\st 1,\prt}$ where  $mk_{\st 1,\prt}=r^{\st 2^{\st T_{\st 1}}}_{\st 1}\bmod N_{\st \prt}$, through repeated squaring of $r_{\st 1}$ modulo $N_{\st \prt}$. It initiates  vectors $\vec{m}$ and $\vec\zeta$.

\item if $j>1: $ computes base $r_{\st j}$ as   $r_{\st j}=\prf(j||0, mk_{\st j-1, \prt})$. Next, it  finds  $mk_{\st j, \prt}$ where  $mk_{\st j, \prt}=r^{\st 2^{\st T_{\st j}}}_{\st j}\bmod N_{\st \prt}$, via repeated squaring of $r_{\st j}$ modulo $N_{\st \prt}$.
\end{itemize}

It derive two keys from $mk_{\st j, \prt}$: 
$$ k_{\st j, \prt}=\prf(1, mk_{\st j, \prt}), \hspace{4mm} s_{\st j, \prt}=\prf(2, mk_{\st j, \prt})$$

\item\label{multi-instance-multi-client-::Removing-blinding-factors-server-side-x} \underline{\textit{Removing blinding factors}}: re-generates $2\cdot\cor$ pseudorandom values using $k_{\st j, \prt}$ and  $s_{\st j, \prt}$:
$$ \forall i, 1\leq i \leq \cor: \hspace{4mm}  z_{\st i, j, \prt}=\prf(i, k_{\st j, \prt}), \hspace{4mm} w_{\st i, j, \prt}=\prf(i, s_{\st j, \prt})$$

Next, it uses the blinding factors to unblind $\vec{o}_{\st j, \prt} = [o_{\st 1, j, \prt}, \ldots, o_{\st \cor, j, \prt}]$:
$$\forall i, 1\leq i \leq \cor: \hspace{4mm}  \pi_{\st i, j, \prt}  = \big((w_{\st i, j, \prt})^{\st -1}\cdot o_{\st i, j, \prt}\big) -z_{\st i, j, \prt} \bmod \prm$$


\item\label{multi-instance-multi-client-::Extracting a polynomial-server-side-x}   \underline{\textit{Extracting a polynomial}}: interpolates a polynomial $\bm{\pi}_{\st j, \prt}(x)$, given pairs $(x_{\st 1}, \pi_{\st 1, j, \prt}), \ldots, (x_{\st \cor}, $ $\pi_{\st \cor, j, \prt})$.

\item\label{single-client-pub-solution-x} \underline{\textit{Publishing the solution}}: considers the constant term of $\bm{\pi}_{\st j, \prt}(x)$ as the plaintext message, $m_{\st j, \prt}$. It appends $(m_{\st j, \prt},j)$ to $\vec{m}$ and $mk_{\st j, \prt}$ to $\vec{\zeta}$. If $j=z$, then it publishes $\vec{m}$ and $\vec{\zeta}$.


\end{enumerate}

\end{steps}


\item\label{multi-client::verification}  \underline{\textbf{Verification}}. $\ver(\vec{m}, \vec{\zeta}, ., pp_{\st \prt}, \vec{g}, \vec{pp}^{\st (\text{Evl})},  {pp}^{\st (\text{Evl})}_{\st \prt}, pk_{\st \srv}, \cmd, \hat{\cmd})\rightarrow \ddot{v}\in\{0,1\}$

A verifier (that can be anyone, not just $\prtt_{\st\prt}\in C$) takes the following steps. 


\begin{steps}[leftmargin=15mm]


\item\label{multi-client-multi-instance-case:single-client-many-messages}\hspace{-2mm}. when verifying a solution related to the linear combination of messages (i.e., when $\cmd=\ep$), where all messages belong to the same client $\prtt_{\st \prt}$, i.e., when $\hat{\cmd}=\singclient$.



\begin{enumerate}

\item\label{step-single-client-verify-opening-x}  \underline{\textit{Checking the commitment's opening}}: verify the validity of  $(\rt, tk)\in \vec\zeta$ with the help of $com'\in{pp}^{\st (\text{Evl})}_{\st \prt}$. 
$$\comver\big(com', (\rt , tk )\big)\stackrel{\st?}=1$$

If the above check passes, it moves on to the next step. Otherwise, it returns $\ddot{v}=0$ and takes no further action. 

\item\label{sinlge-client:Checking-resulting-polynomial-valid-roots-x} \underline{\textit{Checking the resulting polynomial's valid roots}}: checks if the resulting polynomial contains
 the root $\rt  \in \vec\zeta$, by taking the following steps.

\begin{enumerate}

\item derives two keys from $tk$: 
$$ k'=\prf(1, tk), \hspace{4mm} s' = \prf(2, tk)$$

\item\label{sinlge-client:verification-case-1removes the blinding factors-x} removes the blinding factors from $\vec{g}=[g_{\st 1}, \ldots, g_{\st \cor}]$ that were provided by server \srv in step \ref{single-client-publish-encrypted-solutions}. Specifically, for every $i$, 
$ 1\leq i \leq \cor:$
\begin{equation*}
\begin{split}
\theta_{\st i}&=\underbrace{\big( \prf(i, s')\big)^{\st -1}}_{\st (w'_{\st i})^{\st -1}}\cdot\big(g_{\st i}- \overbrace{\prf(i, k'}^{\st z'_{\st i}})\big) \bmod \prm\\ &=
\gamma_{\st i}  \cdot \sum\limits_{\st j=1}^{\st z}   q_{\st j, \prt}\cdot \pi_{\st i, j,  \prt} \bmod \prm
\end{split}
\end{equation*}

\item\label{step-single-client-check-roots-x}  interpolates a polynomial $\bm{\theta}(x)$, given $(x_{\st 1}, \theta_{\st 1}),\ldots, (x_{\st \cor}, \theta_{\st \cor})$.  Note that polynomial $\bm{\theta}(x)$ will have the form: 
\begin{equation*}
\begin{split}
\bm\theta(x)&=(x-\rt)\cdot \sum\limits_{\st j=1}^{\st z} q_{\st j, \prt}\cdot (x+m_{\st j, \prt}) \bmod \prm\\ &= \bm\psi(x)-\rt\cdot \sum\limits_{\st j=1}^{\st z}q_{\st j, \prt}\cdot m_{\st j, \prt}   \bmod \prm
\end{split}
\end{equation*}
where $\bm\psi(x)$ is a polynomial of degree $2$ whose constant term is $0$.

\item\label{single-client::eval-and-check-root-x} checks whether $\rt  \in \vec\zeta$ is a root of $\bm \theta(x)$, i.e.,  $\bm\theta(\rt)\stackrel{\st ?}=0$. It proceeds to the next step if the check passes.  It returns $\ddot{v}=0$ and takes no further action, otherwise.

\end{enumerate}

\item\label{step-sinlge-client-check-res-x} \underline{\textit{Checking the final result}}: retrieves the result (i.e., the linear combination of $m_{\st 1, \prt},\ldots, m_{\st z, \prt}$)  from polynomial $\bm\theta(x)$'s constant term: $t=-\rt\cdot\sum\limits_{\st j=1}^{\st z}q_{\st j, \prt}\cdot m_{\st j, \prt}$ as follows:
\begin{equation*}
\begin{split}
res' &= -t\cdot \rt^{\st -1}\bmod \prm\\  &= \sum\limits_{\st j=1}^{\st z}q_{\st j, \prt}\cdot m_{\st j, \prt}
\end{split}
\end{equation*}

It checks $res' \stackrel{\st ?}=m$, where $m\in \vec{m}$ is the result that \srv sent to it, in step \ref{single-client-publish-lc-proof-x} of \ref{multi-instance-multi-client-solving-linear-combination-x}. 

\item \underline{\textit{Accepting or rejecting the result}}: if all the checks pass, it accepts $\vec{m}$ and returns $\ddot{v}=1$. Otherwise, it returns $\ddot{v}=0$.

\end{enumerate}



\item\label{multi-client-multi-instance-case:multi-client-many-messages}\hspace{-2mm}. when solving a puzzle related to the linear combination of messages (i.e., when $\cmd=\ep$), where each message belongs to a different client, i.e., when $\hat{\cmd}=\mulclient$.

\begin{enumerate}

\item\label{multi-client-multi-instance::step-multi-client-verify-opening} \underline{\textit{Checking the commitments' openings}}: verifies the validity of  every $(\rt_{\st \prt}, tk_{\st \prt})\in \vec\zeta$, provided by \srv in \ref{multi-instance-multi-client-solving-linear-combination}, step \ref{multi-client-publish-lc-proof}: 
$$\forall \prtt_{\st\prt}\in \idx:\hspace{4mm}  \comver\big(com'_{\st \prt},(\rt_{\st \prt}, tk_{\st \prt})\big)\stackrel{\st?}=1$$

where $com'_{\st \prt}\in \vec{pp}^{\st (\text{Evl})}$. 
If all of the verifications pass, it proceeds to the next step. Otherwise, it returns $\ddot{v}=0$ and takes no further action.

\item\label{multi-client-multi-instance:Checking-resulting-polynomial-valid-roots} \underline{\textit{Checking the resulting polynomial's valid roots}}: checks if the resulting polynomial contains all the roots in $\vec\zeta$, by taking the following steps. 

\begin{enumerate}

\item derives two keys from $tk_{\st \prt}$: 
$$ k'_{\st \prt}=\prf(1, tk_{\st \prt}), \hspace{4mm} s'_{\st \prt}=\prf(2, tk_{\st \prt})$$

\item\label{multi-client-multi-instance:verification-case-1removes the blinding factors}  removes the blinding factors from $[g_{\st 1}, \ldots, g_{\st \cor}]\in \vec{g}$ that were provided by \srv in step \ref{multi-client-publish-encrypted-solutions}. 

$\forall i, 1\leq i \leq \cor:$
\begin{equation*}
\begin{split}
\theta_{\st i}&=\big(\prod\limits_{\st \forall \prtt_{\st\prt}\in \idx}\prf(i, s'_{\st \prt})\big)^{\st -1}\cdot\big(g_{\st i}-\sum\limits_{\st \forall \prtt_{\st\prt}\in\idx} \prf(i, k'_{\st \prt})\big) \bmod \prm\\ &=
\prod\limits_{\st \forall \prtt_{\st\prt}\in\idx} \gamma_{\st i, \prt} \cdot \sum\limits_{\st \forall \prtt_{\st\prt}\in C}   q_{\st \prt}\cdot \pi_{\st i,\prt} \bmod \prm
\end{split}
\end{equation*}

\item\label{multi-client-multi-instance:interpolate-poly}  interpolates a polynomial $\bm{\theta}(x)$, using pairs $(x_{\st 1}, \theta_{\st 1}),\ldots, (x_{\st \cor}, \theta_{\st \cor})$.  This results in a polynomial $\bm{\theta}(x)$ having the form: 
\begin{equation*}
\begin{split}
\bm\theta(x) &=\prod\limits_{\st \forall \prtt_{\st\prt}\in\idx} (x-\rt_{\st \prt})\cdot \sum\limits_{\st \forall \prtt_{\st\prt}\in C}q_{\st \prt}\cdot (x+m_{\st \prt}) \bmod \prm\\
&=\bm\psi(x)+\prod\limits_{\st \forall \prtt_{\st\prt}\in  \idx}(-\rt_{\st \prt})\cdot \sum\limits_{\st \forall\prtt_{\st\prt}\in C}q_{\st \prt}\cdot m_{\st \prt}   \bmod \prm
\end{split}
\end{equation*}

where $\bm\psi(x)$ is a polynomial of degree $\tl+1$ whose constant term is $0$.

\item\label{multi-client-multi-instance::step-multi-client-check-roots}  if the following checks pass, it will proceed to the next step; it checks if every $\rt_{\st \prt}\in\vec{\zeta}$ is a root of $\bm \theta(x)$, i.e., $\bm\theta(\rt_{\st \prt})\stackrel{\st ?}=0$.  Otherwise, it returns $\ddot{v}=0$ and takes no further action.

\end{enumerate}

\item\label{multi-client-multi-instance::step-multi-client-check-res} \underline{\textit{Checking the final result}}: retrieves the result (i.e., the linear combination of the messages $m_{\st 1},\ldots, m_{\st n}$)  from polynomial $\bm\theta(x)$'s constant term: $cons = \prod\limits_{\st \forall \prtt_{\st\prt}\in  \idx}(-\rt_{\st \prt})\cdot\sum\limits_{\st \forall\prtt_{\st\prt}\in C}q_{\st \prt}\cdot m_{\st \prt}$ as follows:
\begin{equation*}
\begin{split}
res' &=cons\cdot (\prod\limits_{\st \forall\prtt_{\st\prt}\in  \idx }(-\rt_{\st \prt}))^{\st -1}\bmod \prm\\ &= \sum\limits_{\st \forall\prtt_{\st\prt}\in C}q_{\st \prt}\cdot m_{\st \prt}
\end{split}
\end{equation*}

It checks $res' \stackrel{\st ?}=m$, where $m\in\vec{m}$ is the result that \srv sent to it.

\item\underline{\textit{Accepting or rejecting the result}}:  If all the checks pass, it accepts $\vec{m}$ and returns $\ddot{v}=1$. Otherwise, it returns $\ddot{v}=0$.


\end{enumerate}

\item\label{multi-client-multi-instance::verifying-a-solution-of-single-puzzle}\hspace{-2mm}. when verifying a solution of a single puzzle belonging to $\prtt_{\st \prt}$, i.e., when $\cmd=\scp$: 

\begin{enumerate}

\item\label{step-multi-client-check-single-puzzle-1} \underline{\textit{Checking the commitment' opening}}:  checks whether opening $m_{\st j, \prt}\in \vec{m}$ and  $mk_{\st j, \prt}\in \vec\zeta$ matches the commitment: 
$$\comver\big(com_{\st j, \prt}, (m_{\st j, \prt}, mk_{\st j, \prt})\big)\stackrel{\st?}=1$$

where $com_{\st j, \prt}\in pp_{\st \prt}$.

\item\label{step-multi-client-check-single-puzzle-2}  \underline{\textit{Accepting or rejecting the solution}}: accepts the solution $\vv{m}$ and returns $\ddot{v}=1$ if the above check passes. It rejects the solution and returns $\ddot{v}=0$, otherwise. 
\end{enumerate}

\end{steps}
\end{enumerate}

\begin{theorem}[informal]\label{theo:security-of-VH-TLP}
If  \mhtlp and \tf are secure, then \mmhtlp is secure. 
\end{theorem}

\begin{proof}[sketch] From the security perspective, \mmhtlp does not introduce any new security mechanism and relies on those proposed in  \mhtlp and \tf. Thus, its security (i.e.,  privacy and solution validity) boils down to the security of \mhtlp and \tf. 
\hfill$\square$
\end{proof}


%% file: cost-analysis.tex

\section{Evaluation}

In this section, we evaluate the costs and features of our schemes and compare them with those of existing TLPs that support homomorphic linear combinations, namely, with the TLPs  proposed in \cite{tempora-fusion,MalavoltaT19,liu2022towards,dujmovic2023time}. We exclude the TLP in \cite{SrinivasanLMNPT23}, from our analysis, as its authors acknowledge that it is far from practically efficient. Tables \ref{complexity-of-our-schemes} and \ref{feature-of-our-schemes} summarize the results.

\subsection{Asymptotic Cost}\label{sec::Cost-Analysis}


\subsubsection{\mhtlp.}\label{sec::cost-of-tf}



We begin by analyzing the computation cost of a client. 

\noindent\underline{\textit{Client's Costs.}} The computation costs of a client are as follows. In the Puzzle Generation phase, in each step \ref{single-GC-TLP::compute-a-values} and \ref{single-GC-TLP::compute-r-values}, a client  performs $z$ modular exponentiation over $\phi(N)$ and $N$ respectively. Furthermore, in steps  \ref{single-GC-TLP::compute-r-values},  \ref{single-client:derive-keys-for-puzzle}, and \ref{single-client::Generating-blinding-factors}, in total the client invokes $9\cdot z-1$ instances of \prf. 
In step \ref{sinlge-client:rep-message-as-poly}, it performs $z$ modular addition. In step \ref{sinlge-client:gen-y-coord-of-poly}, it evaluates a polynomial of degree one at three $x$-coordinates, which will involve $3\cdot z$ modular additions. 
In step \ref{sinlge-client:enc-y-coord-of-poly}, the client performs $3\cdot z$ additions and multiplications to encrypt the $y$-coordinates. In step \ref{sinlge-client:commit-y-coord-of-poly}, the client invokes the hash function $z$ times to commit to each message. 

In the Linear Combination Phase, in step \ref{sinlge-client:Generating-temporary-secret-keys}, it performs two modular exponentiations, one over $\phi(N)$ and the other over $N$. In the same step, it invokes \prf twice. 
In step \ref{single-client:Generating-temporary-pr-vals}, it invokes $11\cdot z+3$ instances of \prf. In the same step, it performs $z-1$ modular addition. In step \ref{sinlge-client:Generating-temporary-root}, it performs 3 additions. In step \ref{sinlge-client:Committing-to-the-root}, the client invokes the hash function one. In step \ref{single-client::eval:ole-detect}, the client performs $ 6\cdot z$ additions and $12\cdot z$ multiplications. In the same step, it invokes $3\cdot z$ instances of $\ole^{\st +}$. Therefore, the computation complexity of the client is $O(z)$.


The communication costs of a client are as follows. 
In the Key Generation phase, step \ref{single-client::Publishing-public-parameters}, the client publishes a single public key of size about $2048$ bits. 
In the Puzzle Generation phase, step \ref{sinlgeclient--Managing-messages}, the client publishes $4\cdot z+1$ values.   
In the Linear Combination phase,  
in step \ref{single-client::eval:ole-detect}, it invokes $3\cdot z$ instances of $\ole^{\st +}$ where each instance imposes $O(1)$ communication cost. 
In step \ref{sinlge-client::Publishing-public-parameters--x}, the client publishes two elements. 
Therefore, the client's communication complexity is $O(z)$.


\noindent\underline{\textit{Verifier's Costs.}} The computation costs of a verifier include the following operations. 
In the Verification phase, the computation cost of a verifier in \ref{single-client::verifying-linear-combination} is as follows. In step \ref{step-single-client-verify-opening-}, it invokes an instance of the hash function. 
In step \ref{sinlge-client:Checking-resulting-polynomial-valid-roots-}, it invokes $6$ instances of \prf. 
In step \ref{sinlge-client:verification-case--removes-the-blinding-factors}, it performs $3$ additions and $1$ multiplication.  In step \ref{step-single-client-check-roots}, it interpolates a polynomial of degree $2$.  
In step \ref{single-client::eval-and-check-root}, it evaluates a polynomial of degree $2$ at a single point, requiring $2$ additions and multiplications.  
In step \ref{step-sinlge-client-check-res-}, it performs a single multiplication. Thus, the verifier's computation complexity in  \ref{single-client::verifying-linear-combination} is $O(1)$. 
In the Verification phase, the computation cost of a verifier in \ref{sinlge-client::verifying-a-solution-of-single-puzzle} involves only a single invocation of the hash function to check the opening of a commitment, for each puzzle.  
Hence, the computation complexity of the verifier $O(z)$.

%
\noindent\underline{\textit{Server's Costs.}} Now, we consider the computation cost of a server. 
In step \ref{single-client::eval:ole-detect}, server \srv engages $3\cdot z$ instances of   $\ole^{\st +}$ with each client. 
In step \ref{sinlge-client::Computing-encrypted-linear-combination-},  \srv performs $3\cdot z$ modular addition. 
During the Solving  Puzzles phase, in \ref{single-client::solving-linear-combination} step \ref{single-client-extract-temp-key},  \srv performs $Y$ repeated modular squaring and invokes two instances of \prf. In step \ref{sinlge-client::server-side-computatoin-removing-pr-vals}, \srv performs $3$ additions and $3$ multiplications.  
In step \ref{step::sinlge-client-interpolate-poly}, it interpolates a polynomial of degree $2$ that involves $O(1)$ addition and multiplication operations (note that the complexity is constant with regard to the number of puzzles). 
In step \ref{sinlge-client::server-side-Extracting-the-linear-combination}, it performs a single modular multiplication. In step \ref{sinlge-client::Extracting-valid-roots}, it factorizes a polynomial of degree $2$ to find its root, costing $O(1)$. 
In the same step, it invokes the hash function once to identify the valid roots. Thus, the computation complexity of \srv in \ref{single-client::solving-linear-combination} is $O(Y+z)$.

 In \ref{single-client-solving-only-one-puzzle},  the cost of  \srv for a client $\prtt_{\st \prt}$ involves the following operations.  \srv performs $O(\mxsqr\cdot \sum\limits_{\st j=1}^{\st z}\bar{\Delta}_{\st j, \prt})$ modular squaring over $N$ to find the master keys. It invokes $9\cdot z-1$ instances of \prf. It performs $3$ addition and $3$ multiplication to decrypt $y$-coordinates. It interpolates a polynomial of degree $2$ using $3$ coordinates, requiring $O(1)$ addition and multiplication operations. Hence, the complexity of \srv in \ref{single-client-solving-only-one-puzzle} is $O(\mxsqr\cdot \sum\limits_{\st j=1}^{\st z}\bar{\Delta}_{\st j, \prt})$.


Next, we analyze the communication costs of \srv. In the Setup phase, \srv publishes $4$ messages. In the Linear Combination phase, step \ref{single-client::eval:ole-detect},  it invokes $3\cdot z$ instances of $\ole^{\st +}$ with the client, where each instance imposes $O(1)$ communication cost. In step \ref{single-client::publish-encrypted-solutions}, it publishes $3$ messages. In the Solving a Puzzle phase, step \ref{single-client-publish-lc-proof}, it publishes $3$ messages. In \ref{single-client-solving-only-one-puzzle}, step \ref{single-client-pub-solution-}, the server publishes two messages. Hence, the communication complexity of \srv is $O(z)$.

\input{multi-instance-version}

%% file: multi-instance-version.tex

\subsubsection{\mmhtlp.}\label{sec::cost-of-Multi-Instance-multi-client}

As before, we begin by evaluating the computation cost of a client.

\noindent\underline{\textit{Client's Costs.}} The computation costs of a client are as follows. Recall that in \mmhtlp, the number of $x$-coordinates is linear with the number of leaders $\tl$, whereas in \mhtlp  it is $3$. However, the client in \mmhtlp takes the same types of steps as it takes in  \mhtlp.  Thus, the computation cost complexity of a client in  \mmhtlp is $O(\tl\cdot z)$.

We proceed to analyze the communication costs of a client. In Phase \ref{multi-instance-multi-client-::key-gen}, the client publishes a single public key of size about $2048$ bits. 
In Phase \ref{multi-instance-multi-client-tf-Puzzle-Generation-phase}, the client publishes $(\cor+1)\cdot z+1$ messages.  
In Phase \ref{multi-client-multi-instance-phase::multi-instance-TF-Linear-Combination}, it invokes $\cor\cdot z$ instances of $\ole^{\st +}$ where each instance imposes $O(1)$ communication cost. In the same phase, it publishes two elements. In Phase \ref{multi-client-multi-instance-tf-LinearCombination-phase-x},  we will consider the communication cost of a leader client, as it transmits more messages than non-leader clients. The leader client transmits to each client a key for \prf. It also sends $\cor$ encrypted $y$-coordinates of a random root to the rest of the clients. It invokes \cor instances of $\ole^{\st +}$. The leader client also publishes three elements $(h_{\st \prt}, com'_{\st \prt}, Y)$. Therefore, the leader client's communication complexity is $O(\cor\cdot n)$. Hence, the client's communication complexity is $O((\tl+n)\cdot z)$.


\noindent\underline{\textit{Verifer's Costs.}} We will analyze only the computation costs of a verifier, as the protocol imposes no communication overhead on the verifier. 
In the Verification phase, in  \ref{multi-client-multi-instance-case:single-client-many-messages}, a verifier performs the same type of computation it does in Case 1 of \mhtlp, however, in the former the number of $x$-coordinates is \cor (instead of being $3$ in \mhtlp). Hence, the verifier's computation complexity in \ref{single-client::verifying-linear-combination}  is $O(\tl)$. In \ref{multi-client-multi-instance-case:multi-client-many-messages}, 
the computation cost of the verifier is as follows. In step \ref{multi-client-multi-instance::step-multi-client-verify-opening}, it invokes \tl instances of the hash function. In step \ref{multi-client-multi-instance:Checking-resulting-polynomial-valid-roots} it invokes $2\cdot(\cor\cdot\tl+1)$ instances of \prf. 
In step \ref{multi-client-multi-instance:verification-case-1removes the blinding factors}, it performs $\cor\cdot\tl+1$ additions and $\cor\cdot\tl$ multiplication.  
In step \ref{multi-client-multi-instance:interpolate-poly}, it interpolates a polynomial of degree $\tl+1$. This involves $O(\tl)$ addition and  $O(\tl)$ multiplication. In step \ref{multi-client-multi-instance::step-multi-client-check-roots}, it evaluates a polynomial of degree $\tl+1$ at \tl points, resulting in $\tl^{\st 2}+\tl$ additions and $\tl^{\st 2}+\tl$ multiplication. 
Moreover, in step \ref{multi-client-multi-instance::step-multi-client-check-res}, it performs $\tl+1$ multiplication. Therefore, its complexity in  \ref{multi-client-multi-instance-case:multi-client-many-messages} is $\tl^{\st 2}+\tl$. 
The computation cost of the verifier in \ref{multi-client-multi-instance::verifying-a-solution-of-single-puzzle} involves a single invocation of the hash function to check the opening of a commitment for each puzzle. Thus, its complexity in this case is $O(z)$. 
We conclude that when it verifies $z$ puzzles of a client and a linear combination of $n$ clients' messages, the total computation complexity of the verifier is $O(\tl^{\st 2}+\tl+z)$.


\noindent\underline{\textit{Server's Costs.}} Initially, we will focus on the computation costs of \srv. During Phase \ref{multi-client-multi-instance-phase::multi-instance-TF-Linear-Combination} (Linear Combination for a Single Client), in step \ref{multi-client-multi-instance::eval:ole-detect}, \srv engages $\cor\cdot z$ instances of   $\ole^{\st +}$ with a client. 
In step \ref{multi-client-multi-instance::Computing-encrypted-linear-combination}, \srv performs $\cor\cdot z$ modular addition. Thus, its computation complexity in this phase is $O(\tl\cdot z)$. 
Within Phase \ref{multi-client-multi-instance-tf-LinearCombination-phase-x} (Linear Combination for Multiple Clients),  in step \ref{multi-instance-multi-client-eval:ole-detect},  \srv engages $\cor$ instances of   $\ole^{\st +}$ with each client. In step \ref{multi-instance-multi-client-::Computing-encrypted-linear-combination},  \srv performs $\cor\cdot n$ modular addition. Therefore, its complexity in this phase is $O(\tl\cdot n)$. 

During Phase \ref{multi-instance-multi-client-::Solving-a-Puzzle}, \ref{multi-instance-multi-client-solving-linear-combination-x}, step \ref{multi-instance-multi-client-extract-temp-key-x}, server \srv performs $Y$ repeated modular squaring and invokes two instances of \prf. 
In step \ref{multi-instance-multi-client-::server-side-computatoin-removing-pr-vals-x}, it performs $\cor$ addition and $\cor$ multiplication.  
In step \ref{multi-instance-multi-client--interpolate-poly-x}, it interpolates a polynomial using $\cor$ data points, which results in $O(\tl)$ computation complexity. 
In step \ref{multi-instance-multi-client-::server-side-Extracting-the-linear-combination-x}, it performs a single modular multiplication. In step \ref{multi-instance-multi-client-::Extracting-valid-roots-x}, it factorizes a polynomial of degree $2$ to find its root, which will cost $O(1)$. 
In step \ref{multi-instance-multi-client-::Extracting-valid-roots-x}, it invokes the hash function once. Therefore, the overall computation complexity of \srv in Phase \ref{multi-instance-multi-client-::Solving-a-Puzzle}, \ref{multi-instance-multi-client-solving-linear-combination-x} is $O(Y+\tl\cdot(z+n))$.

In Phase \ref{multi-instance-multi-client-::Solving-a-Puzzle}, \ref{multi-instance-multi-client-solving-linear-combination}, the cost of \srv is as follows. 
In step \ref{multi-instance-multi-client-find-master-key}, \srv performs $Y$ modular squaring to find master key $mk_{\st \prt}$ for each leader client. In the same step, it invokes $2$ instances of \prf for each leader client. 
In step \ref{multi-instance-multi-client-::server-side-computatoin-removing-pr-vals}, it invokes $2\cdot(\cor+1)$ instances of \prf. In the same step, it performs $\cor+1$ addition and \cor multiplication. In step \ref{multi-instance-multi-client-step::multi-client-interpolate-poly}, it interpolates a polynomial using $\cor$ points, involving $O(\tl)$ addition and  $O(\tl)$ multiplication operations. In step \ref{multi-instance-multi-client-::server-side-Extracting-the-linear-combination}, it performs $\tl+1$ multiplication. In step \ref{multi-instance-multi-client-:Extracting-valid-roots}, it factorizes a polynomial of degree $\tl+1$ with the computation complexity of $O(\tl^{\st 2})$. Hence, the total computation complexity of \srv (for $n$ clients) in Phase \ref{multi-instance-multi-client-::Solving-a-Puzzle}, \ref{multi-instance-multi-client-solving-linear-combination} is $O(\tl\cdot n+\tl^{\st 2}+\tl\cdot Y)$.

In Phase \ref{multi-instance-multi-client-::Solving-a-Puzzle}, \ref{multi-instance-multi-client--solving-only-one-puzzle-x}, the costs of \srv for a client $\prtt_{\st \prt}$ are as follows. In step \ref{multi-instance-multi-client-::find-base-x},   \srv performs $O(\mxsqr\cdot \sum\limits_{\st j=1}^{\st z}\bar{\Delta}_{\st j, \prt})$ modular squaring over $N$ to find the master keys $mk_{\st 1, \prt},\ldots, mk_{\st z, \prt}$. In steps \ref{multi-instance-multi-client-::find-base-x} and \ref{multi-instance-multi-client-::Removing-blinding-factors-server-side-x}, in total, it invokes $z\cdot (3+\cor)-1$ instances of \prf. In step \ref{multi-instance-multi-client-::Removing-blinding-factors-server-side-x}, it performs $\cor$ addition and $\cor$ multiplication. In step \ref{multi-instance-multi-client-::Extracting a polynomial-server-side-x}, it interpolates a polynomial using $\cor$ coordinates, involving $O(\tl)$ addition and multiplication. Hence, the complexity of \srv in \ref{multi-instance-multi-client--solving-only-one-puzzle-x} is $O(\tl+\mxsqr\cdot \sum\limits_{\st j=1}^{\st z}\bar{\Delta}_{\st j, \prt})$.

Next, we move on to the communication costs of \srv. In Phase \ref{multi-instance-multi-client-::Setup-server}, \srv  publishes $\cor$ messages. 
In Phase \ref{multi-client-multi-instance-phase::multi-instance-TF-Linear-Combination}, it invokes $\cor\cdot z$  instances of $\ole^{\st +}$, where each instance imposes $O(1)$ communication cost. In the same phase, it publishes $\cor$ encrypted $y$-coordinates. In Phase \ref{multi-client-multi-instance-tf-LinearCombination-phase-x}, it invokes $\cor\cdot n$  instances of $\ole^{\st +}$ and also publishes  $\cor$ encrypted $y$-coordinates. 
In  Phase \ref{multi-instance-multi-client-::Solving-a-Puzzle},  \ref{multi-instance-multi-client-solving-linear-combination-x}, it publishes $3$ messages. 
In  Phase \ref{multi-instance-multi-client-::Solving-a-Puzzle}, \ref{multi-instance-multi-client-solving-linear-combination}, \srv publishes $2\cdot \tl+1$ messages. 
In Phase \ref{multi-instance-multi-client--solving-only-one-puzzle-x}, \srv publishes $3\cdot z$ messages. 
Thus, the total communication complexity of \srv is $O(\tl\cdot n+z)$. 

\subsubsection{The Scheme Proposed in \cite{tempora-fusion}.}\label{sec::tempora-fusion-cost} Initially, we consider a client's costs in this multi-client scheme.

\noindent\underline{\textit{Client's Costs.}} 
During the Puzzle Generation phase, a client performs two modular exponentiations, one over $\phi(N)$ and another over $N$. Within the same phase, it invokes \prf and performs modular addition and multiplication linear with the number of leaders $\tl$. During the Linear Combination phase, it performs modular arithmetics,  invocations of \prf, and executions of $\ole^{\st +}$ linearly with $\tl$. Thus, the client's overall computation complexity is $O(\tl)$. The communication cost of the client is $O(\tl\cdot n)$ as it transmits to each client $\tl$ encrypted $y$-coordinates of a random root.

\noindent\underline{\textit{Verifier's Costs.}} During the verification of the result of the linear combination it (a) invokes $O(\tl)$ instances of the hash function, (b)  invokes $O(\tl^{\st 2})$ instances of \prf, and (c) performs $O(\tl^{\st 2})$ addition and multiplication.  Its cost during the verification of a solution related to a client's single puzzle is $O(1)$ as it involves a single invocation of a hash function.  
Hence, when it verifies $z$ puzzles of a client and a linear combination of $n$ clients' messages, the verifier's computation complexity is $O(\tl^{\st 2}+z)$. 

\noindent\textit{\underline{Server's Computation Cost.}}
During computing the linear combination of clients' puzzles,  \srv invokes $O(\tl)$ instances of   $\ole^{\st +}$ with each client. 
In the same phase, it performs $O(\tl\cdot n)$ modular addition. 
During the Solving  Puzzles phase, when it needs to deal with puzzles related to the linear combination, \srv performs $O(\tl\cdot Y)$ repeated modular squaring and invokes $O(\tl)$ instances of \prf. Within the same phase, it performs $O(\tl^{\st 2})$ addition and multiplication.  It also factorizes a polynomial, with the cost of  $O(\tl^{\st 2})$. 
Therefore, the computation complexity of \srv  in this case is $O(\tl^{\st 2}+\tl\cdot n+\tl\cdot Y)$.  During the Solving  Puzzles phase, when it needs to deal with a single puzzle of a client $\prtt_{\st \prt}$, server \srv performs $\mxsqr\cdot \Delta_{\st \prt}$ modular squaring to find master key $mk_{\st \prt}$. It invokes $O(\tl)$ instances of \prf and performs  $O(\tl)$ addition and multiplication. Therefore, the complexity of \srv in this case is $O(\tl+\mxsqr\cdot \Delta_{\st \prt})$. In the multi-instance case, where the client has $z$ puzzles where each puzzle $j$-th puzzle needs to be disclosed after period $\Delta_{\st j, \prt}$, \srv needs to deal with each of the puzzles separately, which leads to the total computation complexity of $O(\tl+\mxsqr\cdot \sum\limits_{\st j=1}^{\st z}\Delta_{\st j, \prt})$. The communication cost of \srv is dominated by $\ole^{\st +}$ invocations, which is linear with the total number of leaders and clients, i.e., $O(\tl\cdot n)$.

\subsubsection{TLP in \cite{MalavoltaT19}.}\label{sec::LH-Malavolta}

The homomorphic linear combination TLP proposed in \cite[p.634]{MalavoltaT19}, requires a trusted setup involving a trusted party.

\noindent\underline{\textit{Trusted Party's Costs.}} 
 In the Setup phase, it computes a set of private and public parameters and publishes the public ones.  In this phase, the trusted party, only once, performs a modular squaring over $\phi(N)$. Thus, this party's computation cost is $O(1)$. The trusted party's communication complexity is also $O(1)$, as it only publishes $4$ values. 
 
\noindent\underline{\textit{Client's Costs.}} 
 In the Puzzle Generation phase, a client performs $3$ modular exponentiations, one over $N$ and the other two over $N^{\st 2}$. Thus, the computation complexity of the client is $O(1)$, with respect to $n$ which is the total number of clients involved.  The client's communication complexity is $O(1)$.

\noindent\underline{\textit{Server's Costs.}} 
 To solve a puzzle (related to a single client's puzzle or a puzzle encoding a linear combination of solutions), a server performs $\mxsqr\cdot \Delta$ repeated modular squaring, similar to conventional TLPs. To compute a homomorphic linear combination of puzzles, the server performs  $n$ modular multiplication over $N$ and $n$ modular multiplication over $N^{\st 2}$. Therefore, the computation complexity of the server is $O(n)$. The server's communication complexity is $O(1)$. In the multi-instance case, where a client $\prtt_{\st \prt}$  has $z$ puzzles, where each $j$-th puzzle must be found after period $\Delta_{\st j, \prt}$, the server needs to deal with each puzzle independently, leading to the additional computation complexity of $O(\mxsqr\cdot \sum\limits_{\st j=1}^{\st z}\Delta_{\st j, \prt})$.

\subsubsection{TLP in \cite{liu2022towards}.}\label{sec::cost-esorics-2022} The additive TLP proposed in \cite{liu2022towards} heavily relies on the above additive TLP of  Malavolta and Thyagarajan \cite{MalavoltaT19}. As a result, the overall complexities of the client and server in this scheme are similar to that of the additive TLP in \cite{MalavoltaT19} with a main difference. Namely,  the server in this TLP needs to perform $O(\frac{\mxsqr\cdot \Delta}{\log(\mxsqr\cdot \Delta)})$ group operations to generate a proof. To check $z$ puzzles of a client, a verifier's complexity is $O(z)$. This scheme also does not provide any mechanism to efficiently handle the multi-instance setting, imposing additional computation complexity of $O(\mxsqr\cdot \sum\limits_{\st j=1}^{\st z}\Delta_{\st j, \prt})$ on the server, when each client has $z$ puzzles. The communication complexity for the parties in this TLP is comparable to that described in \cite{MalavoltaT19}.

\subsubsection{TLP in \cite{dujmovic2023time}.}\label{sec::cost-dujmovic2023time} This TLP is also built upon the additive TLP introduced by  Malavolta and Thyagarajan \cite{MalavoltaT19}. Consequently, the computation complexity for a trusted party during the Setup phase and for each client during the Puzzle Generation phase is comparable to that in the TLP of  Malavolta and Thyagarajan. However, this TLP requires the server to perform $O(n^{\st 2}+\mxsqr\cdot\Delta)$ operations to combine the puzzles and solve the combined puzzle.  Since this scheme cannot efficiently handle the muti-instance setting, the server must deal with each puzzle independently, yielding additional computation complexity of $O(\mxsqr\cdot \sum\limits_{\st j=1}^{\st z}\Delta_{\st j, \prt})$. The parties' communication complexity in this TLP is similar to that in \cite{MalavoltaT19}.


\subsection{Features}

\subsubsection{\mhtlp.}\label{sec::cost-of-mhtlp-features} This scheme can efficiently handle the multiple-instance setting and does not require a trusted setup. It enables anyone to efficiently check the correctness of a solution for a client's puzzle and a linear combination of puzzles. This scheme allows different clients to have different time parameters for their puzzles.

\subsubsection{\mmhtlp.}\label{sec::cost-of-mmhtlp-features} This TLP supports multi-client as well as efficiently handling the multiple-instance setting.  This scheme also does not require a trusted setup. It also allows anyone to efficiently verify the correctness of a solution for a client's puzzle and a linear combination of puzzles. It allows different clients to have different time parameters for their puzzles.  

\subsubsection{TLP in \cite{tempora-fusion}.}\label{sec::tf-feature} This scheme supports multi-client and does not require a trusted setup. It also supports efficient verification of a solution for a client's puzzle and a linear combination of puzzles. Similar to the above two schemes, it is flexible regarding the time parameters of different puzzles.

\subsubsection{TLP in \cite{MalavoltaT19}.}\label{sec::LH-Malavolta-feature}  The original homomorphic linear combination proposed in  \cite[p.634]{MalavoltaT19} requires all time parameters to be identical. This constraint limits its applicability, as different clients may prefer their solutions to be disclosed at different times. To address this, the authors suggested an extension that involves a trusted third party releasing a set of public parameters, each corresponding to a different time parameter during the setup phase. However, this solution also restricts clients' flexibility because they must choose from only the time parameters initially generated by the trusted third party. The scheme supports multiple clients, however, it does not support the multi-instance setting. It does not provide any verification mechanism to allow a verifier to check the solution that the server finds, which contributes to its overall lower cost compared to those that support verification.

\subsubsection{TLP in \cite{liu2022towards}.}\label{sec::esorics-2022-feature}  One of the TLPs introduced in  \cite{liu2022towards} supports multi-client and homomorphic linear combinations. It allows a server to prove the validity of a solution for a single client's puzzle, by relying on computationally expensive public-key-based primitives. However, this scheme does not support verifying the correctness of the linear combinations. This scheme also requires the involvement of a trusted party to generate a set of public and private parameters. This scheme does not support any efficient solution for the multi-instance setting. This scheme also lacks flexibility regarding the time parameter, as it assumes all clients use an identical time parameter.



\subsubsection{TLP in \cite{dujmovic2023time}.}\label{sec::feature-dujmovic2023time} 

This scheme supports multi-client and allows a server to check whether a puzzle has been created correctly. However, it does not offer any solution for the verification of a solution related to a single puzzle or homomorphic linear combinations. This scheme also requires a trusted party and lacks flexibility with respect to the time parameter, as it presumes that all clients use the same time parameter. It offers the batch-solving feature, that enables the server to combine $n$ puzzles into a single combined puzzle, such that after solving this puzzle, the server can find the solution to each puzzle that was integrated into the combined puzzle.

\subsection{Comparison} 

\subsubsection{Cost.} The overall computation and communication complexity of all schemes, except the one in \cite{dujmovic2023time}, is linear with the number of puzzles $z$ and the number of clients $n$.  
However, the computation complexity of the TLP in \cite{dujmovic2023time} is quadradic regarding $n$. Note that the complexities of \mhtlp, \mmhtlp, and the TLP in \cite{tempora-fusion} are quadratic regarding the total number of leader clients $\tl$. However, $\tl$ can be set to a small value, e.g., between $3$ and $10$, depending on the setting and security assumption. Moreover, only \mhtlp, \mmhtlp, and the TLP in \cite{tempora-fusion} can efficiently deal with the multi-instance setting. 

\subsubsection{Feature.} Among the six schemes, only  \mmhtlp provides both multi-client and multi-instance capabilities. It stands out as the scheme that offers the most features. Additionally,  \mhtlp, \mmhtlp, and the TLP proposed in \cite{tempora-fusion} (a) do not require a trusted setup, (b) support verification of both the solution to a client’s puzzle and the solution to a linear combination of puzzles, and (c) allow flexible time parameters. Conversely, only the TLP in \cite{dujmovic2023time} supports batch verification. 

\vspace{-2mm}

\subsection{An Overview of Concrete Cost}

The three main operations that impose non-negligible costs to the participants of our schemes are polynomial factorization, invocations of \prf, and $\ole^{\st +}$ execution. In our schemes, the computation complexity of a verifier is quadratic with the number of leaders $\tl$, which determines the degree of the polynomial to be factorized. The runtime of polynomial factorization is also influenced by the field size, $\log_{\st 2}(\prm)$. As shown in  \cite{tempora-fusion}, the total combined computation cost imposed due to factorization and \prf invocations is about $6$ milliseconds when $\tl=10$ and $\log_{\st 2}(\prm)=256$-bit. The running time of $\ole^{\st +}$ is low as well, for instance about $1$ second for $2^{\st 14}$ input elements, as shown in \cite{SchoppmannGR019}. Thus, we estimate our schemes will impose an additional cost of about 10 seconds when the total number of clients is $20$. This estimation excludes the standard cost of solving puzzles.

%% file: conclusion.tex

\vspace{-2mm}
\section{Conclusion and Future Work}

Time-Lock Puzzles (TLPs) have been developed to securely transmit private information into the future without relying on a third party. They have applications in various domains, including transparent scheduled payments in private banking, e-voting, and secure aggregation in federated learning. To enhance the scalability of TLPs, multi-instance TLPs have been designed, enabling a server to efficiently handle multiple instances of a client's puzzles. Separately, homomorphic TLPs have been developed to allow (verifiable) computation on the puzzles of different clients.

In this work, we proposed two schemes \mhtlp and \mmhtlp to bridge these two research lines. Initially, we proposed Multi-instance verifiable partially Homomorphic TLP  (\mhtlp), the first multi-instance TLP that supports efficient verifiable homomorphic linear combinations on puzzles. It enables a client to generate many puzzles and transmit them to the server at once. In this setting, the server does not need to simultaneously deal with them; instead, it can solve them one after the other. \mhtlp enables the server to learn the linear combination of the puzzles' solutions after a certain time. It allows public verification of a single puzzle's solution and the computation's result.

Next, we introduced Multi-instance Multi-client verifiable partially Homomorphic TLP (\mmhtlp). This new variant combines the features of both (partially) homomorphic TLP and multi-instance TLP. It supports verifiable partially homomorphic operations on the puzzles belonging to single or multiple clients while maintaining the multi-instance feature. It enables single or multiple clients to ask the server to perform homomorphic linear combinations of their puzzles. This scheme allows anyone to verify whether the server has performed the computation correctly and provided a correct solution. 

We have conducted a thorough analysis of these two schemes. Our analysis indicates that the overall overhead of our schemes is linear with respect to the total number of clients and the number of puzzles. By comparing our solutions to the state-of-the-art TLPs, we observed that \mmhtlp offers a set of appealing features not simultaneously provided by any existing TLP.

Batch solving is an intriguing feature that allows a server to combine multiple puzzles into a single composite puzzle.  By solving this composite puzzle, one can determine the solution to each individual puzzle involved \cite{dujmovic2023time}. It will be interesting to explore how \mmhtlp can be enhanced to offer this property while maintaining its current features and efficiency.

%% file: OLE-plus.tex

\section{Enhanced OLE's Ideal Functionality and Protocol}\label{apndx:F-OLE-plus}

The enhanced \ole ensures that the receiver cannot learn anything about the sender's inputs,  when it sets its input to $0$, i.e., $c=0$. The enhanced \ole's protocol (denoted by $\ole^{\st +}$) is presented in Figure \ref{fig:OLE-plus-protocol}. 

\begin{figure}[ht]
\setlength{\fboxsep}{1pt}
\begin{center}
\begin{boxedminipage}{12.3cm}
\begin{small}
\begin{enumerate}
\item  Receiver (input $c \in \mathbb{F} $): Pick a random value, $r\stackrel{\st\$}\leftarrow  \mathbb{F} $, and send  $(\mathtt{inputS}, (c^{\st -1}, r))$ to the first $\mathcal{F}_{\st\ole}$.
%
%
\item Sender (input $a, b \in \mathbb{F} $): Pick a random value, $u \stackrel{\st\$}\leftarrow  \mathbb{F} $, and send $(\mathtt{inputR}, u)$ to the first $\mathcal{F}_{\st\ole}$, to learn $t =  c^{\st -1}\cdot u
 + r$. Send $(\mathtt{inputS},(t + a, b - u))$ to the second $\mathcal{F}_{\st\ole}$.
\item Receiver: Send $(\mathtt{inputR}, c)$ to the second $\mathcal{F}_{\st\ole}$ and obtain $k = (t+a)\cdot c+(b-u)=a\cdot c + b + r\cdot c$. Output $s=k - r\cdot c=a\cdot c + b$.

\end{enumerate}
\end{small}
\end{boxedminipage}
\end{center}
\caption{
\small {Enhanced Oblivious Linear function Evaluation  ($\ole^{\st +}$)  \cite{GhoshN19}}.} 
\label{fig:OLE-plus-protocol}
\end{figure}

%% file: V-TLP.tex

\section{The \tf Protocol}\label{sec::tf-protocol}

In this section, we present \tf, initially introduced in \cite{tempora-fusion}

\begin{enumerate}

\item\label{multi-client::Setup-server}  {  {Setup}}. $\ssetup(1^{\st \lambda}, \tl, t)\rightarrow (., pk_{\st \srv})$

 The server \srv only once takes the following steps:

\begin{enumerate}

\item   generates a sufficiently large prime number $\prm$, where $\log_{\st 2}(\prm)$ is a security parameter, e.g.,  $\log_{\st 2}(\prm)\geq 128$.

\item   let $\tl$ be the total number of leader clients. It sets $\cor=\tl+2$ and $\vec{x}=[x_{\st 1}, \ldots, x_{\st \cor}]$, where $x_{\st i}\neq x_{\st j}$, $x_{\st i}\neq 0$, and $x_{\st i}\notin U$.

\item publishes $pk_{\st \srv}=(\prm, \vec{x}, t)$. 

\end{enumerate}

\item\label{multi-client::key-gen}  {  {Key Generation}}. $\csetup(1^{\st \lambda})\rightarrow \kp_{\st \prt}$

Each party $\prtt_{\st\prt}$ in $C=\set$ takes the following steps:

\begin{enumerate}


\item  computes $N_{\st \prt}=\prm_{\st 1}\cdot \prm_{\st 2}$, where $\prm_{\st i}$  is a large randomly chosen prime number, where $\log_{\st 2}(\prm_{\st i})$ is a security parameter. It computes Euler's totient function of $N_{\st \prt}$, as: $\phi(N_{\st \prt})=( \prm_{\st 1}-1)\cdot( \prm_{\st 2}-1)$.



\item\label{multi-client::Publishing-public-parameters}   stores secret key $sk_{_{\st\prt}}=\pnp$ and publishes public key $pk_{\st\prt}=N_{\st \prt}$. 
\end{enumerate}

\item\label{tf-Puzzle-Generation-phase}  {  {Puzzle Generation}}. $\pgen(m_{\st \prt}, \kp_{\st \prt}, pk_{\st \srv}, \Delta_{\st \prt}, \mxsqr)\rightarrow(\vec{o}_{\st \prt}, prm_{\st \prt})$

Each $\prtt_{\st\prt}$ independently takes the following steps to generate a puzzle for a message $m_{\st \prt}$.
\begin{enumerate}

\item  checks the bit-size of $p$ and  elements of $\vec{x}$ in $pk_{\st \srv}$, to ensure $\log_{\st 2}(p)\geq 128$, $x_{\st i}\neq x_{\st j}, x_{\st i}\neq 0$, and $x_{\st i}\notin U$. If it does not accept the parameters, it returns $(\bot, \bot)$ and does not take further action. 

\item\label{multi-client:Generating-secret-keys} generates a master key $mk_{\st \subprt}$ and two secret keys $k_{\st \prt}$ and  $s_{\st \prt}$ as follows:

\begin{enumerate}

\item\label{multi-client:set-expo} sets exponent $a_{\st \prt}$ as:  
%
$a_{\st \prt}=2^{\st T_{\st\subprt}}\bmod \pnp$.

where $T_{\st\prt}=\Delta_{\st \prt}\cdot \mxsqr$ and $\pnp \in \kp_{\st \prt}$.

\item\label{multi-client:set-master-key}  selects a base uniformly at random: $r_{\st \prt} \stackrel{\st\$}\leftarrow\mathbb{Z}_{\st N_{\st \subprt}}$ and then sets a master key $mk_{\st \subprt}$ as follows: 
$$mk_{\st \prt}= r^{\st a_{{\st \subprt}}}_{\st \prt}\bmod N_{\st \prt}$$

\item\label{multi-client:derive-keys-for-puzzle} derive two keys from $mk_{\st \prt}$ as: 
$ k_{\st \prt}=\prf(1, mk_{\st \prt}), \hspace{4mm} s_{\st \prt}=\prf(2, mk_{\st \prt})$.

\end{enumerate}
%

\item\label{multi-client:derive-PR-values}   generates $2\cdot \cor$ pseudorandom blinding factors using $k_{\st \prt}$ and  $s_{\st \prt}$:
$$\forall i, 1\leq i \leq \cor: \hspace{4mm}  z_{\st i, \prt}=\prf(i, k_{\st \prt}), \hspace{4mm} w_{\st i,\prt}=\prf(i, s_{\st \prt})$$

\item {{encodes plaintext message as follows}}:  

\begin{enumerate}
\item\label{multi-client:rep-message-as-poly} represents plaintext message $m_{\st \prt}$ as a polynomial, such that the polynomial's constant term is the message. Specifically, it computes polynomial $\bm{\pi}_{\st\prt}(x)$ as: 
$\bm{\pi}_{\st\prt}(x) =x+m_{\st \prt} \bmod \prm$. 

\item\label{multi-client:gen-y-coord-of-poly} computes \cor $y$-coordinates of  $\bm{\pi}_{\st\prt}(x)$ as: 
$\forall i, 1\leq i \leq \cor: \hspace{4mm}\pi_{\st i,\prt}= \bm{\pi}_{\st\prt}(x_{\st i}) \bmod \prm$,  
where $x_{\st i}\in \vec{x}$ and $p \in pk_{\st \srv}$.

\end{enumerate}

\item\label{multi-client:enc-y-coord-of-poly}   encrypts the $y$-coordinates using the blinding factors as follows: 
$$\forall i, 1\leq i \leq \cor: \hspace{4mm} o_{\st i,\prt} = w_{\st i,\prt}\cdot(\pi_{\st i,\prt} +  z_{\st i,\prt}) \bmod \prm$$

\item\label{multi-client:commit-y-coord-of-poly} commits to the plaintext message:   
$ com_{\st \prt} = \comcom(m_{\st \prt}, mk_{\st \prt})$. 

\item\label{tf-Managing-messages}  publishes  $\vec{o}_{\st \prt} = [o_{\st 1,\prt},\ldots, o_{\st \cor,\prt}]$  and  $pp_{\st \prt}=(com_{\st \prt}, T_{\st \prt}, r_{\st \prt}, N_{\st \prt})$. It locally keeps secret parameters $sp_{\st \prt}=(k_{\st \prt}, s_{\st \prt})$ and deletes everything else, including $m_{\st \prt}, \bm{\pi}_{\st\prt}(x), \pi_{\st 1,\prt},\ldots, \pi_{\st \cor,\prt}$. It sets $prm_{\st {\st \prt}}=(sp_{\st \prt}, pp_{\st \prt})$.

\end{enumerate}

\item\label{tf-LinearCombination-phase}  {  {Linear Combination}}. $\eval(\langle \srv(\vec{o}, \Delta, $ $ \mxsqr, \vec{pp}, \vec{pk}, pk_{\st \srv}), \prtt_{\st 1}(\Delta, \mxsqr, \kp_{\st {\st 1}}, prm_{\st {\st 1}}, q_{\st 1}, pk_{\st \srv}), \ldots, \prtt_{\st n}(\Delta, $\\ $\mxsqr, $ $\kp_{\st {\st n}}, prm_{\st {\st n}}, q_{\st n},$  $pk_{\st \srv}) \rangle)\rightarrow(\vec{g}, \vec{pp}^{\st(\text{Evl})})$

In this phase, the parties produce certain messages that allow \srv to find a linear combination of the clients' plaintext messages after time $\Delta$.

\begin{enumerate}


\item\label{multi-client:Randomly-selecting-leaders}  all parties in $C$ agree on a random key $\hat{r}$, e.g., by participating in a coin tossing protocol \cite{Blum82}. Each  $\prtt_{\st\prt}$ deterministically finds index of $\tl$ leader clients: $\forall j, 1\leq j\leq \tl: idx_{\st j}=\g(j||\hat{r})$. Let \idx be a vector contain these $\tl$ clients.


\item\label{multi-client::Granting-the-computation}  each leader client $\prtt_{\st\prt}$  in \idx takes the following steps. 



\begin{enumerate}

\item\label{multi-client:Generating-temporary-secret-keys}  generates a temporary master key $tk_{\st\prt}$ and two secret keys $k'_{\st \prt}$ and $s'_{\st \prt}$ for itself. Also, it generates a secret key $f_{\st {\st l}}$ for each client. To do that, it takes the following steps. It computes the exponent:  
$b_{\st \prt}=2^{\st Y}\bmod \pnp$.

where $Y=\Delta\cdot \mxsqr$.
 It selects a  base uniformly at random: $h_{\st \prt}\stackrel{\st\$}\leftarrow\mathbb{Z}_{\st N_{\st \subprt}}$ and then sets a temporary master key $tk_{\st\prt}$ as: 
$tk_{\st\prt} = h_{\st \prt}^{\st b_{\st\subprt}} \bmod N_{\st \prt}$.

It derives two keys from $tk_{\st \prt}$ as: 
$ k'_{\st \prt}=\prf(1, tk_{\st \prt}), \hspace{4mm} s'_{\st \prt}=\prf(2, tk_{\st \prt})$.

It picks a random key $f_{\st {\st l}}$ for each client $\cl_{\st l}$ excluding itself, i.e.,  $f_{\st {\st l}}\stackrel{\st\$}\leftarrow\{0, 1\}^{\st poly(\lambda)}$, where $\cl_{\st l}\in C\setminus\prtt_{\st\prt}$. It sends $f_{\st {\st l}}$ to each $\cl_{\st l}$.

\item\label{multi-client:Generating-temporary-pr-vals}  derives \cor pseudorandom values from $s'_{\st \prt}$: 
 $$\forall i, 1\leq i \leq \cor:  \hspace{4mm}  w'_{\st i,\prt}=\prf(i, s'_{\st \prt})$$



\item\label{multi-client:Generating-temporary-root}  picks a random root: $\rt_{\st \prt}\stackrel{\st\$}\leftarrow\mathbb{F}_{\st \prm}$. 
It represents $\rt_{\st \prt}$ as a polynomial, such that the polynomial's root is $\rt_{\st \prt}$. Specifically, it computes polynomial $\bm{\gamma}_{\st \prt}(x)$ as: 
 $\bm{\gamma}_{\st \prt}(x) =x- \rt_{\st \prt} \bmod \prm$.
 
 Then, it computes \cor $y$-coordinates of  $\bm{\gamma}_{\st \prt}(x)$ as: 
$\forall i, 1\leq i \leq \cor: \hspace{4mm}\gamma_{\st i, {\st \prt}} = \bm{\gamma}_{\st \prt}(x_{\st i}) \bmod \prm$. 


It encrypts each $y$-coordinate   $\gamma_{\st i, _{\st \prt}}$ using blinding factor $w'_{\st i,\prt}$: 
$$\forall i, 1\leq i \leq \cor: \hspace{4mm}  \gamma'_{\st i, {\st \prt}}=\gamma_{\st i, {\st \prt}}\cdot w'_{\st i,\prt} \bmod \prm$$

It sends $\vv{ \gamma}'_{\st \prt}=[ \gamma'_{\st 1, {\st \prt}},\ldots,  \gamma'_{\st \cor, {\st \prt}}]$ to the rest of the clients.



\item\label{multi-client:Generating-blinding-factors}  receives $(\bar f_{\st {\st l}}, \vv{ \gamma}'_{\st{l}})$ from every other client which are in \idx. It regenerates its original blinding factors: 
 $$\forall i, 1\leq i \leq \cor: \hspace{4mm}   z_{\st i,\prt}=\prf(i, k_{\st \prt}), \hspace{4mm}   w_{\st i,\prt}=\prf(i, s_{\st \prt})$$

 where $k_{\st \prt}$ and $s_{\st \prt}$ are in $\vec{prm}_{\st \prt}$ and  were generated in step \ref{multi-client:derive-keys-for-puzzle}.  It also generates new ones: 
 $$\forall i, 1\leq i \leq \cor: \hspace{4mm}   z'_{\st i,\prt}=\prf(i, k'_{\st \prt})$$

It sets  values $v_{\st i, \prt}$ and $y_{\st i, \prt}$ as follows. $\forall i, 1\leq i \leq \cor:$
\begin{equation*}
\begin{split}
 v_{\st i, \prt} &=\gamma'_{\st i,\prt}\cdot  \prod\limits_{\st \forall \cl_{_{\st l}}\in \idx \setminus\prtt_{\st\prt}}
\gamma'_{\st i,{\st l}} \bmod \prm\\
y_{\st i, \prt}&=  -\sum\limits_{\st \forall \cl_{_{\st l}}\in C\setminus\prtt_{\st\prt}}\prf(i, f_{\st {l}}) +  \sum\limits_{\st \forall \cl_{_{\st l}}\in \idx\setminus\prtt_{\st\prt}}\prf(i, \bar{f}_{\st {l}})      \bmod 
\end{split}
\end{equation*}

where $\prtt_{\st\prt}\in \idx$.

\item\label{eval:ole-detect}  
obliviously, without having to access a plaintext solution, prepares the puzzle (held by \srv) for the computation. To do that, it participates in an instance of $\ole^{\st +}$ with \srv, for every $i$, where $1\leq i \leq \cor$.  The inputs of $\prtt_{\st\prt}$ to $i$-th instance of $\ole^{\st +}$ are: 
\begin{equation*}
\begin{split}
e_{\st i} &= q_{\st \prt}\cdot v_{\st i, \prt}\cdot (w_{\st i,\prt})^{\st -1}\bmod \prm\\
 e'_{\st i} &= -(q_{\st \prt}\cdot v_{\st i, \prt}\cdot  z_{\st i,\prt}) + z'_{\st i,\prt} + y_{\st i, \prt}\bmod \prm
\end{split}
\end{equation*}

The input of \srv to the $i$-th instance of $\ole^{\st +}$ is $\prtt_{\st\prt}$'s encrypted $y$-coordinate: $e''_{\st i} = o_{\st i,\prt}$ (where $o_{\st i,\prt}\in \vec{o}$). Accordingly, $i$-th instance of $\ole^{\st +}$ returns to \srv:
\begin{equation*}
\begin{split}
d_{\st i, \prt}& = e_{\st i}\cdot e''_{\st i} + e'_{\st i}\\ &= q_{\st \prt}\cdot v_{\st i, \prt}\cdot  \pi_{\st i,\prt}+z'_{\st i, \prt}+y_{\st i, \prt}\bmod \prm\\ &= q_{\st \prt}\cdot \gamma_{\st i, {\st \prt}}\cdot w'_{\st i, \prt}
\cdot (\prod\limits_{\st \forall \cl_{_{\st l}}\in \idx\setminus\prtt_{\st\prt}}
\gamma_{\st i,{\st l}}\cdot w'_{\st i, {\st l}})
\cdot  \pi_{\st i,\prt}+z'_{\st i,\prt}+y_{\st i, \prt}\bmod \prm
\end{split}
\end{equation*}


where $q_{\st \prt}$ is the party's coefficient. If $\prtt_{\st\prt}$ detects misbehavior during the execution of  $\ole^{\st +}$, it sends a special symbol $\bot$ to all parties and halts.

\item\label{multi-client:Committing-to-the-root}  computes $com'_{\st \prt}=\comcom(\rt_{\st \prt}, tk_{\st \prt})$.

\item\label{multi-client::Publishing-public-parameters--}   publishes ${pp}_{\st \prt}^{\st(\text{Evl})}=(h_{\st \prt}, com'_{\st \prt}, Y)$.  Note that all $\prtt_{\st\prt} \in \idx$  use identical $Y$.  Let $\vec{pp}^{\st(\text{Evl})}$ contain all the  triples ${pp}_{\st \prt}^{\st(\text{Evl})}$ published by $\prtt_{\st\prt}$, where $\prtt_{\st\prt} \in \idx$.


\end{enumerate}

\item %
each non-leader client $\prtt_{\st\prt}$ takes the following steps.


\begin{enumerate}
\item  receives $(\bar{f}_{\st {l}}, \vv{ \gamma}'_{\st {l}})$  from every other client which is in $\idx$. 
It regenerates its original blinding factors: 
 $$\forall i, 1\leq i \leq \cor: \hspace{4mm}   z_{\st i,\prt}=\prf(i, k_{\st \prt}), \hspace{4mm}   w_{\st i,\prt}=\prf(i, s_{\st \prt})$$

It set  values $v_{\st i, \prt}$ and $y_{\st i, \prt}$ as follows. $\forall i, 1\leq i \leq \cor:$
\begin{equation*}
\begin{split}
 v_{\st i,\prt} &= \prod\limits_{\st \forall \cl_{_{\st l}}\in \idx}
\gamma'_{\st i,{\st l}} \bmod \prm\\
y_{\st i, \prt}&=    \sum\limits_{\st \forall \cl_{_{\st l}}\in \idx}\prf(i, \bar{f}_{\st {l}})         \bmod \prm
\end{split}
\end{equation*}

\item\label{eval:ole-detect-2}  participates in an instance of $\ole^{\st +}$ with the server \srv, for every $i$, where $1\leq i \leq \cor$.  The inputs of $\prtt_{\st\prt}$ to $i$-th instance of $\ole^{\st +}$ are: 
\begin{equation*}
\begin{split}
e_{\st i} &= q_{\st \prt}\cdot  v_{\st i, \prt}\cdot   (w_{\st i,\prt})^{\st -1}\bmod \prm\\
e'_{\st i} &= -(q_{\st \prt}\cdot v_{\st i, \prt}\cdot    z_{\st i,\prt}) + y_{\st i, \prt}\bmod \prm
\end{split}
\end{equation*}

The input of \srv to $i$-th instance of $\ole^{\st +}$ is $\prtt_{\st\prt}$'s encrypted $y$-coordinate: $e''_{\st i} = o_{\st i,\prt}$. Accordingly, $i$-th instance of $\ole^{\st +}$ returns to \srv:
\begin{equation*}
\begin{split}
 d_{\st i, \prt}& = e_{\st i} \cdot e''_{\st i} +  e'_{\st i}\\ 
 &=  q_{\st \prt}\cdot v_{\st i, \prt} \cdot \pi_{\st i,\prt} + y_{\st i, \prt}\bmod \prm\\ 
 &= q_{\st \prt}\cdot (\prod\limits_{\st \forall \cl_{_{\st l}}\in \idx\setminus\prtt_{\st\prt}}
\gamma_{\st i, {\st l}}\cdot w'_{\st i, {\st l}})
\cdot  \pi_{\st i,\prt}+y_{\st i, \prt}\bmod \prm
\end{split}
\end{equation*}

where $q_{\st \prt}$ is the party's coefficient. If $\prtt_{\st\prt}$ detects misbehavior during the execution of  $\ole^{\st +}$, it sends a special symbol $\bot$ to all parties and halts.



\end{enumerate}

\item\label{multi-client::Computing-encrypted-linear-combination} server \srv sums all of the outputs of $\ole^{\st +}$ instances that it has invoked, $\forall i, 1\leq i \leq \cor:$ 
\begin{equation*}
\begin{split}
g_{\st i} &= \sum\limits_{\st \forall \prtt_{\st\prt}\in C} d_{\st {i, \prt}} \bmod \prm\\ &=    (
\prod\limits_{\st \forall \prtt_{\st\prt}\in  \idx} {\gamma_{\st i, \prt}\cdot w'_{\st i, \prt}}
\cdot \sum\limits_{\st \forall \prtt_{\st\prt}\in C}   q_{\st \prt}\cdot \pi_{\st i,\prt} )+    \sum\limits_{\st \forall \prtt_{\st\prt}\in \idx}  z'_{\st i,\prt}  \bmod  \prm
\end{split}
\end{equation*}

Note that in  $g_{\st  i, j}$ does not exist any $y_{\st  i, j}$, because  $y_{\st  i, j}$ in different $d_{\st i, j}$ canceled out each other after they summed up. 

\item\label{publish-encrypted-solutions}  server \srv publishes $\vec{g}=[g_{\st 1}, \ldots, g_{\st \cor}]$.

\end{enumerate}

\item\label{multi-client::Solving-a-Puzzle}  {  {Solving a Puzzle}}. $\solv(\vec{o}_{\st \prt}, pp_{\st \prt}, \vec{g}, \vec{pp}^{\st (\text{Evl})}, \vec{pk}, pk_{\st \srv},\cmd)\rightarrow(m, \zeta)$

Server \srv takes the following steps. 



\begin{steps}[leftmargin=15mm]

\item\label{solving-linear-combination}\hspace{-2mm}. when solving a puzzle related to the linear combination, i.e., when $\cmd=\ep$: 

\begin{enumerate}

\item\label{test}  for  each $\prtt_{\st\prt}\in \idx$:

\begin{enumerate}
\item\label{extract-temp-key} finds $tk_{\st \prt}$ where  $tk_{\st \prt}=h^{\st 2^{\st Y}}_{\st \prt}\bmod N_{\st \prt}$ through repeated squaring of $h_{\st \prt}$ modulo $N_{\st \prt}$, where $(h_{\st \prt}, Y)\in \vec{pp}^{\st (\text{Evl})}$ and $N_{\st \prt}\in \vec{pk}$.

\item derives two keys from $tk_{\st \prt}$ as: 
$ k'_{\st \prt}=\prf(1, tk_{\st \prt}), \hspace{4mm} s'_{\st \prt}=\prf(2, tk_{\st \prt})$. 

\end{enumerate}

\item\label{multi-client::server-side-computatoin-removing-pr-vals}  removes the blinding factors from $[g_{\st 1}, \ldots, g_{\st \cor}]\in \vec{g}$. 

$\forall i, 1\leq i \leq \cor:$
\begin{equation*}
\begin{split}
\theta_{\st i}&=\big(\prod\limits_{\st \forall \prtt_{\st\prt}\in \idx}{\prf(i, s'_{\st \prt})}\big)^{\st -1}\cdot\big(g_{\st i}-\sum\limits_{\st \forall \prtt_{\st\prt}\in \idx} {\prf(i, k'_{\st \prt}})\big) \bmod \prm\\ &=
(\prod\limits_{\st \forall \prtt_{\st\prt}\in\idx} \gamma_{\st i, \prt}) \cdot \sum\limits_{\st \forall \prtt_{\st\prt}\in C}   q_{\st \prt}\cdot \pi_{\st i,\prt} \bmod \prm
\end{split}
\end{equation*}

\item\label{step::multi-client-interpolate-poly} interpolates a polynomial $\bm{\theta}$, given pairs $(x_{\st 1}, \theta_{\st 1}),\ldots, (x_{\st \cor}, \theta_{\st \cor})$.  Note that $\bm{\theta}$ will have the following form: 
\begin{equation*}
\bm\theta(x) =\prod\limits_{\st \forall \prtt_{\st\prt}\in  \idx} (x-\rt_{\st \prt})\cdot \sum\limits_{\st \forall \prtt_{\st\prt}\in C}q_{\st \prt}\cdot (x+m_{\st \prt}) \bmod \prm
\end{equation*}
%

We can rewrite $\bm\theta(x)$ as follows: 
$$\bm\theta(x) = \bm\psi(x)+\prod\limits_{\st \forall \prtt_{\st\prt}\in  \idx}(-\rt_{\st \prt})\cdot \sum\limits_{\st \forall\prtt_{\st\prt}\in C}q_{\st \prt}\cdot m_{\st \prt}   \bmod \prm$$

where $\bm\psi(x)$ is a polynomial of degree $\tl+1$ whose constant term is $0$.

\item\label{multi-client::server-side-Extracting-the-linear-combination}   retrieves the final result (which is the linear combination of the messages $m_{\st 1},\ldots, m_{\st n}$)  from polynomial $\bm\theta(x)$'s constant term: $cons=\prod\limits_{\st \forall \prtt_{\st\prt}\in  \idx}(-\rt_{\st \prt})\cdot\sum\limits_{\st \forall\prtt_{\st\prt}\in C}q_{\st \prt}\cdot m_{\st \prt}$ as follows:
\begin{equation*}
\begin{split}
res &=cons\cdot (\prod\limits_{\st \forall \prtt_{\st\prt}\in  \idx}(-\rt_{\st \prt}))^{\st -1}\bmod \prm\\ &= \sum\limits_{\st \forall\prtt_{\st\prt}\in C}q_{\st \prt}\cdot m_{\st \prt}
\end{split}
\end{equation*}



\item\label{multi-client::Extracting-valid-roots}  extracts the roots of $\bm{\theta}$. Let set $R$ contain the extracted roots. It identifies the valid roots, by finding every  $\rt_{\st \prt}$ in $R$, such that $\comver(com'_{\st \prt},(\rt_{\st \prt}, tk_{\st \prt}))=1$. Note that \srv performs the check for every $\prtt_{\st\prt}$ in $\idx$. 

%

\item\label{publish-lc-proof}  publishes the solution $m=res$ and the proof  $\zeta=\big\{(\rt_{\st \prt}, tk_{\st \prt})\big\}_{\st \forall \prtt_{u}\in  \idx}$.


\end{enumerate}

\item\label{multi-client::verifying-a-solution-of-single-puzzle-}\hspace{-2mm}.  when solving a puzzle of single client $\prtt_{\st\prt}$,  i.e., when $\cmd=\scp$: 
\begin{enumerate}
\item\label{multi-client::Finding-secret-keys-}   finds  $mk_{\st \prt}$ where  $mk_{\st \prt}=r^{\st 2^{\st T_{\st \subprt}}}_{\st \prt}\bmod N_{\st \prt}$ through repeated squaring of $r_{\st \prt}$ modulo $N_{\st \prt}$, where $(T_{\st \prt}, r_{\st \prt})\in pp_{\st \prt}$. Then, it derives two keys from $mk_{\st \prt}$: 
$$ k_{\st \prt}=\prf(1, mk_{\st \prt}), \hspace{4mm} s_{\st \prt}=\prf(2, mk_{\st \prt})$$

%

\item     re-generates $2\cdot 
\cor$ pseudorandom values using $k_{\st \prt}$ and  $s_{\st \prt}$:
$$\forall i, 1\leq i \leq \cor: \hspace{4mm}  z_{\st i, \prt}=\prf(i, k_{\st \prt}), \hspace{4mm} w_{\st i,\prt}=\prf(i, s_{\st \prt})$$

Then, it uses the blinding factors to unblind $[o_{\st 1,\prt}, \ldots, o_{\st \cor,\prt}]$:
$$\forall i, 1\leq i \leq \cor: \hspace{4mm}  \pi_{\st i,\prt}  = \big((w_{\st i,\prt})^{\st -1}\cdot o_{\st i,\prt}\big) -z_{\st i,\prt} \bmod \prm$$


\item  interpolates a polynomial $\bm{\pi}_{\st \prt}$, given pairs $(x_{\st 1}, \pi_{\st 1,\prt}),\ldots, (x_{\st \cor}, \pi_{\st \cor,\prt})$.

\item\label{multi-client::Publishing-the-single-pizzle-solution}   considers the constant term of $\bm{\pi}_{\st \prt}$ as the plaintext solution, $m_{\st \prt}$. It publishes  the solution $m= m_{\st \prt}$  and the proof $\zeta=mk_{\st \prt}$.

\end{enumerate}

\end{steps}

\item\label{multi-client::verification}   {  {Verification}}. $\ver(m, \zeta, ., pp_{\st \prt}, \vec{g}, \vec{pp}^{\st (\text{Evl})}, pk_{\st \srv}, \cmd)\rightarrow \ddot{v}\in\{0,1\}$

A verifier (that can be anyone, not just $\prtt_{\st\prt}\in C$) takes the following steps. 


\begin{steps}[leftmargin=15mm]


\item\label{multi-client-verifying-computation-solution}\hspace{-2mm}. when verifying a solution related to the linear combination, i.e., when $\cmd=\ep$:

\begin{enumerate}

\item\label{step-multi-client-verify-opening}  verifies the validity of  every $(\rt_{\st \prt}, tk_{\st \prt})\in \zeta$, provided by \srv in \ref{solving-linear-combination}, step \ref{publish-lc-proof}: 
$\forall \prtt_{\st\prt}\in \idx:\hspace{4mm}  \comver\big(com'_{\st \prt},(\rt_{\st \prt}, tk_{\st \prt})\big)\stackrel{\st?}=1$, 
where $com'_{\st \prt}\in \vec{pp}^{\st (\text{Evl})}$. 
If all of the verifications pass, it proceeds to the next step. Otherwise, it returns $\ddot{v}=0$ and takes no further action.

\item\label{multi-client:Checking-resulting-polynomial-valid-roots}  checks if the resulting polynomial contains all the roots in $\zeta$, by taking the following steps. 

\begin{enumerate}

\item derives two keys from $tk_{\st \prt}$ as: 
$ k'_{\st \prt}=\prf(1, tk_{\st \prt}), \hspace{4mm} s'_{\st \prt}=\prf(2, tk_{\st \prt})$. 

\item\label{multi-client:verification-case-1removes the blinding factors}  removes the blinding factors from $[g_{\st 1}, \ldots, g_{\st \cor}]\in \vec{g}$ that were provided by \srv in step \ref{publish-encrypted-solutions}. 

$\forall i, 1\leq i \leq \cor:$
\begin{equation*}
\begin{split}
\theta_{\st i}&=\big(\prod\limits_{\st \forall \prtt_{\st\prt}\in \idx}\prf(i, s'_{\st \prt})\big)^{\st -1}\cdot\big(g_{\st i}-\sum\limits_{\st \forall \prtt_{\st\prt}\in\idx} \prf(i, k'_{\st \prt})\big) \bmod \prm\\ &=
\prod\limits_{\st \forall \prtt_{\st\prt}\in\idx} \gamma_{\st i, \prt} \cdot \sum\limits_{\st \forall \prtt_{\st\prt}\in C}   q_{\st \prt}\cdot \pi_{\st i,\prt} \bmod \prm
\end{split}
\end{equation*}

\item\label{multi-client:interpolate-poly}  interpolates a polynomial $\bm{\theta}$, given pairs $(x_{\st 1}, \theta_{\st 1}),\ldots, (x_{\st \cor}, \theta_{\st \cor})$, similar to step \ref{step::multi-client-interpolate-poly}.  This yields a polynomial $\bm{\theta}$ having the form: 
\begin{equation*}
\begin{split}
\bm\theta(x) &=\prod\limits_{\st \forall \prtt_{\st\prt}\in\idx} (x-\rt_{\st \prt})\cdot \sum\limits_{\st \forall \prtt_{\st\prt}\in C}q_{\st \prt}\cdot (x+m_{\st \prt}) \bmod \prm\\
&=\bm\psi(x)+\prod\limits_{\st \forall \prtt_{\st\prt}\in  \idx}(-\rt_{\st \prt})\cdot \sum\limits_{\st \forall\prtt_{\st\prt}\in C}q_{\st \prt}\cdot m_{\st \prt}   \bmod \prm
\end{split}
\end{equation*}

where $\bm\psi(x)$ is a polynomial of degree $\tl+1$ whose constant term is $0$.

\item\label{step-multi-client-check-roots}  if the following checks pass, it will proceed to the next step. It checks if every $\rt_{\st \prt}$ is a root of $\bm \theta$, by evaluating $\bm \theta$ at $\rt_{\st \prt}$ and checking if the result is $0$, i.e., $\bm\theta(\rt_{\st \prt})\stackrel{\st ?}=0$.  Otherwise, it returns $\ddot{v}=0$ and takes no further action.

\end{enumerate}

\item\label{step-multi-client-check-res}  retrieves  the linear combination of the messages $m_{\st 1},\ldots, m_{\st n}$  from polynomial $\bm\theta(x)$'s constant term: $cons = \prod\limits_{\st \forall \prtt_{\st\prt}\in  \idx}(-\rt_{\st \prt})\cdot\sum\limits_{\st \forall\prtt_{\st\prt}\in C}q_{\st \prt}\cdot m_{\st \prt}$ as follows:
\begin{equation*}
\begin{split}
res' &=cons\cdot (\prod\limits_{\st \forall\prtt_{\st\prt}\in  \idx }(-\rt_{\st \prt}))^{\st -1}\bmod \prm\\ &= \sum\limits_{\st \forall\prtt_{\st\prt}\in C}q_{\st \prt}\cdot m_{\st \prt}
\end{split}
\end{equation*}

It checks $res' \stackrel{\st ?}=m$, where $m=res$ is the result that \srv sent to it.

\item   if all the checks pass, it accepts $m$ and returns $\ddot{v}=1$. Otherwise, it returns $\ddot{v}=0$.


\end{enumerate}

\item\label{multi-client::verifying-a-solution-of-single-puzzle}\hspace{-2mm}. when verifying a solution of a single puzzle belonging to $\prtt_{\st \prt}$, i.e., when $\cmd=\scp$: 

\begin{enumerate}

\item\label{step-multi-client-check-single-puzzle-1}   checks whether opening pair $m=m_{\st \prt}$ and  $\zeta=mk_{\st \prt}$ matches the commitment: 
$$\comver\big(com_{\st \prt}, (m_{\st \prt}, mk_{\st \prt})\big)\stackrel{\st?}=1$$

where $com_{\st \prt}\in pp_{\st \prt}$.

\item\label{step-multi-client-check-single-puzzle-2}  accepts the solution $m$ and returns $\ddot{v}=1$ if the above check passes. It rejects the solution and returns $\ddot{v}=0$, otherwise. 
\end{enumerate}

\end{steps}
\end{enumerate}

\begin{theorem}\label{theo:security-of-VH-TLP}
If the sequential modular squaring assumption holds, factoring $N$ is a hard problem, \prf, $\ole^{\st +}$,  and the commitment schemes are secure, then the protocol presented above is secure. 
\end{theorem}

We refer readers to \cite{tempora-fusion} for the proof of  Theorem \ref{theo:security-of-VH-TLP}.

%% file: RSA-based-TLP.tex

\section{The Original RSA-Based TLP}\label{sec::RSA-based-TLP}

Below, we restate the original RSA-based time-lock puzzle proposed in  \cite{Rivest:1996:TPT:888615}.

\begin{enumerate}[leftmargin=.43cm]
\item \uline{Setup}: $\mathsf{Setup_{\st TLP}}(1^{\st\lambda}, \Delta, \mxsqr)$.
\begin{enumerate}

\item pick at random two large prime numbers, $q_{\st 1}$ and $q_{\st 2}$. Then, compute  $N=q_{\st 1}\cdot q_{\st 2}$. Next, compute Euler's totient function of $N$ as follows, $\phi(N)=(q_{\st 1}-1)\cdot (q_{\st 2}-1)$. 
\item set $T=\mxsqr\cdot \Delta$ the total number of squaring needed to decrypt an encrypted message $m$, where $\mxsqr$ is the maximum number of squaring modulo $N$ per second that the (strongest) solver can perform, and $\Delta$ is the period, in seconds, for which the message must remain private.

\item\label{TLP::pick-k} generate a key for the symmetric-key encryption, i.e., 
$\mathtt{SKE.keyGen}(1^{\st \lambda})\rightarrow k$.

\item choose a uniformly random value $r$, i.e., $r\stackrel{\st\$}\leftarrow\mathbb{Z}^{\st *}_{\st N}$.
\item set $a=2^{\st T}\bmod \phi(N)$.
\item set $pk:=(N,T,r)$ as the public key and $sk:=(q_{\st 1},q_{\st 2},a,k)$ as the secret key.
\end{enumerate}

\item\label{Generate-Puzzle-} \underline{Generate Puzzle}: $\mathsf{GenPuzzle_{\st TLP}}(m,pk,sk)$. 

\begin{enumerate}
\item\label{R-TLP::enc-message} encrypt the message under key $k$ using the symmetric-key encryption, as follows: $o_{\st 1}= \mathtt{SKE.Enc}(k,m)$.
\item\label{TLP::mask-k} encrypt the symmetric-key encryption key $k$, as follows: $o_{\st 2}= k+r^{\st a}\bmod N$.
\item set ${o}:=(o_{\st 1}, o_{\st 2})$ as puzzle and output the puzzle.
\end{enumerate}

\item\underline{Solve Puzzle}: $\mathsf{Solve_{\st TLP}}(pk, {o})$. 

\begin{enumerate}
\item\label{R-TLP::find-b} find $b$, where $b=r^{\st 2^{\st T}}\bmod N$, through repeated squaring of $r$ modulo $N$.
\item\label{R-TLP::dec-key} decrypt the key's ciphertext, i.e., $k=o_{\st 2}-b\bmod N$.
\item\label{R-TLP::dec-message} decrypt the message's ciphertext, i.e., $m=\mathtt{SKE.Dec}(k, o_{\st 1})$.  Output the solution, $m$.
\end{enumerate}
\end{enumerate}

The security of the RSA-based TLP relies on the hardness of the factoring problem, the security of the symmetric key encryption, and the sequential squaring assumption. We restate its formal definition below and refer readers to \cite{Abadi-C-TLP} for the proof.

\begin{theorem}\label{theorem::R-LTP-Sec}
Let $N$ be a strong RSA modulus and $\Delta$ be the period within which the solution stays private. If the sequential squaring holds, factoring $N$ is a hard problem and the symmetric-key encryption is semantically secure, then the RSA-based TLP scheme is a secure TLP.
\end{theorem}

%% file: sequential-squaring.tex

\section{Sequential and Iterated Functions}\label{sec::equential-squering}

\begin{definition} [$\Delta,\delta(\Delta))$-Sequential function]
For a function: $\delta(\Delta)$, time parameter: $\Delta$ and security parameter: $\lambda=O(\log(|X|))$,  $f:X\rightarrow Y$ is a $(\Delta,\delta(\Delta))$-sequential function if the following conditions hold:
\begin{itemize}
\item[$\bullet$] There is an algorithm that for all $x\in X$evaluates $f$ in parallel time $\Delta$, by using $poly(\log(\Delta),\lambda)$ processors.
\item[$\bullet$] For all adversaries $\mathcal{A}$ which execute in parallel time strictly less than $\delta(\Delta)$ with $poly(\Delta,\lambda)$ processors: 
$$Pr\left[y_{\st A}=f(x)\middle |  y_{\st A}\stackrel{\st \$}\leftarrow \mathcal {A}(\lambda, x), x\stackrel{\st \$}\leftarrow X\right]\leq negl(\lambda)$$
where $\delta(\Delta)=(1-\epsilon)\Delta$ and $\epsilon<1$, as stated in \cite{BonehBBF18}.
\end{itemize}
\end{definition}

\begin{definition}[Iterated Sequential function] Let $\beta: X\rightarrow X$ be a $(\Delta,\delta(\Delta))$-sequential function. A function $f: \mathbb{N}\times X\rightarrow X$ defined as $f(k,x)=\beta^{\st (k)}(x)=\overbrace{\beta\circ \beta\circ... \circ \beta}^{\st k \text{\ \ Times}}$ is  an iterated sequential function, with round function $\beta$, if for all $k=2^{\st o(\lambda)}$ the function $h:X\rightarrow X$ defined by  $h(x)=f(k,x)$ is $(k\Delta,\delta(\Delta))$-sequential. 

\end{definition}

The primary property of an iterated sequential function is that the iteration of the round function $\beta$ is the quickest way to evaluate the function. Iterated squaring in a finite group of unknown order, is widely believed to be a suitable candidate for an iterated sequential function. Below, we restate its definition.

\begin{assumption}[Iterated Squaring]\label{assumption::SequentialSquaring} Let N be a strong RSA modulus, $r$ be a generator of $\mathbb{Z}_{\st N}$, $\Delta$ be a time parameter, and $T=poly(\Delta,\lambda)$. For  any $\mathcal{A}$, defined above, there is a negligible function $\mu()$ such that: 

$$ Pr\left[
  \begin{array}{l}
\mathcal{A}(N, r,y) \rightarrow b \\
\hline
r \stackrel{\st \$}\leftarrow \mathbb{Z}_{\st N}, b\stackrel{\st \$}\leftarrow \{0,1\}\\
\text{if} \ \ b=0,\   y \stackrel{\st \$}\leftarrow \mathbb{Z}_{\st N} \\
\text {else}\ y=r^{\st 2^{\st T}}
\end{array}    \right]\leq \frac{1}{2}+\mu(\lambda)$$

\end{assumption}

